\documentclass[11pt]{article}

\usepackage[colorlinks, linkcolor=blue, citecolor=blue]{hyperref}            
\usepackage{color}
\usepackage{graphicx,subfigure,amsmath,amssymb,amsfonts,bm,epsfig,epsf,url,dsfont}
\usepackage{amsthm}
\usepackage{tikz}
\usepackage{times}
\usepackage{bbm}      
\usepackage{booktabs}
\usepackage{cases}
\usepackage{fullpage}
\usepackage[small,bf]{caption}
\usepackage{natbib}
\usepackage[top=1in,bottom=1in,left=1in,right=1in]{geometry}
\usepackage{fancybox}
\usepackage[bottom]{footmisc}

\newcommand{\bR}{\mathbb{R}}

\newcommand{\mL}{\mathcal{L}}
\newcommand{\mO}{\mathcal{O}}
\newcommand{\mG}{\mathcal{G}}
\newcommand{\mN}{\mathcal{N}}

\newcommand{\mP}{\mathcal{P}}

\newcommand{\mW}{\mathcal{W}}

\newcommand{\TV}{d_{\text{TV}}}

\newcommand{\bP}{\mathbb{P}}
\newcommand{\bE}{\mathbb{E}}

\newcommand{\N}{\mathcal{N}}

\newcommand{\pr}[1]{\textsc{#1}}

\newtheorem{question}{Question}[section]
\newtheorem{definition}{Definition}[section]

\newtheorem{fact}{Fact}[section]
\newtheorem{proposition}{Proposition}[section]
\newtheorem{theorem}{Theorem}[section]
\newtheorem{corollary}{Corollary}[section]
\newtheorem{conjecture}{Conjecture}[section]
\newtheorem{lemma}[theorem]{Lemma}

\newenvironment{fminipage}%
  {\begin{Sbox}\begin{minipage}}%
  {\end{minipage}\end{Sbox}\fbox{\TheSbox}}

\newenvironment{algbox}[0]{\vskip 0.2in
\noindent 
\begin{fminipage}{6.3in}
}{
\end{fminipage}
\vskip 0.2in
}

\begin{document}

\title{Optimal Average-Case Reductions to Sparse PCA: \\ From Weak Assumptions to Strong Hardness}

\author{Matthew Brennan\thanks{Massachusetts Institute of Technology. Department of EECS. Email: \texttt{brennanm@mit.edu}.}
\and 
Guy Bresler\thanks{Massachusetts Institute of Technology. Department of EECS. Email: \texttt{guy@mit.edu}.}}
\date{\today}

\maketitle

\begin{abstract}
In the past decade, sparse principal component analysis has emerged as an archetypal problem for illustrating statistical-computational tradeoffs. This trend has largely been driven by a line of research aiming to characterize the average-case complexity of sparse PCA through reductions from the planted clique ($\pr{pc}$) conjecture -- which conjectures that there is no polynomial-time algorithm to detect a planted clique of size $K = o(N^{1/2})$ in $\mG(N, \frac{1}{2})$ \cite{berthet2013optimal, berthet2013complexity, wang2016statistical, gao2017sparse, brennan2018reducibility}. All of these reductions either fail to show tight computational lower bounds matching existing algorithms or show lower bounds for formulations of sparse PCA other than its canonical generative model, the spiked covariance model. Also, these lower bounds all quickly degrade with the exponent in the $\pr{pc}$ conjecture. Specifically, when only given the $\pr{pc}$ conjecture up to $K = o(N^\alpha)$ where $\alpha < 1/2$, there is no sparsity level $k$ at which these lower bounds remain tight. If $\alpha \le 1/3$ these reductions fail to even show the existence of a statistical-computational tradeoff at any sparsity $k$. Our main results are as follows:
\begin{itemize}
\item We give a reduction from $\pr{pc}$ that yields the first full characterization of the computational barrier in the spiked covariance model, providing tight lower bounds at all sparsities $k$. This partially resolves a question raised in \cite{brennan2018reducibility}.
\item We show the surprising result that weaker forms of the $\pr{pc}$ conjecture up to clique size $K = o(N^\alpha)$ for any given $\alpha \in (0, 1/2]$ imply tight computational lower bounds for sparse PCA at sparsities $k = o(n^{\alpha/3})$. This shows that even a mild improvement in the signal strength needed by the best known polynomial-time sparse PCA algorithms would imply that the hardness threshold for $\pr{pc}$ is \emph{subpolynomial}, rather than on the order $N^{1/2}$ as is widely conjectured.
\end{itemize}
Our second result essentially shows that whether or not there are better efficient algorithms for $\pr{pc}$ is irrelevant to the statistical-computational gap for sparse PCA in the practically relevant highly sparse regime. This is the first instance of a suboptimal hardness assumption implying optimal lower bounds for another problem in unsupervised learning.

The reduction proving this result is more algorithmically involved than prior reductions to sparse PCA, making crucial use of a collection of average-case reduction primitives and introducing new techniques based on several decomposition and comparison properties of random matrices. Our lower bounds remain unchanged assuming hardness of planted dense subgraph instead of $\pr{pc}$, which also has implications for the existence of algorithms that are slower than polynomial time. As a key intermediate step, our reduction maps an instance of $\pr{pc}$ to the empirical covariance matrix of sparse PCA samples, which proves to be a delicate task because of dependence among the entries of this matrix.
\end{abstract}

\pagebreak

\tableofcontents

\pagebreak

\section{Introduction}

Principal component analysis (PCA), the task of projecting multivariate samples onto the leading eigenvectors of their empirical covariance matrix, is one of the most popular dimension reduction techniques in statistics. However in modern high-dimensional settings, PCA no longer provides a meaningful estimate of principal components. More precisely, the BBP transition in random matrix theory implies that PCA yields a statistically inconsistent estimator when the data is distributed according to the spiked covariance model \cite{baik2005phase, paul2007asymptotics, johnstone2009consistency}. Sparse PCA was introduced in \cite{johnstoneSparse04} to alleviate this inconsistency in the high-dimensional setting and has found applications in a diverse range of fields. Examples include online visual tracking \cite{wang2013online}, pattern recognition \cite{naikal2011informative}, image compression \cite{majumdar2009image}, electrocardiography \cite{johnstone2009consistency}, gene expression analysis \cite{zou2006sparse, chun2009expression, parkhomenko2009sparse, chan2010using}, RNA sequence classification \cite{tan2014classification} and metabolomics studies \cite{allen2011sparse}. Further background on sparse PCA can be found in \cite{wang2016statistical}.

The canonical formulation of sparse PCA is the spiked covariance model, which has observed data $X = (X_1, X_2, \dots, X_n)$ distributed i.i.d. according to $\mN(0, I_d + \theta vv^\top)$ where $v\in \bR^d$ is a $k$-sparse unit vector and $\theta$ parametrizes the signal strength. The objective is either to estimate the direction $v$ of the spike or detect its existence. Detection is formulated as a hypothesis testing problem on the data $X = (X_1, X_2, \dots, X_n)$ with hypotheses
$$H_0:X \sim \mN(0, I_d)^{\otimes n} \quad\text{ and }\quad  H_1:X \sim \mN(0, I_d + \theta vv^\top)^{\otimes n}\,.$$
There is an extensive literature on both the statistical limits and efficient algorithms for sparse PCA. A large number of algorithms solving sparse PCA have been proposed \cite{amini2009high, ma2013sparse, cai2013sparse, berthet2013optimal, berthet2013complexity, shen2013consistency, krauthgamer2015semidefinite, deshpande2014sparse, wang2016statistical}. As shown in \cite{berthet2013optimal, cai2015optimal, wang2016statistical}, the statistical limit of detection in the spiked covariance model is $\theta = \Theta(\sqrt{k \log d/n})$ and the minimax rate of estimation under the $\ell_2$ norm is $O(\sqrt{k \log d / n\theta^2})$. These rates are achieved by some of the estimators in the papers listed above, none of which can be computed in polynomial time. In contrast, in the sparse regime of $k = \tilde{O}(\sqrt{n})$, the best known polynomial time algorithms require the much larger signal $\theta = \Omega(\sqrt{k^2/n})$ for detection and achieve a minimax rate of estimation under the $\ell_2$ norm of $O(\sqrt{k^2 / n\theta^2})$. This phenomenon raises the general question: is this statistical-computational gap inherent or are there better polynomial time algorithms?

Since the seminal paper of Berthet and Rigollet \cite{berthet2013complexity}, sparse PCA has emerged as the archetypal example of a high-dimensional statistical problem with a statistical-computational gap and serves as an ideal case-study of the general phenomenon. Consequently, a line of research has aimed to characterize the average-case complexity of sparse PCA. This and prior work seeks to determine the computational phase diagram of sparse PCA by answering the following question.

\begin{question}
For a given scaling of parameters $(n, k, d, \theta)$, is sparse PCA information-theoretically impossible, efficiently solvable or in principle feasible but computationally hard?
\end{question}

There are several feasible approaches to providing an answer to this question. One taken in the literature is to show lower bounds for variants of sparse PCA in classes of algorithms such as the sum of squares hierarchy \cite{ma2015sum, hopkins2017power} and statistical query algorithms \cite{diakonikolas2016statistical, lu2018edge}. A more traditional complexity-theoretic approach is to show lower bounds against all polynomial time algorithms by reducing from a conjecturally hard problem. Since the seminal paper of Berthet and Rigollet, a growing line of research \cite{berthet2013optimal, berthet2013complexity, wang2016statistical, gao2017sparse, brennan2018reducibility} has shown computational lower bounds for sparse PCA through reductions from the planted clique ($\pr{pc}$) conjecture \cite{berthet2013optimal, berthet2013complexity, wang2016statistical, gao2017sparse, brennan2018reducibility}.

The $\pr{pc}$ problem is to detect or find a planted clique of size $K$ in the $N$-vertex Erd\H{o}s-R\'{e}nyi graph $\mG(N, \frac{1}{2})$. Because the largest clique in $\mG(N,\frac12)$ is with high probability of size roughly $2\log_2 N$, planted cliques of size $K \ge (2+\epsilon)\log_2 N$ can be detected in $N^{O(\log N)}$ time by searching over all $(2+\epsilon)\log_2 N$-node subsets. The $\pr{pc}$ conjecture is that there are no polynomial-time algorithms for this task if $K = o(N^{1/2})$, which we state formally in Section \ref{sec:problems}. Since its introduction in \cite{kuvcera1995expected} and \cite{jerrum1992large}, $\pr{pc}$ has been studied extensively. Spectral algorithms, approximate message passing, semidefinite programming, nuclear norm minimization and several other polynomial-time combinatorial approaches all appear to fail to recover the planted clique when $K = o(N^{1/2})$ \cite{alon1998finding, feige2000finding, mcsherry2001spectral, feige2010finding, ames2011nuclear, dekel2014finding, deshpande2015finding, chen2016statistical}. In support of the $\pr{pc}$ conjecture, it has been shown that cliques of size $K = o(N^{1/2})$ cannot be detected by the Metropolis process \cite{jerrum1992large}, low-degree sum of squares relaxations \cite{barak2016nearly} and statistical query algorithms \cite{feldman2013statistical}. Recently, \cite{atserias2018clique} showed that super-polynomial length regular resolution is required to certify that Erd\H{o}s-R\'{e}nyi graphs do not contain cliques of size $K = o(N^{1/4})$. In the other direction, \cite{frieze2008new, brubaker2009random} provide an algorithm for finding cliques of size $K = O(N^{1/r})$ given the ability to compute the 2-norm of a certain random $r$-parity tensor\footnote{Computing this 2-norm was shown to be hard in the worst case in \cite{hillar2013most}, although its average-case complexity remains unknown.}. Because $\pr{pc}$ is easier for larger $K$, the assumption that $\pr{pc}$ is hard for clique sizes $K=o(N^\alpha)$ is weaker for smaller $\alpha$. 

Throughout this paper, we say that a computational lower bound is \emph{tight} or \emph{optimal} if it matches what is achievable by the best known polynomial-time algorithms. All of the previous reductions from $\pr{pc}$ to sparse PCA in \cite{berthet2013optimal, berthet2013complexity, wang2016statistical, gao2017sparse, brennan2018reducibility} either fail to show tight computational lower bounds matching existing algorithms or show lower bounds for formulations of sparse PCA other than the spiked covariance model. In particular, a full characterization of the computational feasibility in the spiked covariance model has remained open. These lower bounds also all quickly degrade with the exponent in the $\pr{pc}$ conjecture. Specificallly, when only given the $\pr{pc}$ conjecture up to $K = o(N^\alpha)$ where $\alpha < 1/2$, there is no sparsity level $k$ at which these lower bounds remain tight. If $\alpha \le 1/3$ these reductions fail to even show the existence of a statistical-computational tradeoff at any sparsity $k$. This phenomenon is unsurprising -- as the assumed lower bound for $\pr{pc}$ is loosened, it is intuitive that the resulting lower bound for sparse PCA always weakens as well.

The main results of this work are to resolve the the feasibility diagram for a formulation of sparse PCA in the spiked covariance model and give an unintuitively strong reduction from $\pr{pc}$. More precisely, our results are as follows:
\begin{itemize}
\item We show the surprising result that weaker forms of the $\pr{pc}$ conjecture up to clique size $K = o(N^\alpha)$ for any given $\alpha \in (0, 1/2]$ imply tight computational lower bounds for sparse PCA at sparsities $k = o(n^{\alpha/3})$.
\item We give a reduction from $\pr{pc}$ that yields the first full characterization of the computational barrier in the spiked covariance model, providing tight lower bounds at all sparsities $k$. This completes the feasibility picture for the spiked covariance model depicted in Figure~\ref{fig:spcaphasediagram} and partially resolves a question raised in \cite{brennan2018reducibility}.
\end{itemize}
This first result has several strong implications for the relationship between the computational complexity of sparse PCA and $\pr{pc}$. Even a mild improvement in the signal strength needed by the best known polynomial-time sparse PCA algorithms, to below $\theta = \tilde{\Theta}(\sqrt{k^2/n})$, would imply that the hardness threshold for $\pr{pc}$ is \emph{subpolynomial}, rather than on the order $N^{1/2}$ as is widely conjectured. Whether or not the hardness threshold for planted clique is in fact at $N^{1/2}$ is irrelevant to the statistical-computational gap for sparse PCA in the practically relevant highly sparse regime. This result is also the first instance of a suboptimal hardness assumption implying optimal lower bounds for another problem in unsupervised learning. Before stating our results in Section~\ref{sec:results}, we review prior reductions to sparse PCA, state the resulting bounds and discuss why there is a fundamental limitation to the pre-existing techniques. We also give a high-level overview of how our reductions overcome these barriers.

\subsection{Prior Reductions and Overcoming the Limits of Reducing Directly to Samples}

As in most previous papers on sparse PCA, we primarily focus on the more sparse regime in which $k, n$ and $d$ satisfy the conditions $k = \tilde{O}(\sqrt{n})$ and $k = \tilde{o}(\sqrt{d})$. The following is a brief overview of previous reductions to sparse PCA in the literature.
\begin{itemize}
\item \textbf{Berthet-Rigollet (2013b)}: \cite{berthet2013optimal} gives a reduction from planted clique to show hardness for solving sparse PCA with semidefinite programs that can be computed in polynomial time.
\item \textbf{Berthet-Rigollet (2013a)}: \cite{berthet2013complexity} gives a reduction from planted clique to a composite vs. composite hypothesis testing formulation of sparse PCA, showing lower bounds for polynomial time algorithms that solve sparse PCA uniformly over all noise distributions satisfying a $d$-dimensional sub-Gaussian concentration condition. In this formulation, \cite{berthet2013complexity} shows tight computational lower bounds up to the threshold of $\theta = \tilde{o}(\sqrt{k^2/n})$ when $k = \tilde{o}(\sqrt{d})$ and $k = \tilde{O}(\sqrt{n})$, matching the best known polynomial time algorithms at these sparsities.
\item \textbf{Wang-Berthet-Samworth (2016)}: \cite{wang2016statistical} shows computational lower bounds for the estimation task in a similar distributional-robust sub-Gaussian variant of sparse PCA as in \cite{berthet2013complexity}, again up to the conjectured threshold of $\theta = \tilde{o}(\sqrt{k^2/n})$.
\item \textbf{Gao-Ma-Zhou (2017)}: \cite{gao2017sparse} shows the first computational lower bounds for sparse PCA in its canonical generative model, the spiked covariance model, in which the noise distribution is exactly a multivariate Gaussian. The reduction is to a simple vs. composite hypothesis testing formulation of sparse PCA and shows lower bounds up to the suboptimal threshold of $\theta = \tilde{o}(\sqrt{k^2/n})$ when $k = \tilde{o}(\sqrt{d})$ and $k = \tilde{O}(\sqrt{n})$. This lower bound is only tight when $\theta = \tilde{\Theta}(1)$ and $k = \tilde{\Theta}(\sqrt{n})$.
\item \textbf{Brennan-Bresler-Huleihel (2018)}: \cite{brennan2018reducibility} provides an alternative reduction based on random rotations to strengthen the lower bounds from \cite{gao2017sparse} to hold for a simple vs. simple hypothesis testing formulation of the spiked covariance model up to the suboptimal threshold of $\theta = \tilde{o}(\sqrt{k^2/n})$. It also shows the first tight lower bounds in the high sparsity regime $k = \omega(\sqrt{n})$ up to the threshold of $\theta = \tilde{o}(1)$ when $d = O(n)$, which are matched by polynomial-time algorithms.
\end{itemize}
Assuming the hardness of planted clique with $N$ vertices and clique size $K$, \cite{berthet2013optimal, berthet2013complexity} show lower bounds for sparse PCA with sparsities $k \le K$, samples $n= \Theta(N)$ and dimension $d = \Theta(N)$, and $\theta$ satisfying (1) below. Given this same planted clique hardness assumption, \cite{gao2017sparse} and \cite{brennan2018reducibility} show lower bounds for sparse PCA with sparsity $k = K$,  samples $n= \Theta(N)$ and dimension $d = \Theta(N)$, and $\theta$ satisfying (2), where
$$\text{(1)} \qquad \theta = \tilde{o}\left( \frac{k \cdot K}{n} \right) \quad \quad \qquad \text{(2)} \qquad \theta = \tilde{o}\left( \frac{k^2}{n} \right)\,.$$
The bound (1), upon plugging in $K = o(N^{1/2})$, implies hardness for any $\theta=\tilde o(\sqrt{k^2/n})$. When only given the planted clique conjecture up to $K = o(N^\alpha)$ for an $\alpha < 1/2$, there is no sparsity level $k$ at which either of these lower bounds remain tight to the conjectured boundary of $\theta = \tilde{o}(\sqrt{k^2/n})$. If $\alpha \le 1/3$, then $k \cdot K/n = o(\sqrt{k/n})$ and these reductions fail to show computational lower bounds beyond the statistical limit of $\theta = \Theta(\sqrt{k \log d/n})$. A more detailed discussion of previous reductions from planted clique to sparse PCA is in Appendix \ref{sec:appendixprevred}. 

All of these prior reductions convert the rows of the $\pr{pc}$ adjacency matrix directly to samples from sparse PCA. This turns out to be lossy for a simple reason: while the planted clique itself is sparse both in its column and row support, sparse PCA data is independent across samples and thus structured only along one axis when arranged into an $n\times d$ matrix. Thus, directly converting a $\pr{pc}$ instance to samples from sparse PCA necessarily eliminates the structure across the rows and destroys some of the underlying structure.

The intuition behind our main reduction is that an instance of $\pr{pc}$ naturally corresponds to a shifted and rescaled \emph{empirical covariance matrix} of samples from the spiked covariance model, rather than the samples themselves. This empirical covariance matrix, although it has dependent entries and is distributionally very different from planted clique, has sparse structure along both of its axes. We elaborate on this intuition in Appendix \ref{sec:appendixprevred}. While previous reductions have all produced samples directly, we map an instance of $\pr{pc}$ to this empirical covariance matrix and then map to samples using an isotropic property of the Wishart distribution. This strategy yields an optimal relationship between $(n, k, d, \theta)$ that no longer degrades with $\alpha$ in the $K = o(N^{\alpha})$ assumption in the $\pr{pc}$ conjecture. Executing this strategy proves to be a delicate task because of dependence among the entries of the empirical covariance matrix, requiring a number of new average-case reduction primitives. Specifically, to overcome this dependence, we require several decomposition and comparison results for random matrices proven in Section \ref{sec:Wishart} and crucially use a recent result in \cite{bubeck2016testing} comparing $\pr{goe}$ to Wishart matrices. Our techniques are outlined in more detail in Section \ref{sec:techniques}.

\subsection{Summary of Results and Structure of the Paper}
\label{sec:results}

\paragraph{Strong Hardness from Weak Assumptions.} We now state our main result, which essentially shows that whether or not there are better polynomial time algorithms for $\pr{pc}$ is irrelevant to the statistical-computational gap for sparse PCA in the highly sparse regime of low $k$, which consists of the ``interpretable'' levels of sparsity that often arise in applications. The following is an informal statement of our main theorem, which is formally stated in Corollary \ref{thm:lowerbounds}. Instead of only showing lower bounds from $\pr{pc}$, we reduce from the more general the \emph{planted dense subgraph} detection problem which entails testing for the presence of a planted $\mG(K, p)$ in a $\mG(N, q)$ where $0 < q < p \le 1$. The problem $\pr{pc}$ is recovered by setting $p = 1$ and $q = 1/2$.

\begin{theorem}[Informal Main Theorem]
Fix some $\alpha \in (0, 1/2]$ and $0 < q < p \le 1$. If there is no randomized polynomial-time algorithm solving the planted dense subgraph detection problem with densities $p$ and $q$ on graphs of size $N$ with any dense subgraph size $K = o(N^\alpha)$, then there is no randomized polynomial time algorithm solving detection in the spiked covariance model for all $(n, k, d, \theta)$ with $\theta = \tilde{O}(\sqrt{k^2/n})$, $k = O(d^{\alpha})$ and $k = O(n^{\alpha/3})$.
\end{theorem}

Formulating our theorems in terms of planted dense subgraph means that they subsume statements based on $\pr{pc}$ hardness and have additional consequences such as for quasipolynomial time algorithms for sparse PCA, which are discussed in Section \ref{sec:conclusion}.
%For example, it is conjectured that there are no quasipolynomial time algorithms for planted dense subgraph if $p < 1$. Assuming this stronger conjecture of hardness, implies by our reduction that there are no quasipolynomial time algorithms for sparse PCA\footnote{There is no known quasipolynomial time algorithm for planted dense subgraph, and the problem is believed to be strictly harder than planted clique. For instance the sum of squares degree needed to solve it suggests that exponential time is needed \cite{hopkins2017power}.}.
The reduction proving our main theorem is sketched in Section~\ref{sec:techniques} and formally described in Section~\ref{sec:wishartreduction}. In Section \ref{sec:problems}, we formally introduce the models of sparse PCA we consider and the known statistical limits and polynomial-time algorithms for these models. In Sections \ref{sec:avgcasereductions} and \ref{sec:mappingsubmatrix}, we introduce average-case reductions in total variation and a number of crucial subroutines that we will use in our reductions. In Sections \ref{sec:Wishart} and \ref{sec:proofmain}, we prove our main theorem.

\begin{figure*}[t!]
\centering
\begin{tikzpicture}[scale=0.4]
\tikzstyle{every node}=[font=\footnotesize]
\def\xmin{0}
\def\xmax{16}
\def\ymin{-1}
\def\ymax{11}

\draw[->] (\xmin,\ymin) -- (\xmax,\ymin) node[right] {$\beta$};
\draw[->] (\xmin,\ymin) -- (\xmin,\ymax) node[above] {$\alpha$};

\node at (15, -1) [below] {$\frac{1}{8}$};
\node at (13.33, -1) [below] {$\frac{1}{9}$};
\node at (0, 0) [left] {$\frac{3}{8}$};
\node at (0, 10) [left] {$\frac{5}{8}$};
\node at (0, 5) [left] {$\frac{1}{2}$};

\filldraw[fill=black, draw=black] (13.33, 0.56) -- (13.33, 2.78) -- (15, 2.5) -- (15, 0) -- (13.33, 0.56);
\filldraw[fill=cyan, draw=blue] (0, 5) -- (13.33, 0.56) -- (13.33, 2.78) -- (0, 5);
\filldraw[fill=gray!25, draw=gray] (0, 5) -- (15, 0) -- (15, -1) -- (0, -1) -- (0, 5);
\filldraw[fill=green!25, draw=green] (0, 5) -- (15, 2.5) -- (15, 10) -- (0, 10) -- (0, 5);
\draw[->] (10, 1.67) -- (10, 3.33);
\draw[->] (10, 3.33) -- (10, 1.67);

\node at (5.6, 3.3) [rotate=-21, anchor=north]{$\theta \asymp \sqrt{\frac{k^2}{n}}$};
\node at (11, 2.3) {tight};
\node at (7.5, 9) {information-theoretically impossible};
\node at (10, 5) {PC-hard given};
\node at (2.5, 6) {interpretable};
\node at (2.5, 5.25) {sparsities};
\node at (10, 4.25) {$K = o(N^{1/3})$};
\node at (7.5, -0.25) {polynomial-time algorithms};
\end{tikzpicture}
\begin{tikzpicture}[scale=0.4]
\tikzstyle{every node}=[font=\footnotesize]
\def\xmin{0}
\def\xmax{16}
\def\ymin{-1}
\def\ymax{11}

\draw[->] (\xmin,\ymin) -- (\xmax,\ymin) node[right] {$\beta$};
\draw[->] (\xmin,\ymin) -- (\xmin,\ymax) node[above] {$\alpha$};

\node at (15, -1) [below] {$1$};
\node at (7.5, -1) [below] {$\frac{1}{2}$};
\node at (0, 0) [left] {$0$};
\node at (0, 10) [left] {$1$};
\node at (0, 5) [left] {$\frac{1}{2}$};

\filldraw[fill=cyan, draw=blue] (0, 5) -- (15, 0) -- (7.5, 0) -- (0, 5);
\filldraw[fill=gray!25, draw=gray] (0, 5) -- (7.5, 0) -- (15, 0) -- (15, -1) -- (0, -1) -- (0, 5);
\filldraw[fill=green!25, draw=green] (0, 5) -- (15, 0) -- (15, 10) -- (0, 10) -- (0, 5);

\node at (1.7, 1.53)[rotate=-37, anchor=south] {$\theta \asymp \sqrt{\frac{k^2}{n}}$};
\node at (7.5, 2.5)[rotate=-20, anchor=south]  {$\theta \asymp \sqrt{\frac{k \log d}{n}}$};
\node at (11.25, -0.5) {$\theta = \tilde{\Theta}(1)$};
\node at (7.5, 9) {information-theoretically impossible};
\node at (3, 0.5) {polynomial-time};
\node at (3, -0.25) {algorithms};
\node at (8, 1.25) {PC-hard};
\end{tikzpicture}

\caption{\small Computational phase diagrams showing efficiently solvable, information-theoretically impossible, and computationally hard parameter regimes for sparse PCA with $k = \tilde{\Theta}(n^\beta)$ and $\theta = \tilde{\Theta}(n^{-\alpha})$. The left shows an example of the regime of tight lower bounds from the main theorem assuming the $\pr{pc}$ conjecture up to $K = o(N^{1/3})$. The right shows the tight lower bounds from the second main theorem.}
\label{fig:spcaphasediagram}
\end{figure*}
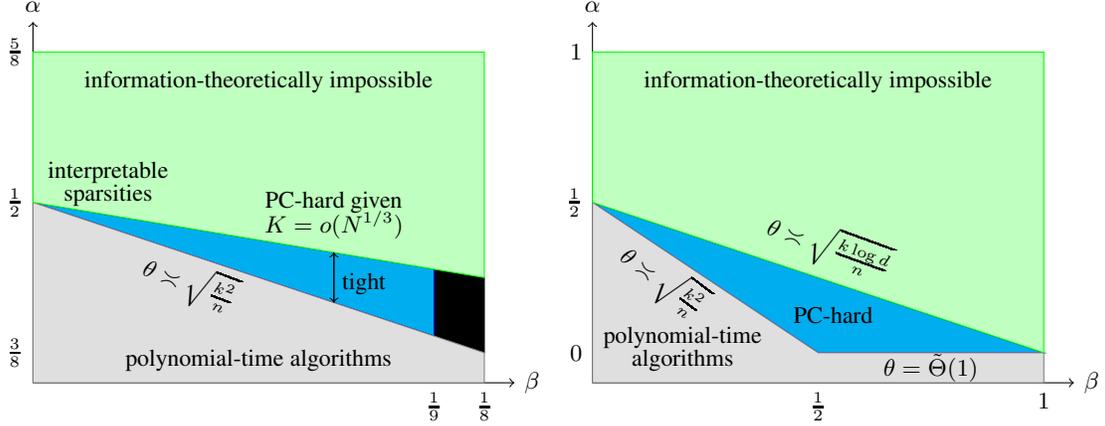

\paragraph{Completing the Computational Phase Diagram for the Spiked Covariance Model.} In Section \ref{sec:randomrotations}, we give a reduction to the spiked covariance model combining a variant of the random rotations reduction of \cite{brennan2018reducibility} with an internal subsampling reduction to yield tight lower bounds for all sparsities $k = o(\sqrt{n})$ in a variant of the spiked covariance model based on the planted clique conjecture for cliques of size $K = o(N^{1/2})$. Together with the lower bounds when $k = \omega(\sqrt{n})$ in Theorem 8.5 of \cite{brennan2018reducibility}, this yields the first complete set of computational lower bounds in a simple vs. composite formulation of the spiked covariance model of sparse PCA that are tight at all sparsities $k$, partially resolving a question in \cite{brennan2018reducibility}, which posed this problem for the simple vs. simple formulation. The $\textsc{fcspca}$ simple vs. composite formulation of the spiked covariance model for which we characterize feasibility, as well as some other formulations of varying flexibility, are described in Section ~\ref{sec:problems}. The $\textsc{fcspca}$ formulation allows a slight variation in the size of the support of $v$ and in the signal strength parameter, and is quite close to the original spiked covariance model. 

The following is an informal statement of our second main theorem, which combines Corollary \ref{thm:chirandomrot} with Theorem 8.5 of \cite{brennan2018reducibility}. The associated diagram is given in right of Figure~\ref{fig:spcaphasediagram}.

\begin{theorem}[Informal Second Main Theorem]
Fix some $0 < q < p \le 1$. If there is no randomized polynomial time algorithm solving the planted dense subgraph detection problem with densities $p$ and $q$ on graphs of size $N$ with dense subgraph size $K = o(\sqrt{N})$, then there is no randomized polynomial time algorithm solving detection in the $\textsc{fcspca}$ simple vs. composite formulation of the spiked covariance model for all $(n, k, d, \theta)$ satisfying either:
\begin{itemize}
\item $\theta = \tilde{O}(\sqrt{k^2/n})$ for sparsities $k = o(\sqrt{d})$ and $k = o(\sqrt{n})$
\item $\theta = \tilde{o}(1)$ for sparsities $k = \tilde{\omega}(\sqrt{n})$ and $d = O(n)$
\end{itemize}
\end{theorem}

In Section \ref{sec:internal}, we also show that a simple internal cloning reduction within sparse PCA shows that hardness at $(n, k, d, \theta)$ implies hardness at $(n, tk, td, \theta)$ for any $t \ge 1$. This shows that the tight lower bounds derived from the planted clique conjecture up to $K = o(N^\alpha)$ in our main theorem can also be extended to lower bounds with $k = \omega(n^{\alpha/3})$ that, although not tight, still show nontrivial statistical-computational gaps. In Section \ref{sec:conclusion}, we discuss simple extensions and implications of our results. We also state several problems left open in this work.

\subsection{Overview of Techniques}
\label{sec:techniques}

To show our main result, it suffices to give a reduction $\mathcal{A}$ mapping an instance of $\pr{pc}$ with $K = o(N^\alpha)$ to an instance of the spiked covariance model with parameters $(n, k, d, \theta)$ satisfying $\theta = \tilde{\Theta}(\sqrt{k^2/n})$ under both $H_0$ and $H_1$. Concretely, the desired $\mathcal{A}$ should simultaneously map $G \sim \mG(N, 1/2)$ to $X = (X_1, X_2, \dots, X_n) \sim \mN(0, I_d)^{\otimes n}$ and $G \sim \mG(N, K, 1/2)$ to $X \sim \mN(0, I_d + \theta vv^\top)^{\otimes n}$, approximately under total variation, where $v \in \mathbb{R}^d$ is a $k$-sparse unit vector with nonzero entries $1/\sqrt{k}$. As mentioned above, our insight is that the optimal dependence $\theta = \tilde{\Theta}(\sqrt{k^2/n})$ can arise from first mapping $G$ approximately to the scaled empirical covariance matrix\footnote{For technical reasons we actually map to $\Sigma_e = \sum_{i = 1}^n Y_i Y_i^\top$ where the $Y_i$ are restrictions of the $X_i$ to $m$ coordinates containing all of the sparse principal component.} $\Sigma_e = \sum_{i = 1}^n X_i X_i^\top$, and then mapping $\Sigma_e$ to $X$. Note that if we represent edge versus no edge in planted clique by $\pm 1$, then the \emph{expectations} of the nondiagonal entries of the planted clique adjacency matrix and the empirical covariance matrix closely resemble are equal up to a rescaling under both $H_0$ and $H_1$. As described next, one of the main challenges is capturing the dependencies among the random entries. 

\paragraph{Empirical covariance matrix.}  Under $H_0$, it holds that $\Sigma_e$ is distributed as a $d \times d$ isotropic Wishart matrix with $n$ degrees of freedom. A recent result in random matrix theory of \cite{bubeck2016testing} implies that if $n/d^3 \to \infty$, then $\Sigma_e$ converges in total variation to $n \cdot I_d + \sqrt{n} \cdot \pr{goe}(d)$, where $\pr{goe}(d)$ denotes the distribution of $\frac{1}{\sqrt{2}} (A + A^\top)$ and $A \sim \mN(0, 1)^{\otimes d \times d}$. We see that the empirical covariance matrix has approximately independent entries and is quite simple in this case. 

The distribution of $\Sigma_e$ is much more complicated under $H_1$, exhibiting dependencies among the entries whose structure depends on the location of the planted spike. A useful characterization of the distribution seems hopeless in general, but a simplification occurs in the regime $n/d^3 \to \infty$, whereby the entries become jointly Gaussian, although still dependent. In Section~\ref{sec:Wishart}, we show that the distribution of $\Sigma_e$ under $H_1$ converges in total variation to $n(I_d + \theta vv^\top) + \frac{1}{\sqrt{2}} (B + B^\top)$ where
$$B = \sqrt{n} \cdot \left( \mN(0, 1)^{\otimes d \times d} + \sqrt{\theta} \cdot v_S w_1^\top + \sqrt{\theta} \cdot w_2 v_S^\top + \theta g v_Sv_S^\top \right)$$
and $w_1, w_2 \sim \mN(0, I_d)$ and $g \sim \mN(0, 1)$ are independent. The independence between these terms implies that $\Sigma_e$ is close in total variation to the following independent sum involving three instances of simpler average-case problems:
$$n \cdot I_m + \sqrt{\frac{n}{6}} \left( M + M^\top + C_L + C_L^\top + C_R + C_R^\top \right)$$
where $M, C_L$ and $C_R$ are such that
\begin{align*}
H_0 : M \sim \mN(0, 1)^{\otimes d \times d} \quad &H_1: M \sim \theta \sqrt{3} \cdot \left(\sqrt{n/2} + g \right) vv^\top + \mN(0, 1)^{\otimes d \times d} \\
H_0 : C_L \sim \mN(0, 1)^{\otimes d \times d} \quad &H_1: C_L \sim \sqrt{3\theta} \cdot w_1 v^\top + \mN(0, 1)^{\otimes d \times d} \\
H_0 : C_R \sim \mN(0, 1)^{\otimes d \times d} \quad &H_1: C_R \sim \sqrt{3\theta} \cdot v w_2^\top + \mN(0, 1)^{\otimes d \times d}\,.
\end{align*}
Observe that $M$ is a variant of Gaussian biclustering with random signal strength parameter and $C_L$ and $C_R$ are weak instances of the spiked covariance model. Crucially, all of these matrices are expressed in terms of the spike direction $v$ in a way that makes it possible to map to them from an instance of planted clique without knowing the clique location. The reduction proceeds by cloning the planted clique instance to obtain multiple independent copies with slightly modified parameters, and then mapping to each of $M, C_L$ and $C_R$. To execute this step, we require a number of subroutines to overcome several distributional subtleties, as described next.

\paragraph{Missing diagonal entries.} The covariance matrix $I_d + \theta v v^\top$ describing the distribution of the spiked covariance model under $H_1$ has nontrivial signal on its diagonals, with $i$th diagonal entry $1+ \theta v_i^2$. However, the diagonal of the $\pr{pc}$ adjacency matrix is uniformly zero and thus does not contain any signal distinguishing clique vertices versus non-clique vertices. Without planting signal in the diagonal one cannot map to the desired distribution in total variation, and thus we are faced with the task of placing $1$'s in the correct entries without knowing the clique location. To produce the missing diagonal entries in the adjacency matrix of $G$, we embed this adjacency matrix as a principal minor in a larger matrix using a procedure from \cite{brennan2019universality}. %This essentially places $1$'s on the \emph{entire} diagonal of $G$, so we have fixed the clique vertices but also introduced far too many ones. This is fixed by embedding the matrix within a larger one having \emph{too few} ones on the diagonal, so that the total number has the correct distribution. 
 
\paragraph{Further distributional maps.} A number of further distributional maps are described in detail in Section~\ref{sec:mappingsubmatrix}. To remove the symmetry of the $\pr{pc}$ adjacency matrix and produce multiple independent copies, we provide two cloning procedures. To Gaussianize the entries of these matrices, we use the rejection kernel framework introduced in \cite{brennan2018reducibility}. We additionally introduce a reduction, termed $\chi^2$ random rotations, to produce the weak instances $C_L$ and $C_R$ of the spiked covariance model appearing above in the decomposition of the empirical covariance matrix $\Sigma_e$. Because these have significantly reduced signal, we can use a lossy procedure. Now the reduction has arrived at an object close in total variation to $\Sigma_e$ and the task is to retrieve the samples $X$ from $\Sigma_e$. 
 
 \paragraph{Inverse Wishart sampling.} In Section \ref{sec:Wishart}, it is shown that a consequence of the isotropy of independent multivariate Gaussians is that if $\Sigma_e' \in \mathbb{R}^{d \times d}$ is distributed as the scaled empirical covariance matrix of $\mN(0, \Sigma)^{\otimes n}$, then $(\Sigma_e')^{1/2} R \sim \mN(0, \Sigma)^{\otimes n}$, where $R$ consists of the top $d$ rows of a matrix sampled from the Haar measure on the orthogonal group $\mathcal{O}_n$ and $n \ge d$. Here $(\Sigma_e')^{1/2}$ is the positive semidefinite square-root of $\Sigma_e'$. What is noteworthy is that the samples produced have distribution described by the true covariance $\Sigma$, even though we started with the \emph{empirical} covariance matrix $\Sigma_e$. Therefore, generating $R$ and computing $\Sigma_e^{1/2} R$ completes the reduction, after applying a further post-processing step of appropriately padding the data matrix. 
 
% In the appendix, we carry out this reduction starting from a general instance of planted dense subgraph.

\paragraph{Subsampling and cloning.} Finally in Section \ref{sec:randomrotations}, we give subsampling and cloning internal reductions within sparse PCA which, when combined with our $\chi^2$ random rotation reduction, imply the lower bounds in our second main theorem. By \emph{internal reduction} we mean that these procedures take as input and then output instances of sparse PCA, but with different parameters. This translates hardness from a given point in the feasibility diagram to other points and we apply it to a point where the $\chi^2$ random rotation reduction gives tight hardness, $k=\tilde \Theta(\sqrt{n})$ and $\theta = \tilde \Theta(1)$, to deduce tight hardness for smaller $k$ and $\theta = \tilde o(\sqrt{k^2/n})$. The subsampling internal reduction is based on the observation that projecting down to a smaller coordinate subspace of data from the spiked covariance model results in another instance of the spiked covariance model. The formal statement follows by carefully accounting for the proportion of the support of the spike that is retained.

\subsection{Related Work on Statistical-Computational Gaps}

This work is part of a growing body of literature giving rigorous evidence for computational-statistical gaps in high-dimensional inference problems. A more detailed survey of this area can be found in the introduction section of \cite{brennan2018reducibility}.

\paragraph{Computational Lower Bounds for Sparse PCA.} In addition to average-case reductions, the average-case complexity of sparse PCA has been examined from the perspective of several restricted models of computation including the statistical query model and sum of squares semidefinite (SOS) programming hierarchy. Statistical query lower bounds for sparse PCA and related problems were shown in \cite{lu2018edge} and \cite{diakonikolas2016statistical}. Degree four SOS lower bounds for sparse PCA in the spiked covariance model were shown in \cite{ma2015sum}. SOS lower bounds were also shown for the spiked Wigner model, a closely related model, in \cite{hopkins2017power}. In \cite{brennan2018reducibility}, average-case reductions from the planted clique conjecture were used to fully characterize the computational phase diagram of the spiked Wigner model. We remark that the entry-wise independence in the spiked Wigner model makes it an easier object to map to with average-case reductions. The worst-case complexity of sparse PCA as an approximation problem has also been considered in \cite{chan2016approximability}.

\paragraph{Average-Case Reductions.} One reason behind the recent focus on showing hardness in restricted models of computation is that average-case reductions are inherently delicate, creating obstacles to obtaining satisfying hardness results. As described in \cite{Barak2017}, these technical obstacles have left us with an unsatisfying theory of average-case hardness. Unlike reductions in worst-case complexity, average-case reductions between natural decision problems need to precisely map the distributions on instances to one another without destroying the underlying signal in polynomial-time. The delicate nature of this task has severely limited the development of techniques. For a survey of recent results in average-case complexity, we refer to \cite{bogdanov2006average}.

In addition to those showing lower bounds for sparse PCA, there have been a number of average-case reductions from planted clique in both the computer science and statistics literature. These include reductions to testing $k$-wise independence \cite{alon2007testing}, biclustering detection and recovery \cite{ma2015computational, cai2015computational, caiwu2018}, planted dense subgraph \cite{hajek2015computational}, RIP certification \cite{wang2016average, koiran2014hidden}, matrix completion \cite{chen2015incoherence}, minimum circuit size and minimum Kolmogorov time-bounded complexity \cite{hirahara2017average} and sparse PCA \cite{berthet2013optimal, berthet2013complexity, wang2016statistical, gao2017sparse}. A web of reductions from planted clique to planted independent set, planted dense subgraph, sparse PCA, the spiked Wigner model, gaussian biclustering and the subgraph stochastic block model was given in \cite{brennan2018reducibility}. The planted clique conjecture has also been used as a hardness assumption for average-case reductions in cryptography \cite{juels2000hiding, applebaum2010public}, as described in Sections 2.1 and 6 of \cite{Barak2017}.

A number of average-case reductions in the literature have started with different average-case assumptions than the planted clique conjecture. Variants of planted dense subgraph have been used to show hardness in a model of financial derivatives under asymmetric information \cite{arora2011computational}, link prediction \cite{baldin2018optimal}, finding dense common subgraphs \cite{charikar2018finding} and online local learning of the size of a label set \cite{awasthi2015label}. Hardness conjectures for random constraint satisfaction problems have been used to show hardness in improper learning complexity \cite{daniely2014average}, learning DNFs \cite{daniely2016complexity} and hardness of approximation \cite{feige2002relations}. There has also been a recent reduction from a hypergraph variant of the planted clique conjecture to tensor PCA \cite{zhang2017tensor}.

\subsection{Notation}

In this paper, we adopt the following notation. Let $\mL(X)$ denote the distribution law of a random variable $X$ and given two laws $\mL_1$ and $\mL_2$, let $\mL_1 + \mL_2$ denote $\mL(X + Y)$ where $X \sim \mL_1$ and $Y \sim \mL_2$ are independent. Given a distribution $\mathcal{P}$, let $\mathcal{P}^{\otimes n}$ denote the distribution of $(X_1, X_2, \dots, X_n)$ where the $X_i$ are i.i.d. according to $\mathcal{P}$. Similarly, let $\mathcal{P}^{\otimes m \times n}$ denote the distribution on $\mathbb{R}^{m \times n}$ with i.i.d. entries distributed as $\mathcal{P}$. Given a finite or measurable set $\mathcal{X}$, let $\text{Unif}[\mathcal{X}]$ denote the uniform distribution on $\mathcal{X}$. Let $\TV$ and $\chi^2$ denote total variation distance and $\chi^2$ divergence, respectively. Let $\mN(\mu, \Sigma)$ denote a multivariate normal random vector with mean $\mu \in \mathbb{R}^d$ and covariance matrix $\Sigma$, where $\Sigma$ is a $d \times d$ positive semidefinite matrix. Let $\chi^2(k)$ denote a $\chi^2$-distribution with $k$ degrees of freedom. Let $\mathcal{B}_0(k)$ denote the set of all unit vectors $v \in \mathbb{R}^d$ with $\| v \|_0 \le k$. Let $[n] = \{1, 2, \dots, n\}$, $\mG_n$ be the set of simple graphs on $n$ vertices and let the Orthogonal group on $\bR^{d \times d}$ be $\mO_d$. Let $\mathbf{1}_S$ denote the vector $v \in \mathbb{R}^n$ with $v_i = 1$ if $i \in S$ and $v_i = 0$ if $i \not \in S$ where $S \subseteq [n]$.

\section{Problem Formulations, Algorithms and Statistical Limits}
\label{sec:problems}

We consider detection problems $\mP$, wherein the algorithm is given a set of observations and tasked with distinguishing between two hypotheses:
\begin{itemize}
\item a \emph{uniform} hypothesis $H_0$, under which observations are generated from the natural noise distribution for the problem; and
\item a \emph{planted} hypothesis $H_1$, under which observations are generated from the same noise distribution but with a latent planted sparse structure.
\end{itemize}
In all of the detection problems we consider, $H_0$ is a simple hypothesis consisting of a single distribution and $H_1$ is either also simple or a composite hypothesis consisting of several distributions. Typically, $H_0$ consists of the canonical noise distribution and $H_1$ either consists of the set of observation models associated with each possible planted sparse structure or a mixture over them.  When $H_1$ is a composite hypothesis, it consists of a set of distributions $P_\theta$ where the parameter $\theta$ varies over a set $\Theta$. We will also denote the set $\{ P_\theta : \theta \in \Theta \}$, which is a singleton if $H_1$ is simple, of distributions as $H_1$ for notational convenience. As discussed in \cite{brennan2018reducibility} and \cite{hajek2015computational}, lower bounds for simple vs. simple hypothesis testing formulations are stronger and technically more difficult than for formulations involving composite hypotheses. The reductions in \cite{berthet2013complexity, wang2016statistical} are to composite vs. composite formulations of sparse PCA, the reduction in \cite{gao2017sparse} is to a simple vs. composite formulation of the spiked covariance model and, when $k^2 = \tilde{O}(n)$, the reduction in \cite{brennan2018reducibility} is to the simple vs. simple formulation of the spiked covariance model.

Given an observation $X$, an algorithm $\mathcal{A}(X) \in \{0, 1\}$ \emph{solves} the detection problem \emph{with nontrivial probability} if there is an $\epsilon > 0$ such that its Type I$+$II error satisfies that
$$\limsup_{n \to \infty} \left( \bP_{H_0}[\mathcal{A}(X) = 1] + \sup_{P \in H_1} \bP_{X \sim P}[\mathcal{A}(X) = 0] \right) \le 1 - \epsilon$$
where $n$ is the parameter indicating the size of $X$. We refer to this quantity as the asymptotic Type I$+$II error of $\mathcal{A}$ for the problem $\mP$. If the asymptotic Type I$+$II error of $\mathcal{A}$ is zero, then we say $\mathcal{A}$ \emph{solves} the detection problem $\mP$. A simple consequence of this definition is that if $\mathcal{A}$ achieves asymptotic Type I$+$II error $1 - \epsilon$ for a composite testing problem with hypotheses $H_0$ and $H_1$, then it also achieves this same error on the simple problem with hypotheses $H_0$ and $H_1' : X \sim P$ where $P$ is any mixture of the distributions in $H_1$. We remark that simple vs. simple formulations are the hypothesis testing problems that correspond to average-case decision problems $(L, \mathcal{D})$ as in Levin's theory of average-case complexity. In particular, showing a lower bound in Type I$+$II error for a simple vs. simple formulation shows that no polynomial time algorithm can solve $(L, \mathcal{D})$ with probability greater than $\frac{1}{2} + o(1)$ where $L$ is a language separating $H_0$ and $H_1$ and $\mathcal{D} = \frac{1}{2} \cdot \mL_{H_0} + \frac{1}{2} \cdot \mL_{H_1}$. For example, $L : \mG_n \to \{0, 1\}$ could correspond to $\pr{clique}(G) \ge k$ for planted clique. A survey of average-case complexity can be found in \cite{bogdanov2006average}.

We now formally define planted clique, planted dense subgraph and sparse PCA. Let $\mG(n, p)$ denote an Erd\H{o}s-R\'{e}nyi random graph with edge probability $p$. Let $\mG(n, k, p, q)$ denote the graph formed by sampling $\mG(n, q)$ and replacing the induced graph on a subset $S$ of size $k$ chosen uniformly at random with a sample from $\mG(k, p)$. The planted dense subgraph problem is defined as follows.

\begin{definition}[Planted Dense Subgraph] The detection problem $\pr{pds}(n, k, p, q)$ has hypotheses
$$H_0: G \sim \mG(n, q) \quad \text{and} \quad H_1 : G \sim \mG(n, k, p, q)$$
\end{definition}

The planted clique problem $\pr{pc}(n, k, p)$ is then $\pr{pds}(n, k, 1, p)$. There are many polynomial-time algorithms in the literature that find the planted clique in $\mG(n, k, p)$, including approximate message passing, semidefinite programming, nuclear norm minimization and several combinatorial approaches \cite{feige2000finding, mcsherry2001spectral, feige2010finding, ames2011nuclear, dekel2014finding, deshpande2015finding, chen2016statistical}. All of these algorithms require that $k = \Omega(\sqrt{n})$ if $p$ is constant, despite the fact that the largest clique in $G(n, p)$ contains $O(\log n)$ vertices with high probability. This leads to the following conjecture.

\begin{conjecture}[\pr{pc} Conjecture]
Fix some constant $p \in (0, 1)$. Suppose that $\{ \mathcal{A}_n \}$ is a sequence of randomized polynomial time algorithms $\mathcal{A}_n : \mG_n \to \{0, 1\}$ and $k_n$ is a sequence of positive integers satisfying that $\limsup_{n \to \infty} \log_n k_n < \frac{1}{2}$. Then if $G$ is an instance of $\pr{pc}(n, k, p)$, it holds that
$$\liminf_{n \to \infty} \left( \bP_{H_0}\left[\mathcal{A}_n(G) = 1\right] + \bP_{H_1}\left[\mathcal{A}_n(G) = 0\right] \right) \ge 1.$$ 
\end{conjecture}

All of our hardness results can be obtained from this conjecture with $p = 1/2$. The $\pr{pc}$ Conjecture implies a similar barrier at $k = o(\sqrt{n})$ for $\pr{pds}(n, k, p, q)$ where $p \neq q$ are constants, which we refer to as the $\pr{pds}$ conjecture. If $p > q$, then this follows from the hardness of $\pr{pc}(n, k, q/p)$ and the reduction keeping each edge in the graph with probability $p$. If $p < q$, then applying the same reduction to the complement graph yields this barrier. %We remark however that while planted clique can be solved in $n^{O(\log n)}$ time by exhaustively searching for $O(\log n)$ sized cliques, no such quasipolynomial time algorithm is known for planted dense subgraph.
We now give several formulations of the sparse PCA detection problem in the spiked covariance model, several of which were considered in \cite{brennan2018reducibility}. The following is a general composite hypothesis testing formulation of the spiked covariance model that we will specialize afterwards, including as a simple hypothesis testing problem.

\begin{definition}[Spiked Covariance Model of Sparse PCA]
The sparse PCA detection problem in the spiked covariance model $\pr{spca}(n, k, d, \theta, A_\theta, B_k)$ has hypotheses
\begin{align*}
&H_0: X_1, X_2, \dots, X_n \sim \mN(0, I_d)^{\otimes n} \quad \text{and} \\
&H_1 : X_1, X_2, \dots, X_n \sim \mN\left(0, I_d + \theta' vv^\top\right)^{\otimes n} \text{ where } \theta' \in A_\theta \text{ and } v \in B_k
\end{align*}
where $A_\theta$ is a subset of the positive real numbers and $B_k$ is a subset of the set of $k$-sparse unit vectors in $\mathbb{S}^{d - 1}$.
\end{definition}

Specializing this general definition with different sets $A_\theta$ and $B_k$, we arrive at the following variants of the spiked covariance model:
\begin{itemize}
\item \textit{Uniform biased sparse PCA} is denoted as $\pr{ubspca}(n, k, d, \theta)$ where $A_\theta = \{ \theta \}$ and $B_k$ is the set $\mathcal{S}_{k, d}$ of all $k$-sparse unit vectors in $\mathbb{S}^{d - 1}$ with nonzero coordinates all equal to $1/\sqrt{k}$.
\item \textit{Composite biased sparse PCA} is denoted as $\pr{cbspca}(n, k, d, \theta)$ where
\begin{align*}
A_\theta &= \left[ \theta \left(1 - \frac{\gamma}{\sqrt{k}} \right), \theta \left(1 + \frac{\gamma}{\sqrt{k}} \right) \right] \\
B_k &= \bigcup_{k' = k - \gamma \sqrt{k}}^{k} \mathcal{S}_{k', d}
\end{align*}
for some fixed parameter $\tau$ satisfying that $\gamma \to \infty$ as $n \to \infty$.
\item \textit{Fully composite unbiased sparse PCA} is denoted as $\pr{fcspca}(n, k, d, \theta)$ where $A_\theta$ is the same as in $\pr{cbspca}$ and
$$B_k = \left\{ v \in \mathbb{S}^{d-1}  : k - \tau \sqrt{k} \le \| v \|_0 \le k \text{ and } |v_i| \ge \frac{1}{\sqrt{k}} \text{ for } i \in \text{supp}(v) \right\}$$
\end{itemize}
The canonical formulation of sparse PCA in the spiked covariance model is the simple vs. simple hypothesis testing formulation where $v$ is drawn uniformly at random from $\mathcal{S}_{k, d}$ under $H_1$. Note that randomly permuting the $d$ coordinates of any instance of $\pr{ubspca}$ exactly produces an instance of this simple vs. simple formulation. This implies that the two models are equivalent in the sense that an algorithm $\mathcal{A}$ solves $\pr{ubspca}$ if and only if there is an algorithm $\mathcal{A'}$ solving this formulation with a runtime differing from that of $\mathcal{A}$ by at most $\text{poly}(n, d)$. Note that an instance of $\pr{ubspca}$ is an instance of $\pr{cbspca}$ which is an instance of $\pr{fcspca}$, and thus lower bounds for these formulations are decreasing in strength. The lower bounds from our main theorem are for the strongest formulation $\pr{ubspca}$ and the complete set of tight lower bounds in our second main theorem are for $\pr{fcspca}$. The relationships between these and other formulations of the spiked covariance model of PCA are further discussed in Appendix \ref{sec:relationshipsspca}.

\paragraph{Algorithms and Statistical Limits.} As mentioned in the introduction, the statistical limit of detection in the spiked covariance model is $\theta = \tilde{\Theta}(\sqrt{k/n})$ and the best-known algorithms solve the problem if $\theta = \tilde{\Omega}(\sqrt{k^2/n})$. The information-theoretic lower bound for $\theta = \tilde{O}(\sqrt{k/n})$ can be found in \cite{berthet2013optimal} if $k^2 = o(d)$ and for all $k \le d$ in \cite{cai2015optimal}. Upper bounds at these two barriers are shown in \cite{berthet2013complexity} if $k = o(\sqrt{d})$ and $k = o(\sqrt{n})$ with the following two algorithms:
\begin{enumerate}
\item \textbf{Semidefinite Programming:} Form the empirical covariance matrix $\hat{\Sigma} = \frac{1}{n} \sum_{i = 1}^n X_i X_i^\top$ and solve the convex program
\begin{align*}
\max_Z \quad &\text{Tr}\left(\hat{\Sigma} Z\right) \\
\text{s.t.} \quad &\text{Tr}(Z) = 1, |Z|_1 \le k, Z \succeq 0
\end{align*}
Thresholding the resulting maximum solves the detection problem as long as $\theta = \tilde{\Omega}(\sqrt{k^2/n})$.
\item \textbf{$k$-Sparse Eigenvalue:} Compute and threshold the $k$-sparse unit vector $u$ that maximizes $u^\top \hat{\Sigma} u$. This can be found by finding the largest eigenvector of each $k \times k$ principal submatrix of $\hat{\Sigma}$. This succeeds as long as $\theta = \tilde{\Omega}(\sqrt{k/n})$.
\end{enumerate}
Note that the semidefinite programming algorithm runs in polynomial time while $k$-sparse eigenvalue in exponential time.

\section{Background on Average-Case Reductions}
\label{sec:avgcasereductions}

\subsection{Reductions in Total Variation and the Computational Model}

As introduced in \cite{berthet2013complexity} and \cite{ma2015computational}, we give approximate reductions in total variation to show that lower bounds for one hypothesis testing problem imply lower bounds for another. These reductions yield an exact correspondence between the asymptotic Type I$+$II errors of the two problems. This is formalized in the following lemma, which is Lemma 3.1 from \cite{brennan2018reducibility}. Its proof is short and follows from the definition of total variation.

\begin{lemma}[Lemma 3.1 in \cite{brennan2018reducibility}] \label{lem:3a}
Let $\mP$ and $\mP'$ be detection problems with hypotheses $H_0, H_1$ and  $H_0', H_1'$, respectively. Let $X$ be an instance of $\mathcal{P}$ and let $Y$ be an instance of $\mP'$. Suppose there is a polynomial time computable map $\mathcal{A}$ satisfying
$$\TV\left(\mL_{H_0}(\mathcal{A}(X)), \mL_{H_0'}(Y)\right) + \sup_{P \in H_1} \inf_{\pi \in \Delta(H_1')} \TV\left( \mL_{P}(\mathcal{A}(X)), \int_{H_1'} \mL_{P'}(Y) d\pi(P') \right) \le \delta$$
If there is a randomized polynomial time algorithm solving $\mP'$ with Type I$+$II error at most $\epsilon$, then there is a randomized polynomial time algorithm solving $\mP$ with Type I$+$II error at most $\epsilon + \delta$.
\end{lemma}

If $\delta = o(1)$, then given a blackbox solver $\mathcal{B}$ for $\mathcal{P'}$, the algorithm that applies $\mathcal{A}$ and then $\mathcal{B}$ solves $\mathcal{P}$ and requires only a single query to the blackbox. An algorithm that runs in randomized polynomial time refers to one that has access to $\text{poly}(n)$ independent random bits and must run in $\text{poly}(n)$ time where $n$ is the size of the instance of the problem. For clarity of exposition, in our reductions we assume that explicit expressions can be exactly computed and assume that $\mN(0, 1)$ and random variables can be sampled in $O(1)$ operations. Note that this implies that $\chi^2(n)$ random variables can be sampled in $O(n)$ operations. %While these assumptions ease the presentation of our reductions, they are unnecessary. In Section \ref{sec:conclusion}, we outline how our reduction implies a reduction to the spiked covariance model where entries are represented to a finite bit precision and no longer requires the ability to sample $\mN(0, 1)$ in $O(1)$ operations.

\subsection{Properties of Total Variation}

Throughout the proof of our main theorem, we will need a number of well-known facts and inequalities concerning total variation distance. %These are stated below and their proofs are included in Appendix \ref{sec:appendix} for completeness.

\begin{fact} \label{tvfacts}
The distance $\TV$ satisfies the following properties:
\begin{enumerate}
\item (Triangle Inequality) Given three distributions $P, Q$ and $R$ on a measurable space $(\mathcal{X}, \mathcal{B})$, it follows that
$$\TV\left( P, Q \right) \le \TV\left( P, R \right) + \TV\left( Q, R \right)$$
\item (Data Processing) Let $P$ and $Q$ be distributions on a measurable space $(\mathcal{X}, \mathcal{B})$ and let $f : \mathcal{X} \to \mathcal{Y}$ be a Markov transition kernel. If $A \sim P$ and $B \sim Q$ then
$$\TV\left(\mL(f(A)), \mL(f(B))\right) \le \TV(P, Q)$$
\item (Tensorization) Let $P_1, P_2, \dots, P_n$ and $Q_1, Q_2, \dots, Q_n$ be distributions on a measurable space $(\mathcal{X}, \mathcal{B})$. Then
$$\TV\left( \prod_{i = 1}^n P_i, \prod_{i = 1}^n Q_i \right) \le \sum_{i = 1}^n \TV\left( P_i, Q_i \right)$$
\item (Conditioning on an Event) For any distribution $P$ on a measurable space $(\mathcal{X}, \mathcal{B})$ and event $A \in \mathcal{B}$, it holds that
$$\TV\left( P(\cdot | A), P \right) = 1 - P(A)$$
\item (Conditioning on a Random Variable) For any two pairs of random variables $(X, Y)$ and $(X', Y')$ each taking values in a measurable space $(\mathcal{X}, \mathcal{B})$, it holds that
$$\TV\left( \mL(X), \mL(X') \right) \le \TV\left( \mL(Y), \mL(Y') \right) + \bE_{y \sim Y} \left[ \TV\left( \mL(X | Y = y), \mL(X' | Y' = y) \right)\right]$$
where we define $\TV\left( \mL(X | Y = y), \mL(X' | Y' = y) \right) = 1$ for all $y \not \in \textnormal{supp}(Y')$.
\end{enumerate}
\end{fact}

Given an algorithm $\mathcal{A}$ and distribution $\mP$ on inputs, let $\mathcal{A}(\mP)$ denote the distribution of $\mathcal{A}(X)$ induced by $X \sim \mP$. If $\mathcal{A}$ has $k$ steps, let $\mathcal{A}_i$ denote the $i$th step of $\mathcal{A}$ and $\mathcal{A}_{i\text{-}j}$ denote the procedure formed by steps $i$ through $j$. Each time this notation is used, we clarify the intended initial and final variables when $\mathcal{A}_{i}$ and $\mathcal{A}_{i\text{-}j}$ are viewed as Markov kernels. The next lemma encapsulates the structure of all of our analyses of average-case reductions.

\begin{lemma} \label{lem:tvacc}
Let $\mathcal{A}$ be an algorithm that can be written as $\mathcal{A} = \mathcal{A}_m \circ \mathcal{A}_{m-1} \circ \cdots \circ \mathcal{A}_1$ for a sequence of steps $\mathcal{A}_1, \mathcal{A}_2, \dots, \mathcal{A}_m$. Suppose that the probability distributions $\mP_0, \mP_1, \dots, \mP_m$ are such that $\TV(\mathcal{A}_i(\mP_{i-1}), \mP_i) \le \epsilon_i$ for each $1 \le i \le m$. Then it follows that
$$\TV\left( \mathcal{A}(\mP_0), \mP_m \right) \le \sum_{i = 1}^m \epsilon_i$$
\end{lemma}

\begin{proof}
This follows from a simple induction on $m$. Note that the case when $m = 1$ follows by definition. Now observe that by the data-processing and triangle inequalities in Fact \ref{tvfacts}, we have that if $\mathcal{B} = \mathcal{A}_{m-1} \circ \mathcal{A}_{m-2} \circ \cdots \circ \mathcal{A}_1$ then
\begin{align*}
\TV\left( \mathcal{A}(\mP_0), \mP_m \right) &\le \TV\left( \mathcal{A}_m \circ \mathcal{B}(\mP_0), \mathcal{A}_m(\mP_{m - 1}) \right) + \TV\left(\mathcal{A}_m(\mP_{m - 1}), \mP_m \right) \\
&\le \TV\left( \mathcal{B}(\mP_0), \mP_{m - 1} \right) + \epsilon_m \\
&\le \sum_{i = 1}^m \epsilon_i
\end{align*}
where the last inequality follows from the induction hypothesis applied with $m - 1$ to $\mathcal{B}$. This completes the induction and proves the lemma.
\end{proof}

\section{Mapping to Submatrix Problems and $\chi^2$ Random Rotations}
\label{sec:mappingsubmatrix}

In this section, we introduce a number of subroutines that we will use in both of our reductions to sparse PCA in Sections \ref{sec:wishartreduction} and \ref{sec:randomrotations}.

\subsection{Graph and Bernoulli Matrix Cloning}

We begin with $\textsc{Graph-Clone}$ and $\textsc{Bernoulli-Matrix-Clone}$, shown in Figure \ref{fig:clone}. The procedure $\textsc{Graph-Clone}$ was introduced in \cite{brennan2019universality} and produces several independent samples from a planted subgraph problems given a single sample. The procedure $\textsc{Bernoulli-Matrix-Clone}$ is nearly identical, producing independent samples from planted Bernoulli submatrix problems given a single sample. Their properties as a Markov kernel are captured in the next two lemmas.

Let $\mG(n, S, p, q)$ denote the distribution of planted dense subgraph instances from $\mG(n, k, p, q)$ conditioned on the subgraph being planted on the vertex set $S$ where $|S| = k$. Similarly let $\mathcal{M}(n, S, p, q)$ denote the distribution of matrices in $\{0, 1\}^{n \times n}$ with independent entries where $M_{ij} \sim \text{Bern}(p)$ if $i, j \in S$ and $M_{ij} \sim \text{Bern}(q)$ otherwise. Let $\mathcal{M}(n, k, p, q)$ denote the mixture of $\mathcal{M}(n, S, p, q)$ induced by picking the $k$-subset $S \subseteq [n]$ uniformly at random.

\begin{figure}[t!]
\begin{algbox}
\textbf{Algorithm} \textsc{Graph-Clone}

\vspace{1mm}

\textit{Inputs}: Graph $G \in \mG_n$, the number of copies $t$, parameters $0 < q < p \le 1$ and $0 < Q < P \le 1$ satisfying $\frac{1 - p}{1 - q} \le \left( \frac{1 - P}{1 - Q} \right)^t$ and $\left( \frac{P}{Q} \right)^t \le \frac{p}{q}$

\begin{enumerate}
\item Generate $x^{ij} \in \{0, 1\}^t$ for each $1 \le i < j \le n$ such that:
\begin{itemize}
\item If $\{i, j \} \in E(G)$, sample $x^{ij}$ from the distribution on $\{0, 1\}^t$ with
$$\bP[x^{ij} = v] = \frac{1}{p - q} \left[ (1 - q) \cdot P^{|v|_1} (1 - P)^{t - |v|_1} - (1 - p) \cdot Q^{|v|_1} (1 - Q)^{t - |v|_1} \right]$$
\item If $\{i, j \} \not \in E(G)$, sample $x^{ij}$ from the distribution on $\{0, 1\}^t$ with
$$\bP[x^{ij} = v] = \frac{1}{p - q} \left[ p \cdot Q^{|v|_1} (1 - Q)^{t - |v|_1} - q \cdot P^{|v|_1} (1 - P)^{t - |v|_1} \right]$$
\end{itemize}
\item Output the graphs $(G_1, G_2, \dots, G_t)$ where $\{i, j\} \in E(G_k)$ if and only if $x^{ij}_k = 1$
\end{enumerate}
\vspace{1mm}

\textbf{Algorithm} \textsc{Bernoulli-Matrix-Clone}

\vspace{1mm}

\textit{Inputs}: Matrix $M \in \{0, 1\}^{n\times n}$, the number of copies $t$, parameters $0 < q < p \le 1$ and $0 < Q < P \le 1$ satisfying $\frac{1 - p}{1 - q} \le \left( \frac{1 - P}{1 - Q} \right)^t$ and $\left( \frac{P}{Q} \right)^t \le \frac{p}{q}$

\begin{enumerate}
\item Generate the 3-tensor $X \in \{0, 1\}^{t \times n \times n}$ such that for all $i, j \in [n]$:
\begin{itemize}
\item If $M_{ij} = 1$, the row $X_{\cdot ij}$ is sampled from the distribution on $\{0, 1\}^t$ with
$$\bP[X_{\cdot ij} = v] = \frac{1}{p - q} \left[ (1 - q) \cdot P^{|v|_1} (1 - P)^{t - |v|_1} - (1 - p) \cdot Q^{|v|_1} (1 - Q)^{t - |v|_1} \right]$$
\item If $M_{ij} = 0$, the row $X_{\cdot ij}$ is sampled from the distribution on $\{0, 1\}^t$ with
$$\bP[X_{\cdot ij} = v] = \frac{1}{p - q} \left[ p \cdot Q^{|v|_1} (1 - Q)^{t - |v|_1} - q \cdot P^{|v|_1} (1 - P)^{t - |v|_1} \right]$$
\end{itemize}
\item Output the matrices $(M_1, M_2, \dots, M_t)$ where $M_i$ is the $i$th layer $X_{i \cdot \cdot}$ of $X$
\end{enumerate}
\vspace{1mm}

\end{algbox}
\caption{Subroutines $\textsc{Graph-Clone}$ and $\textsc{Bernoulli-Matrix-Clone}$ for producing independent samples from planted graph and Bernoulli submatrix problems.}
\label{fig:clone}
\end{figure}

\begin{lemma}[Lemma 5.2 in \cite{brennan2019universality}] \label{lem:graphcloning}
Let $t \in \mathbb{N}$, $0 < q < p \le 1$ and $0 < Q < P \le 1$ satisfy that
$$\frac{1 - p}{1 - q} \le \left( \frac{1 - P}{1 - Q} \right)^t \quad \text{and} \quad \left( \frac{P}{Q} \right)^t \le \frac{p}{q}$$
Then the algorithm $\mathcal{A} = \textsc{Graph-Clone}$ runs in $\textnormal{poly}(t, n)$ time and satisfies that for each $S \subseteq [n]$,
$$\mathcal{A}\left( \mG(n, q) \right) \sim \mG(n, Q)^{\otimes t} \quad \text{and} \quad \mathcal{A}\left( \mG(n, S, p, q) \right) \sim \mG(n, S, P, Q)^{\otimes t}$$
\end{lemma}

\begin{lemma}[Bernoulli Matrix Clonng] \label{lem:matrixcloning}
Let $t \in \mathbb{N}$, $0 < q < p \le 1$ and $0 < Q < P \le 1$ satisfy that
$$\frac{1 - p}{1 - q} \le \left( \frac{1 - P}{1 - Q} \right)^t \quad \text{and} \quad \left( \frac{P}{Q} \right)^t \le \frac{p}{q}$$
Then the algorithm $\mathcal{A} = \textsc{Bernoulli-Matrix-Clone}$ runs in $\textnormal{poly}(t, n)$ time and satisfies that for each $S \subseteq [n]$,
$$\mathcal{A}\left( \textnormal{Bern}(q)^{\otimes n \times n} \right) \sim \left( \textnormal{Bern}(Q)^{\otimes n \times n} \right)^{\otimes t} \quad \text{and} \quad \mathcal{A}\left( \mathcal{M}(n, S, p, q) \right) \sim \mathcal{M}(n, S, P, Q)^{\otimes t}$$
\end{lemma}

\begin{proof}
Let $R_0(v) = \bP[X_{\cdot ij} = v | M_{ij} = 1]$ and $R_1(v) = \bP[X_{\cdot ij} = v | M_{ij} = 0]$ as defined in Figure \ref{fig:clone}. As shown in the proof of Lemma 5.2 of \cite{brennan2019universality}, the given conditions on $p, q, P$ and $Q$ imply that $R_0$ and $R_1$ are well-defined probability distributions satisfying the mixture identities
\begin{align*}
(1 - p) \cdot R_0 + p \cdot R_1 &= \text{Bern}(P)^{\otimes t} \\
(1 - q) \cdot R_0 + q \cdot R_1 &= \text{Bern}(Q)^{\otimes t}
\end{align*}
Therefore if $M_{ij} \sim \text{Bern}(p)$, then $(X_{1ij}, X_{2ij}, \dots, X_{tij}) \sim \text{Bern}(P)^{\otimes t}$ and if $M_{ij} \sim \text{Bern}(q)$, then $(X_{1ij}, X_{2ij}, \dots, X_{tij}) \sim \text{Bern}(Q)^{\otimes t}$. Since the entries of $M$ are independent in each of $\text{Bern}(q)^{\otimes n \times n}$ and $\mathcal{M}(n, S, p, q)$, this implies that $(M_1, M_2, \dots, M_t) \sim \left( \textnormal{Bern}(Q)^{\otimes n \times n} \right)^{\otimes t}$ if $M \sim \textnormal{Bern}(q)^{\otimes n \times n}$ and $(M_1, M_2, \dots, M_t) \sim \mathcal{M}(n, S, P, Q)^{\otimes t}$ if $M \sim \mathcal{M}(n, S, p, q)$, completing the proof of the lemma.
\end{proof}

\subsection{Embedding Adjacency Matrices as a Principal Minors and Gaussianization}

\begin{figure}[t!]
\begin{algbox}
\textbf{Algorithm} \textsc{To-Bernoulli-Submatrix}

\vspace{1mm}

\textit{Inputs}: Graph $G \in \mG_n$, edge probabilities $0 < q < p \le 1$ with $q = n^{-O(1)}$ and target dimension $N \ge \left(\frac{p}{Q} + \Omega(1) \right)n$ where $Q = 1 - \sqrt{(1 - p)(1 - q)} + \mathbf{1}_{\{p = 1\}} \left( \sqrt{q} - 1 \right)$
\begin{enumerate}
\item Apply $\textsc{Graph-Clone}$ to $G$ with edge probabilities $P = p$ and $Q = 1 - \sqrt{(1 - p)(1 - q)} + \mathbf{1}_{\{p = 1\}} \left( \sqrt{q} - 1 \right)$ and $t = 2$ clones to obtain $(G_1, G_2)$
\item Sample $s_1 \sim \text{Bin}(n, p)$, $s_2 \sim \text{Bin}(N, Q)$ and a set $S \subseteq [N]$ with $|S| = n$ uniformly at random. Sample $T_1 \subseteq S$ and $T_2 \subseteq [N] \backslash S$ with $|T_1| = s_1$ and $|T_2| = \max\{ s_2 - s_1, 0 \}$ uniformly at random. Now form the matrix $M^1 \in \{0, 1\}^{N \times N}$ where
$$M^1_{ij} = \left\{ \begin{array}{ll} \mathbf{1}_{\{\pi(i), \pi(j)\} \in E(G_1)} & \text{if } i < j \text{ and } i, j \in S \\ \mathbf{1}_{\{\pi(i), \pi(j)\} \in E(G_2)} & \text{if } i > j \text{ and } i, j \in S \\ \mathbf{1}_{\{ i \in T_1 \}} & \text{if } i = j \text{ and } i, j \in S \\ \mathbf{1}_{\{i \in T_2\}} & \text{if } i = j \text{ and } i, j \not \in S \\ \sim_{\text{i.i.d.}} \text{Bern}(Q) & \text{if } i \neq j \text{ and } i \not \in S \text{ or } j \not \in S \end{array} \right.$$
where $\pi : S \to [n]$ is a bijection chosen uniformly at random
\item Output the matrix $M$
\end{enumerate}
\vspace{1mm}

\end{algbox}
\caption{Subroutine $\textsc{To-Bernoulli-Submatrix}$ for mapping from a planted dense subgraph to a Bernoulli submatrix problem.}
\label{fig:tosubmatrix}
\end{figure}

We now describe $\textsc{To-Bernoulli-Submatrix}$, which is shown in \ref{fig:tosubmatrix}. This reduction is the first two steps of $\textsc{To-Submatrix}$ in \cite{brennan2019universality}. These steps: (1) clone the upper half of the adjacency matrix of the input graph problem to produce an independent lower half; and (2) plant diagonal entries while randomly embedding this matrix as a principal minor into a larger matrix to hide the planted diagonal entries in total variation. The first part of the proof of Theorem 5.1 from \cite{brennan2019universality} yields the total variation guarantees for $\textsc{To-Bernoulli-Submatrix}$, stated in the lemma below. We remark that reductions to submatrix problems prior to \cite{brennan2019universality} either: (1) hid planted diagonals in total variation by randomly permuting column indices without growing the dimension of the submatrix problem and thus forfeited the symmetry of the submatrix problem; or (2) required the $\pr{pc}$ conjecture for $k = n^{1/2 - \epsilon}$ and used an averaging trick to produce the diagonal entries. The diagonal planting trick in Step 2 of $\textsc{To-Bernoulli-Submatrix}$ maintains the symmetry of the submatrix problem without requiring this strong form of the $\pr{pc}$ conjecture. This will be crucial in our reduction to sparse PCA in Section \ref{sec:wishartreduction}.

\begin{lemma}[Reducing to Bernoulli Submatrix Problems] \label{lem:submatrix}
Let $0 < q < p \le 1$ be such that $q = n^{-O(1)}$ and let $N \ge \left(\frac{p}{Q} + \epsilon \right)n$ where $Q = 1 - \sqrt{(1 - p)(1 - q)} + \mathbf{1}_{\{p = 1\}} \left( \sqrt{q} - 1 \right)$ and $\epsilon > 0$. Let $k \le n$ satisfy
$$k \le \frac{Q \epsilon n}{2} \quad \text{and} \quad \frac{k^2}{N} \le \min \left\{ \frac{Q}{1 - Q}, \frac{1 - Q}{Q} \right\}$$
The algorithm $\mathcal{A} = \textsc{To-Bernoulli-Submatrix}$ runs in $\textnormal{poly}(N)$ time and satisfies that
\begin{align*}
\TV\left( \mathcal{A}(\mG(n, k, p, q)), \, \mathcal{M}\left(N, k, p, Q \right) \right) &\le 4\cdot \exp\left( - \frac{Q^2 \epsilon^2 n^2}{16N} \right) \\
&\quad \quad + \sqrt{\frac{k^2(1 - Q)}{2NQ}} + \sqrt{\frac{k^2Q}{2N(1 - Q)}} \\
\TV\left( \mathcal{A}(\mG(n, q)), \, \textnormal{Bern}\left( Q \right)^{\otimes N \times N} \right) &\le 4 \cdot \exp\left( - \frac{Q^2 \epsilon^2 n^2}{16N} \right)
\end{align*}
\end{lemma}

To describe the next subroutine $\textsc{Gaussianize}$, we first will need the rejection kernel framework introduced in \cite{brennan2018reducibility}. The next lemma captures the total variation guarantees of the Gaussian rejection kernels shown in Figure \ref{fig:rej-kernel}.

\begin{figure}[t!]
\begin{algbox}
\textbf{Algorithm} $\textsc{rk}_G(\mu, B)$

\vspace{2mm}

\textit{Parameters}: Input $B \in \{0, 1\}$, Bernoulli probabilities $0 < q < p \le 1$, Gaussian mean $\mu$, number of iterations $N$, let $\varphi_\mu(x) = \frac{1}{\sqrt{2\pi}} \cdot \exp\left(- \frac{1}{2}(x - \mu)^2 \right)$ denote the density of $\mN(\mu, 1)$
\begin{enumerate}
\item Initialize $z \gets 0$
\item Until $z$ is set or $N$ iterations have elapsed:
\begin{enumerate}
\item[(1)] Sample $z' \sim \mN(0, 1)$ independently
\item[(2)] If $B = 0$, if the condition
$$p \cdot \varphi_0(z') \ge q \cdot \varphi_{\mu}(z')$$
holds, then set $z \gets z'$ with probability $1 - \frac{q \cdot \varphi_\mu(z')}{p \cdot \varphi_0(z')}$
\item[(3)] If $B = 1$, if the condition
$$(1 - q) \cdot \varphi_\mu(z' + \mu) \ge (1 - p) \cdot \varphi_0(z' + \mu)$$
holds, then set $z \gets z' + \mu$ with probability $1 - \frac{(1 - p) \cdot \varphi_0(z' + \mu)}{(1 - q) \cdot \varphi_\mu(z' + \mu)}$
\end{enumerate}
\item Output $z$
\end{enumerate}
\vspace{1mm}
\textbf{Algorithm} $\textsc{Gaussianize}$

\vspace{2mm}

\textit{Parameters}: Matrix $M \in \{0, 1\}^{N \times N}$, Bernoulli probabilities $0 < Q < P \le 1$ with $Q = n^{-O(1)}$ and a target mean matrix $0 \le \mu_{ij} \le \tau$ where $\tau > 0$ is a parameter
\begin{enumerate}
\item Form the matrix $X \in \mathbb{R}^{N \times N}$ by setting
$$X_{ij} \gets \textsc{rk}_{G}(\mu_{ij}, M_{ij})$$
for each $i, j \in [N]$ where each $\textsc{rk}_{G}$ is run with $N_{\text{it}} = \lceil 6\delta^{-1} \log N \rceil$ iterations where $\delta = \min \left\{ \log \left( \frac{P}{Q} \right), \log \left( \frac{1 - Q}{1 - P} \right) \right\}$
\item Output the matrix $X$
\end{enumerate}
\vspace{1mm}
\end{algbox}
\caption{Gaussian instantiation of the rejection kernel algorithm from \cite{brennan2018reducibility} and the reduction $\textsc{Gaussianize}$ for mapping from Bernoulli to Gaussian planted submatrix problems.}
\label{fig:rej-kernel}
\end{figure}

\begin{lemma}[Lemma 5.4 in \cite{brennan2018reducibility}] \label{lem:5c}
Let $n$ be a parameter and suppose that $p = p(n)$ and $q = q(n)$ satisfy that $0 < q < p \le 1$, $\min(q, 1 - q) = \Omega(1)$ and $p - q \ge n^{-O(1)}$. Let $\delta = \min \left\{ \log \left( \frac{p}{q} \right), \log \left( \frac{1 - q}{1 - p} \right) \right\}$. Suppose that $\mu = \mu(n) \in (0, 1)$ satisfies that
$$\mu \le \frac{\delta}{2 \sqrt{6\log n + 2\log (p-q)^{-1}}}$$
Then the map $\textsc{rk}_{\text{G}}$ with $N = \left\lceil 6\delta^{-1} \log n \right\rceil$ iterations can be computed in $\text{poly}(n)$ time and satisfies
$$\TV\left(\textsc{rk}_{\text{G}}(\mu, \textnormal{Bern}(p)), \mN(\mu, 1) \right) = O(n^{-3}) \quad \text{and} \quad \TV\left(\textsc{rk}_{\text{G}}(\mu, \textnormal{Bern}(q)), \mN(0, 1) \right) = O(n^{-3})$$
\end{lemma}

In contrast with \cite{brennan2018reducibility}, we apply $\textsc{rk}_G$ when $\mu$ is chosen to be random. We will always apply the lemma conditioned on the value of $\mu$ and hence only require it for deterministic $\mu$. We remark that, throughout the paper, we will use the notation $\textsc{rk}_G(B)$ to denote the random variable output by a run of the procedure in \ref{fig:rej-kernel} using independently generated randomness. The proof of the lemma consists of showing that the outputs of $\pr{rk}_G(\text{Bern}(p))$ and $\pr{rk}_G(\text{Bern}(q))$ are close to $\mN(\mu, 1)$ and $\mN(0, 1)$ conditioned to lie in the set of $x$ with $\frac{1 - p}{1 - q} \le \frac{\varphi_\mu(x)}{\varphi_0(x)} \le \frac{p}{q}$ and then showing that this event occurs with probability close to one. We now present $\textsc{Gaussianize}$, shown in Figure \ref{fig:rej-kernel}, which maps a planted submatrix problem with Bernoulli entries to one with Gaussian entries. This reduction allows for inhomogeneous means $\mu_{ij}$ in the planted component.

\begin{lemma}[Gaussianization] \label{lem:gaussianize}
Given a parameter $N$, let $0 < Q < P \le 1$ be such that $P - Q = N^{-O(1)}$ and $\min(Q, 1 - Q) = \Omega(1)$, let $\mu_{ij}$ be such that $0 \le \mu_{ij} \le \tau$ for each $i, j \in [N]$ where the parameter $\tau > 0$ satisfies that
$$\tau \le \frac{\delta}{2 \sqrt{6\log N + 2\log (P - Q)^{-1}}} \quad \text{where} \quad \delta = \min \left\{ \log \left( \frac{P}{Q} \right), \log \left( \frac{1 - Q}{1 - P} \right) \right\}$$
The algorithm $\mathcal{A} = \textsc{Gaussianize}$ runs in $\textnormal{poly}(N)$ time and satisfies that
\begin{align*}
\TV\left( \mathcal{A}(\mathcal{M}(n, S, P, Q)), \, \mu \circ \mathbf{1}_S \mathbf{1}_S^\top + \mN(0, 1)^{\otimes N \times N} \right) &= O(N^{-1}) \\
\TV\left( \mathcal{A}\left(\textnormal{Bern}(Q)^{\otimes N \times N}\right), \, \mN(0, 1)^{\otimes N \times N} \right) &= O(N^{-1})
\end{align*}
for all subsets $S \subseteq [N]$ where $\circ$ denote the entrywise Hadamard product between two matrices.
\end{lemma}

\begin{proof}
Fix a subset $S \subseteq [N]$ and let $X = \mathcal{A}(M)$ where $M \sim \mathcal{M}(n, S, P, Q))$. Applying Lemma \ref{lem:5c} yields that
\begin{align*}
\TV\left(\textsc{rk}_{\text{G}}(\mu_{ij}, \textnormal{Bern}(P)), \mN(\mu_{ij}, 1) \right) &= O(N^{-3}) \\
\TV\left(\textsc{rk}_{\text{G}}(\mu_{ij}, \textnormal{Bern}(Q)), \mN(0, 1) \right) &= O(N^{-3})
\end{align*}
Since $X$ has independent entries, the tensorization property of Fact \ref{tvfacts} implies that
\begin{align*}
\TV\left( \mL(X), \, \mu \circ \mathbf{1}_S \mathbf{1}_S^\top + \mN(0, 1)^{\otimes N \times N} \right) &\le \sum_{i, j = 1}^N \TV\left( \mL(X_{ij}), \mN(\mu_{ij} \mathbf{1}_{i, j \in S}, 1) ) \right) \\
&= \sum_{(i, j) \in S^2} \TV\left(\textsc{rk}_{\text{G}}(\mu_{ij}, \textnormal{Bern}(P)), \mN(\mu_{ij}, 1) \right) \\
&\quad \quad + \TV\left(\textsc{rk}_{\text{G}}(\mu_{ij}, \textnormal{Bern}(Q)), \mN(0, 1) \right) \\
&= O(N^{-1})
\end{align*}
which proves the first inequality in the statement of the lemma. The second inequality follows by the same argument with $S = \emptyset$.
\end{proof}

\subsection{$\chi^2$ Random Rotations}

The next lemma provides the guarantees for $\chi^2$-\textsc{Random-Rotations}, which will be used as a crucial subroutine to produce the cross-term matrices in the main reduction to sparse PCA in Section \ref{sec:wishartreduction} and in the subsampling reduction to sparse PCA in Section \ref{sec:randomrotations}. The analysis of this reduction is similar to $\textsc{Random-Rotations}$ in \cite{brennan2018reducibility} but avoids using the Gaussian variant of finite de Finetti by planting entries with $\chi^2$-distributed means. Let $v_S = \frac{1}{\sqrt{k}} \cdot \mathbf{1}_S \in \mathbb{R}^m$ for each $S \subseteq [m]$ of size $k$.

\begin{figure}[t!]
\begin{algbox}
\textbf{Algorithm} $\chi^2\textsc{-Random-Rotation}$

\vspace{1mm}

\textit{Inputs}: Matrix $M \in \{0, 1\}^{N \times N}$, Bernoulli probabilities $0 < Q < P \le 1$, planted subset size $k$ that divides $N$ and a parameter $\tau > 0$
\begin{enumerate}
\item Sample $r_1, r_2, \dots, r_N \sim_{\text{i.i.d.}} \sqrt{\chi^2(N/k)}$ and truncate the $r_j$ with $r_j \gets \min\left\{ r_j, 2 \sqrt{N/k} \right\}$ for each $j \in [N]$
\item Update $M \gets \textsc{Gaussianize}(M)$ applied with Bernoulli probabilities $P$ and $Q$, parameter $\tau$ and target mean values $\mu_{ij} = \frac{1}{2} \tau \cdot r_j \cdot \sqrt{k/N}$ for each $i, j \in [N]$
\item Sample an orthogonal matrix $R \in \mathbb{R}^{N\times N}$ from the Haar measure on the orthogonal group $\mathcal{O}_N$ and output the columns of the matrix $MR$
\end{enumerate}
\vspace{1mm}

\end{algbox}
\caption{Subroutine $\chi^2\textsc{-Random-Rotation}$ for random rotations to instances of Sparse PCA.}
\label{fig:chisrr}
\end{figure}

\begin{lemma}[$\chi^2$ Random Rotations] \label{lem:randomrotations}
Given a parameter $N$, let $0 < Q < P \le 1$ be such that $P - Q = N^{-O(1)}$ and $\min(Q, 1 - Q) = \Omega(1)$, let $k \le N$ be such that $k$ divides $N$ and let $\tau > 0$ be such that
$$\tau \le \frac{\delta}{2 \sqrt{6\log N + 2\log (P - Q)^{-1}}} \quad \text{where} \quad \delta = \min \left\{ \log \left( \frac{P}{Q} \right), \log \left( \frac{1 - Q}{1 - P} \right) \right\}$$
The algorithm $\mathcal{A} = \chi^2\textsc{-Random-Rotation}$ runs in $\textnormal{poly}(N)$ time and satisfies that
\begin{align*}
\TV\left( \mathcal{A}(\mathcal{M}(n, S, P, Q)), \, \mN\left(0, I_N + \frac{\tau^2 k^2}{4N} \cdot v_S v_S^\top \right)^{\otimes N} \right) &\le O(N^{-1}) + k(4e^{-3})^{N/2k} \\
\TV\left( \mathcal{A}\left( \textnormal{Bern}(Q)^{\otimes N \times N}\right), \, \mN(0, 1)^{\otimes N \times N} \right) &= O(N^{-1})
\end{align*}
for all subsets $S \subseteq [N]$ with $|S| = k$.
\end{lemma}

\begin{proof}
We begin by conditioning on all of $r = (r_1, r_2, \dots, r_N)$ generated in the first two steps $\mathcal{A}_{\text{1-2}}$ of $\mathcal{A}$. Almost surely for each $i \in [N]$, we have that $|r_i| \le 2 \sqrt{N/k}$ and therefore
$$\mu_{ij} = \frac{1}{2} \tau \cdot r_j \cdot \sqrt{k/N} \le \tau$$
for each $i, j \in [N]$. Consider conditioning on the event $r = r' = (r_1', r_2', \dots, r_N') \in [0, \tau]^N$. Let $D$ be the matrix
$$D = \frac{\tau}{2} \sqrt{\frac{k}{N}} \cdot \mathbf{1}_S u_S^\top + \mN(0, 1)^{\otimes N \times N}$$
where $v_S$ denotes the vector with $(u_S)_j = r_j'$ if $j \in S$ and $(u_S)_j = 0$ if $j \not \in S$. Applying Lemma \ref{lem:gaussianize} with $\mu_{ij} = \frac{1}{2} \tau \cdot r_j \cdot \sqrt{k/N}$ yields that
\begin{align*}
\TV\left( \mL(\mathcal{A}_{\text{1-2}}(\mathcal{M}(n, S, P, Q)) | r = r'), \, \mL(D | r = r') \right) &= O(N^{-1}) \\
\TV\left( \mL(\mathcal{A}_{\text{1-2}}\left( \textnormal{Bern}(Q)^{\otimes N \times N} \Big| r = r' \right), \, \mN(0, 1)^{\otimes N \times N} \right) &= O(N^{-1})
\end{align*}
Since the rows of $\mN(0, 1)^{\otimes N \times N}$ are isotropic, for any orthogonal matrix $R$ we have that $\mN(0, 1)^{\otimes N \times N}$ is identically distributed to $XR$ where $X \sim \mN(0, 1)^{\otimes N \times N}$. This implies that $\mL\left( X \right)$ can be written as $\mL\left( X R_2 \right)$ where $R_2 \sim \mu_{\mathcal{O}_N}$ is independent of $X$. Now observe that if $R_1 \sim \mu_{\mathcal{O}_N}$ is independent of $R_2$ then $R_2 R_1 \sim \mu_{\mathcal{O}_N}$ and importantly is independent of $R_1$. Combining these observations yields that if $R_1, R_2 \sim \mu_{\mathcal{O}_N}$ are independently sampled, then
\begin{align*}
\mL\left(D R_1 | r = r' \right) &= \mL\left( \frac{\tau}{2} \sqrt{\frac{k}{N}} \cdot \mathbf{1}_S u_S^\top R_1  + \mN(0, 1)^{\otimes N \times N} \cdot R_2 R_1 \Big| r = r' \right) \\
&= \mL\left(\frac{\tau}{2} \sqrt{\frac{k}{N}} \cdot \| u_S \|_2 \cdot \mathbf{1}_S u'^\top + \mN(0, 1)^{\otimes N \times N} \cdot R_3 \Big| r = r' \right) \\
&= \mL\left(\frac{\tau}{2} \sqrt{\frac{k}{N}} \cdot \| u_S \|_2 \cdot \mathbf{1}_S u'^\top + \mN(0, 1)^{\otimes N \times N} \Big| r = r' \right)
\end{align*}
where $R_3$ is an independent sample from $\mu_{\mathcal{O}_N}$ and $u'$ is a unit vector distributed uniformly on $\mathbb{S}^{N-1}$. Let $Y \sim \chi^2(N/k)$ and observe that for each $i \in [N]$, the $\chi^2$ tail bound in \cite{dasgupta2003elementary} stating $\bP\left[ \chi^2(m) > zm \right] \le (z e^{1 - z})^{m/2}$ for all $z > 1$ together with the conditioning property in Fact \ref{tvfacts} implies that
$$\TV\left( \mL(r_i), \sqrt{\chi^2(N/k)} \right) \le \bP\left[Y > 4N/k \right] \le (4e^{-3})^{N/2k}$$
By the tensorization property and data processing inequality in Fact \ref{tvfacts}, we now have
$$\TV\left( \mL(\| u_S \|_2^2), \chi^2(N) \right) \le \TV\left( \mL(u_S), \sqrt{\chi^2(N/k)}^{\otimes k} \right) \le \sum_{i \in S} \TV\left( \mL(r_i), \sqrt{\chi^2(N/k)} \right) \le k(4e^{-3})^{N/2k}$$
Now note that the isotropy of the distribution $\mN(0, I_N)$ implies that $\mL(w | \| w \|_2 = t)$ is the uniform distribution $\mL( t u')$ on the sphere $t \cdot \mathbb{S}^{N - 1}$ for each $t \ge 0$, where $w \sim \mN(0, I_N)$. Now note that $\| w \|_2$ is distributed as $\sqrt{\chi^2(N)}$ and therefore $\mL(w) = \mL(\sqrt{Z} \cdot u')$ where $Z \sim \chi^2(N)$ is independent of $u'$. We now marginalize over the conditioning on $r = r'$. By the data processing inequality in Fact \ref{tvfacts},
\begin{align*}
&\TV\left( \mL(\mathcal{A}_3(D)), \frac{\tau k}{2\sqrt{N}} \cdot v_S w^\top + \mN(0, 1)^{\otimes N \times N} \right) \\
&\quad \quad = \TV\left( \mL(DR_1), \frac{\tau k}{2\sqrt{N}} \cdot v_S w^\top + \mN(0, 1)^{\otimes N \times N} \right) \\
&\quad \quad = \TV\left(\mL\left( \frac{\tau}{2} \sqrt{\frac{k}{N}} \cdot \| u_S \|_2 \cdot \mathbf{1}_S u'^\top + \mN(0, 1)^{\otimes N \times N} \right), \mL\left( \frac{\tau}{2} \sqrt{\frac{k}{N}} \cdot \sqrt{Z} \cdot \mathbf{1}_S u'^\top + \mN(0, 1)^{\otimes N \times N} \right) \right) \\
&\quad \quad \le \TV\left( \mL(\| u_S \|_2^2), \mL(Z) \right) \le k(4e^{-3})^{N/2k}
\end{align*}
Now by Lemma \ref{lem:tvacc} and the conditioning property in Fact \ref{tvfacts}, we have
\begin{align*}
&\TV\left( \mathcal{A}(\mathcal{M}(n, S, P, Q), \frac{\tau k}{2\sqrt{N}} \cdot v_S w^\top + \mN(0, 1)^{\otimes N \times N} \right) \\
&\quad \quad \le \TV\left( \mathcal{A}_{\text{1-2}}(\mathcal{M}(n, S, P, Q)), \mL(D) \right) +  \TV\left( \mathcal{A}_3(D), \frac{\tau k}{2\sqrt{N}} \cdot v_S w^\top + \mN(0, 1)^{\otimes N \times N} \right) \\
&\quad \quad \le \bE_{r' \sim \mL(r)} \left[ \TV\left( \mL(\mathcal{A}_{\text{1-2}}(\mathcal{M}(n, S, P, Q)) | r = r'), \mL(D | r = r') \right) \right] + k(4e^{-3})^{N/2k} \\
&\quad \quad \le O(N^{-1}) + k(4e^{-3})^{N/2k}
\end{align*}
Now note that the columns of $X = \frac{\tau k}{2\sqrt{N}} \cdot v_S w^\top + \mN(0, 1)^{\otimes N \times N}$ are i.i.d. and have zero-mean jointly Gaussian entries. Since the two terms in $X$ are independent, the covariance matrix of the entries of a column is
$$\bE\left[X_1 X_1^\top\right] = I_N + \frac{\tau^2 k^2}{4N} \cdot v_S v_S^\top$$
which completes the proof of the first inequality in the proposition. Now observe that
$$\mL\left(\mathcal{A}_3\left( \mN(0, 1)^{\otimes N \times N} \right) \right) = \mL\left(\mN(0, 1)^{\otimes N \times N} \cdot R \right) = \mN(0, 1)^{\otimes N \times N}$$
since the rows of $\mN(0, 1)^{\otimes N \times N}$ are isotropic. The data-processing and conditioning properties in Fact \ref{tvfacts} now imply that
\begin{align*}
&\TV\left( \mathcal{A}\left( \textnormal{Bern}(Q)^{\otimes N \times N}\right), \, \mN(0, 1)^{\otimes N \times N} \right) \\
&\quad \quad \le \TV\left( \mathcal{A}_{\text{1-2}}\left( \textnormal{Bern}(Q)^{\otimes N \times N}\right), \mN(0, 1)^{\otimes N \times N} \right) \\
&\quad \quad \le \bE_{r' \sim \mL(r)} \left[\TV\left( \mL\left(\mathcal{A}_{\text{1-2}}\left( \textnormal{Bern}(Q)^{\otimes N \times N}\right) \Big| r = r'\right), \mN(0, 1)^{\otimes N \times N} \right) \right] \\
&\quad \quad = O(N^{-1})
\end{align*}
which completes the proof of the proposition.
\end{proof}

\section{Wishart Sampling and Strong Lower Bounds for Sparse PCA}
\label{sec:wishartreduction}

\begin{figure}[t!]
\begin{algbox}
\textbf{Algorithm} \textsc{Clique-to-Wishart}

\vspace{2mm}

\textit{Inputs}: Graph $G \in \mG_N$ with clique size $k$, edge probabilities $0 < q < p \le 1$, target $\pr{ubspca}$ parameters $\theta \lesssim \sqrt{\frac{k^{2}}{n \log N}}$, $d \ge \left( \frac{p}{Q'} + \epsilon \right) N + k$ where $Q' = 1 - \sqrt{(1 - p)(1 - q)} + \mathbf{1}_{\{p = 1\}} \left( \sqrt{q} - 1 \right)$ and $n = \omega(N^3)$ where $d$ and $n$ are both bounded above by some $\text{poly}(N)$

\vspace{2mm}

\begin{enumerate}
\item \textit{Plant Diagonals}: Compute $X \in \{0, 1\}^{m \times m}$ as $X \gets \textsc{To-Bernoulli-Submatrix}(G)$ applied with initial dimension $N$, edge probabilities $p$ and $q$ and target dimension $m$ where $m$ is the smallest multiple of $k$ larger than $\left( \frac{p}{Q'} + \epsilon \right) N$
\item \textit{Clone}: Compute $(X_1, X_2, X_3) \gets \textsc{Bernoulli-Matrix-Clone}(X)$ applied with $t = 3$, starting probabilities $p$ and $Q'$ and $P = p$ and $Q = 1 - (1 - p)^{5/6}(1 - q)^{1/6} + \mathbf{1}_{\{p = 1\}} \left( q^{1/6} - 1 \right)$
\item \textit{Generate Mean Matrix:} Compute $M \in \mathbb{R}^{m \times m}$ as follows
\begin{enumerate}
\item[(1)] Sample $g \sim \N(0, 1)$ and if $|g| \ge 4 \sqrt{\log m}$ then set $g \gets 0$
\item[(2)] Compute the matrix $M$ as
$$M \gets \textsc{Gaussianize}(X_1)$$
applied with dimension $m$, Bernoulli probabilities $P$ and $Q$ and mean parameters $\mu_{ij} = \frac{\theta \sqrt{3}}{k} \left( \sqrt{n/2} + g \right)$ for all $i, j \in [m]$
\end{enumerate}
\item \textit{Generate Cross-Term Matrices:} Compute $C_L, C_R \in \mathbb{R}^{m \times m}$ as
\begin{align*}
C_L^\top &\gets \chi^2\textsc{-Random-Rotation}(X_2) \\
C_R &\gets \chi^2\textsc{-Random-Rotation}(X_3)
\end{align*}
each applied with dimension $m$, Bernoulli probabiities $P$ and $Q$ and $\tau = \frac{2 \sqrt{3 m\theta}}{k}$
\item \textit{Form Scaled Empirical Covariance Matrix}: Compute the matrix $\Sigma_{e} \in \mathbb{R}^{m \times m}$ as
$$\Sigma_e = n \cdot I_m + \sqrt{\frac{n}{6}} \left( M + M^\top + C_L + C_L^\top + C_R + C_R^\top \right)$$
\item\textit{Inverse Wishart Sampling:} Form $Y \in \mathbb{R}^{m \times n}$ by setting $R_{\Sigma} \in \mathbb{R}^{m \times n}$ to be the top $m$ rows of a $n \times n$ orthogonal matrix sampled from the Haar measure on $\mathcal{O}_n$ and then computing $Y = \Sigma_e^{1/2} R_\Sigma$ if $\Sigma_e \succeq 0$ and $Y = 0$ otherwise
\item \textit{Pad and Output}: Form the matrix $Z \in \mathbb{R}^{d \times n}$ by randomly embedding the rows of $Y$ as rows of $Z$ and sampling the remaining $d - m$ rows of $Z$ i.i.d. from $\mN(0, I_n)$ and then output the columns $Z_1, Z_2, \dots, Z_n$ of $Z$
\end{enumerate}
\vspace{1mm}
\end{algbox}
\caption{Wishart sampling reduction showing strong lower bounds for $\pr{ubspca}$.}
\label{fig:wishart}
\end{figure}

With the subroutines introduced in the previous section, we now present our Wishart sampling reduction $\textsc{Clique-to-Wishart}$ showing strong lower bounds for sparse PCA. The algorithm $\textsc{Clique-to-Wishart}$ is described in Figure \ref{fig:wishart} and its total variation guarantees are stated in Theorem \ref{thm:wishartsamp} below. %A diagram illustrating the steps of the reduction is shown in Figure \ref{fig:wishartdiagram}.

The reduction $\textsc{Clique-to-Wishart}$ can be seen to run in randomized $\text{poly}(N)$ time as follows. The subroutines introduced in the previous section each have been shown to run in $\text{poly}(N)$ time for their given input parameters. The square root of a symmetric real positive semidefinite matrix in Step 6 can be computed in polynomial time by diagonalizing with Jacobi's algorithm and then compute the square root of the resulting diagonalized matrix. As shown in Proposition 7.2 of \cite{eaton1983multivariate}, a sample from the Haar measure $\mu_{\mathcal{O}_n}$ on the orthogonal group $\mathcal{O}_n$ can be generated by applying Gram-Schmidt orthogonalization to the rows of a sample from $\mN(0, 1)^{\otimes n \times n}$. Therefore a sample from $\mu_{\mathcal{O}_n}$ can be generated in $\text{poly}(n)$ time.

\begin{theorem}[Wishart Sampling Reduction] \label{thm:wishartsamp}
Fix constants $\epsilon, c > 0$ and $0 < q < p < 1$. Let $k, N, d, n$ and $\theta > 0$ be parameters satisfying that $d \ge \left( \frac{p}{Q'} + \epsilon \right)N + k$ where $Q' = 1 - \sqrt{(1 - p)(1 - q)} + \mathbf{1}_{\{p = 1\}} \left( \sqrt{q} - 1 \right)$,
$$\theta \le c \cdot \sqrt{\frac{k^2}{n \log N}}$$
and both $n = \omega(N^3)$ and $N = \omega(k^2)$ hold as $k \to \infty$. Let $\mathcal{A}(G)$ denote $\textsc{Clique-to-Wishart}$ applied with these input parameters and input graph $G$. If $n$ and $d$ are both bounded above by some $\textnormal{poly}(N)$, then $\mathcal{A}$ runs in randomized polynomial time and there is a sufficiently small constant $c > 0$ such that
$$\TV\left( \mathcal{A}\left( \pr{pds}(N, k, p, q) \right), \, \pr{ubspca}(n, k, d, \theta) \right) \to 0$$
under both $H_0$ and $H_1$ as $k \to \infty$.
\end{theorem}

%\begin{figure*}[t!]
%\centering
%\begin{tikzpicture}[scale=0.33]
%\tikzstyle{every node}=[font=\footnotesize]
%\draw (0, 3) -- (0, -3) -- (6, -3) -- (6, 3) -- (0, 3);
%\draw [fill=red!20] (0, -3) -- (2, -3) -- (2, 0) -- (0, 0) -- (0, -3);
%\node at (3, 2) {$2 A(G) - J$};
%\node at (-1.1, -1.5) {$N/2$};
%\node at (1, 0.6) {$n$};
%\draw[->] (2, -1.5) -- (12, 0);
%\node at (9, 1.6) {extract};
%\node at (9, 0.8) {lower left};
%\draw [fill=red!20] (13, 1.5) -- (15, 1.5) -- (15, -1.5) -- (13, -1.5) -- (13, 1.5);
%\draw[->] (16, 0) -- (20, 0);
%\node at (18, 1.6) {pad and};
%\node at (18, 0.8) {sign};
%\draw [fill=red!20] (21, 2.5) -- (23, 2.5) -- (23, -0.5) -- (21, -0.5) -- (21, 2.5);
%\draw [fill=red!50] (21, -0.5) -- (23, -0.5) -- (23, -2.5) -- (21, -2.5) -- (21, -0.5);
%\node at (22, 3) {$n$};
%\node at (23.6, 0) {$d$};
%\end{tikzpicture}
%\caption{Visual depiction of the reduction $\textsc{Clique-to-Wishart}$ applied under $H_1$.}
%\label{fig:wishartdiagram}
%\end{figure*}

Assuming the $\pr{pc}$ conjecture only up to $k = O(N^{\alpha})$ for some $\alpha \in (0, 1/2)$, this reduction shows that there is no randomized polynomial time algorithm solving $\pr{ubspca}(n, k, d, \theta)$ up to the conjectured boundary of $\theta = \tilde{o}(\sqrt{k^2/n})$ for sparsities satisfying $k = O(n^{\alpha/3})$. More formally, Lemma \ref{lem:3a} implies the following corollary of Theorem \ref{thm:wishartsamp} giving our main lower bounds for sparse PCA.

\begin{corollary}[Strong Lower Bounds for Sparse PCA] \label{thm:lowerbounds}
Fix some $\alpha \in (0, 1/2)$ and $c > 0$. Let $(n_t, k_t, d_t, \theta_t)$ be a sequence of parameters with $n_t \ge c^{-1}k_t^{3/\alpha}$, $d_t \ge c^{-1} k_t^{1/\alpha}$ and
$$\theta_t \le c \cdot \sqrt{\frac{k_t^{2}}{n _t \log n_t}}$$
There is a sufficiently small constant $c$ such that for any sequence of randomized polynomial time tests $\mathcal{A}_t$, the asymptotic Type I$+$II error of $\mathcal{A}_t$ on $\pr{ubspca}(n_t, k_t, d_t, \theta_t)$ is at least $1$ assuming the $\textsc{pds}$ conjecture for $k = O(n^\alpha)$ and some pair of constant densities $0 < q < p \le 1$.
\end{corollary}

The next two sections are devoted to proving Theorem \ref{thm:wishartsamp}. In Section \ref{sec:Wishart}, we establish the decomposition theorems for planted Wishart matrices motivating the reduction. In Section \ref{sec:proofmain}, we complete the proof of Theorem \ref{thm:wishartsamp}.

\section{Decompositions of Wishart Matrices}
\label{sec:Wishart}

We now establish the decomposition and comparison theorems on random matrices that will be used in the proof of Theorem \ref{thm:wishartsamp}. A key ingredient in the analysis of $\textsc{Clique-to-Wishart}$ will be a a theorem of \cite{bubeck2016testing}, showing that in high-dimensional setting the fluctuations of the entries of a sample from the isotropic Wishart distribution resemble independent Gaussians in total variation. From this theorem, we deduce that in the same high-dimensional setting, the entries of spiked Wishart distributions are close in total variation to being jointly Gaussian. We then give a representation of this jointly Gaussian matrix as a sum of four simple independent Gaussian matrices, which is crucial to the construction of $\Sigma_e$ in $\textsc{Clique-to-Wishart}$. 

We now define the Wishart and $\textsc{goe}$ matrix distributions, which appear throughout our analysis of $\textsc{Clique-to-Wishart}$.

\begin{definition}[Wishart Matrix]
Let $\Sigma$ be a $d\times d$ positive semidefinite matrix. Then let $\mW_n(\Sigma, d)$ denote the distribution of $\sum_{i = 1}^n X_i X_i^\top$ where $X_1, X_2, \dots, X_n$ are i.i.d. samples from $\mN(0, \Sigma)$.
\end{definition}

\begin{definition}[Gaussian Orthogonal Ensemble]
Let $\textsc{goe}(d)$ denote the distribution on $d \times d$ symmetric real matrices with diagonal entries sampled i.i.d. from $\mN(0, 2)$ and above-diagonal entries sampled i.i.d. from $\mN(0, 1)$.
\end{definition}

The next theorem is proven in \cite{bubeck2016testing} through a direct comparison of probability density functions, through an entropic central limit theorem in \cite{bubeck2016entropic} and also in \cite{jiang2015approximation}. 

\begin{theorem}[Theorem 7 in \cite{bubeck2016testing}] \label{thm:bubeck}
If $n/d^3 \to \infty$ as $n, d \to \infty$, then
$$\TV\left( \mW_n(I_d, d), \, n \cdot I_d + \sqrt{n} \cdot \textsc{goe}(d) \right) \to 0$$
\end{theorem}

Using this theorem, we obtain Lemma \ref{lem:plantedwishart} giving a jointly Gaussian approximation of the sample covariance matrix of $\pr{ubspca}$ under $H_1$. A related theorem is proven in \cite{nourdin2018asymptotic} with techniques from Malliavin calculus. We first show the following simple fact which will simplify the computations in Lemma \ref{lem:plantedwishart}.

\begin{fact} \label{lem:gausscorr}
If $x, y, z, w \in \mathbb{R}^d$ and $G \sim \mN(0, 1)^{\otimes d \times d}$, then
$$\bE\left[\left( x^\top G y \right) \left( z^\top G w \right) \right] = \langle x, z \rangle \cdot \langle y, w \rangle$$
\end{fact}

\begin{proof}
Since $\bE[G_{ij} G_{kl}] = \mathbf{1}_{\{(i, j) = (k, l)\}}$, we have that
\begin{align*}
\bE\left[\left( x^\top G y \right) \left( z^\top G w \right) \right] &= \bE\left[\left( \sum_{i, j = 1}^n G_{ij} x_i y_j \right) \left( \sum_{k, l = 1}^n G_{ij} x_k y_l \right) \right] = \sum_{i, j, k, l = 1}^n \bE\left[ G_{ij}G_{kl} \right] \cdot x_i y_j z_k w_l \\
&= \sum_{i, j = 1}^n x_i y_j z_i w_j = \langle x, z \rangle \cdot \langle y, w \rangle
\end{align*}
which completes the proof of the lemma.
\end{proof}

\begin{lemma} \label{lem:plantedwishart}
Given a $S \subseteq [d]$ of size $|S| = k$, let $v_S = k^{-1/2} \cdot \mathbf{1}_S$ be the uniform unit vector supported on $S$. Let $\textsc{gw}_d(n, S, \theta)$ denote the distribution of $d \times d$ real matrices $X$ such that the entries of $X$ are zero-mean and jointly Gaussian with covariances satisfying 
$$\frac{1}{n} \cdot \bE\left[X_{ij} X_{kl}\right] = \left( \mathbf{1}_{\{i = k\}} + \frac{\theta}{k} \cdot \mathbf{1}_{\{i, k \in S\}} \right) \left( \mathbf{1}_{\{j = l\}} + \frac{\theta}{k} \cdot \mathbf{1}_{\{j, l \in S\}} \right)$$
for all $i, j, k, l \in [d]$. Let $\textsc{sgw}_d(n, S, \theta)$ denote the distribution of $\frac{1}{\sqrt{2}} \left( X + X^\top \right)$ where $X \sim \textsc{gw}_d(n, S, \theta)$. If $n/d^3 \to \infty$ as $n, d \to \infty$, then for any $\theta = \theta(n) > -1$ we have
$$\max_S \TV\left( \mW_n(I_d + \theta v_S v_S^\top, d), \, n \left( I_d + \theta v_S v_S^\top \right) + \textsc{sgw}_d(n, S, \theta) \right) \to 0$$
where the maximum is over all $S \subseteq [d]$ with $|S| = k$.
\end{lemma}

\begin{proof}
Let $\Sigma = I_d + \theta v_S v_S^\top$ and let $X_1, X_2, \dots, X_n$ be i.i.d. samples from $\mN(0, I_d)$. Note that the entries of $Y_i = \Sigma^{1/2} X_i$ are jointly Gaussian and have covariance matrix $\bE[Y_iY_i^\top ] = \Sigma^{1/2} \cdot \bE[X_iX_i^\top ] \cdot \Sigma^{1/2} = \Sigma$ since $\bE[X_iX_i^\top] = I_d$. It follows that $Y_1, Y_2, \dots, Y_n$ are i.i.d. samples from $\mN(0, \Sigma)$. Therefore it follows that
$$W_\Sigma = \sum_{i = 1}^n Y_i Y_i^\top = \Sigma^{1/2} \left( \sum_{i = 1}^n X_i X_i^\top \right) \Sigma^{1/2} = \Sigma^{1/2} W_{I_d} \Sigma^{1/2}$$
is distributed as $\mW_n(\Sigma, d)$ where $W_{I_d}$ is distributed as $\mW_n(I_d, d)$. Let $G \sim \mN(0, 1)^{\otimes d \times d}$ and note that $\frac{1}{\sqrt{2}} \left( G + G^\top \right)$ is distributed as $\textsc{goe}(d)$. Let $X = \sqrt{n} \cdot \Sigma^{1/2} G \Sigma^{1/2}$ and let
$$W_\Sigma^G = \Sigma^{1/2} \left( n \cdot I_d + \sqrt{\frac{n}{2}} \cdot \left( G + G^\top \right) \right) \Sigma^{1/2} = n \cdot \Sigma + \frac{1}{\sqrt{2}} \left( X + X^\top \right)$$
The data-processing inequality now implies that
\begin{align*}
\TV\left( \mW_n(I_d + \theta v_S v_S^\top, d), \mL\left( W_\Sigma^G \right) \right) &= \TV\left( \mL\left( \Sigma^{1/2} W_{I_d} \Sigma^{1/2} \right), \mL\left( \Sigma^{1/2} \left( n \cdot I_d + \sqrt{n} \cdot \textsc{goe}(d) \right) \Sigma^{1/2} \right) \right) \\
&\le \TV\left( \mW_n(I_d, d), n \cdot I_d + \sqrt{n} \cdot \textsc{goe}(d) \right)
\end{align*}
Since $G$ is a matrix of independent Gaussians and $X = \sqrt{n} \cdot \Sigma^{1/2} G \Sigma^{1/2}$, the entries of $X$ are jointly Gaussian. Hence showing that the covariances between the entries of $X$ are as defined above and applying Theorem \ref{thm:bubeck} would complete the proof of the lemma. Let $e_i \in \mathbb{R}^d$ be the unit vector with $i$th coordinate equal to one and all other coordinates equal to zero. This implies
$$\frac{1}{\sqrt{n}} \cdot X_{ij} = \frac{1}{\sqrt{n}} \cdot e_i^\top X e_j = e_i^\top \Sigma^{1/2} G \Sigma^{1/2} e_j$$
By Fact \ref{lem:gausscorr} and since $\Sigma^{1/2}$ is symmetric, we have that
\begin{align*}
\frac{1}{n} \cdot \bE\left[X_{ij} X_{kl}\right] &= \bE\left[ \left( e_i^\top \Sigma^{1/2} G \Sigma^{1/2} e_j \right) \left( e_k^\top \Sigma^{1/2} G \Sigma^{1/2} e_l \right) \right] \\
&= \left\langle \Sigma^{1/2} e_i, \Sigma^{1/2} e_k \right\rangle \cdot \left\langle \Sigma^{1/2} e_j, \Sigma^{1/2} e_l \right\rangle \\
&= \left( e_i^\top \Sigma e_k \right) \cdot \left( e_j^\top \Sigma e_l \right) = \Sigma_{ik} \Sigma_{jl} \\
&= \left( \mathbf{1}_{\{i = k\}} + \frac{\theta}{k} \cdot \mathbf{1}_{\{i, k \in S\}} \right) \left( \mathbf{1}_{\{j = l\}} + \frac{\theta}{k} \cdot \mathbf{1}_{\{j, l \in S\}} \right)
\end{align*}
since $\Sigma$ has $(i, j)$th entry $\Sigma_{ij} = \mathbf{1}_{\{i = j\}} + \frac{\theta}{k} \cdot \mathbf{1}_{\{i, j \in S\}}$, proving the lemma.
\end{proof}

The next lemma gives a method of representing matrices $X \sim \textsc{gw}_d(n, S, \tau)$ as a sum of four independent Gaussian terms. This is crucially used in our reduction to express $\Sigma_e$ as the sum of independent Gaussian matrices.

\begin{lemma} \label{lem:indrep}
Suppose $W$ is distributed as $\mN(0, 1)^{\otimes d \times d}$, $w_1, w_2$ are distributed as $\mN(0, I_d)$ and $g$ is distributed as $\mN(0, 1)$ and that $W, w_1, w_2$ and $g$ are independent. Then for all $n, \theta > 0$ and $S \subseteq [d]$,
$$X = \sqrt{n} \cdot \left( W + \sqrt{\theta} \cdot v_S w_1^\top + \sqrt{\theta} \cdot w_2 v_S^\top + \theta g v_Sv_S^\top \right)$$
is distributed as $\textsc{gw}_d(n, S, \theta)$.
\end{lemma}

\begin{proof}
Note that the entries of $X$ are jointly Gaussian. Since the same is true of $\textsc{gw}_d(n, S, \theta)$, it suffices to verify that the covariances between the entries of $X$ are as computed in the previous lemma. Let $G^{(1)} = \sqrt{\theta} \cdot v_S w_1^\top$ and $G^{(2)} = \sqrt{\theta} \cdot w_2 v_S^\top$ and $G^{(3)} = \theta g v_Sv_S^\top$. Now observe that the covariances between the entries of $G^{(1)}$ and $G^{(2)}$ are
$$\frac{1}{n} \cdot \bE\left[ G^{(1)}_{ij} G^{(1)}_{kl} \right] = \frac{1}{n} \cdot \bE\left[ G^{(2)}_{ji} G^{(2)}_{lk} \right] = \left\{ \begin{matrix}
\frac{\theta}{k} & \textnormal{if } i, k \in S \textnormal{ and } j = l \\
0 & \textnormal{otherwise} \end{matrix} \right.$$
Furthermore, the covariances between the entries of $G^{(3)}$ are
$$\frac{1}{n} \cdot \bE\left[ G^{(3)}_{ij} G^{(3)}_{kl} \right] = \left\{ \begin{matrix}
\frac{\theta^2}{k^2} & \textnormal{if } i, j, k, l \in S \\
0 & \textnormal{otherwise} \end{matrix} \right.$$
Any consistent vectorization of the indices of $X, W$ and the $G^{(i)}$ yields that the $d^2 \times d^2$ covariance matrix of the entries of $n^{-1/2} X$ is the sum of the covariance matrices of $W$ and the $G^{(i)}$ since these four matrices are independent. Adding the covariances above and $\bE[W_{ij} W_{kl}] = \mathbf{1}_{\{(i, j) = (k, l)\}}$ yields that the entries of $X$ have the covariance matrix
\begin{align*}
\frac{1}{n} \cdot \bE\left[X_{ij} X_{kl}\right] &= \mathbf{1}_{\{(i, j) = (k, l)\}} + \frac{\theta}{k} \cdot \mathbf{1}_{\{i, k \in S, j = l\}} + \frac{\theta}{k} \cdot \mathbf{1}_{\{i = k, j, l \in S\}} + \frac{\theta^2}{k^2} \cdot \mathbf{1}_{\{i, j, k, l \in S\}} \\
&= \left( \mathbf{1}_{\{i = k\}} + \frac{\theta}{k} \cdot \mathbf{1}_{\{i, k \in S\}} \right) \left( \mathbf{1}_{\{j = l\}} + \frac{\theta}{k} \cdot \mathbf{1}_{\{j, l \in S\}} \right)
\end{align*}
as computed for $\textsc{gw}_d(n, S, \tau)$ in the previous lemma. This proves the desired result.
\end{proof}

We now give the main ingredients showing the correctness of the inverse Wishart sampling step of $\textsc{Clique-to-Wishart}$. If $n \ge d$, let $\mathcal{F}_{d, n}$ denote the set of $n \times d$ matrices $X$ with $X^\top X = I_d$. We now define the notion of left invariance of a measure on $\mathbb{R}^{n \times d}$ under the action of the orthogonal group $\mathcal{O}_n$.

\begin{definition}[Left Invariance Under $\mathcal{O}_n$]
A measure $\mu$ on $\mathbb{R}^{n \times d}$ is left invariant under $\mathcal{O}_n$ if $\mL(X) = \mL(RX)$ for all $R \in \mathcal{O}_n$ where $X \sim \mu$.
\end{definition}

As shown in \cite{eaton1983multivariate} in Example 6.16, there is a unique probability measure $\mu_{\mathcal{F}_{d, n}}$ on $\mathcal{F}_{d, n}$ that is left invariant under $\mathcal{O}_n$. In Proposition 7.2 of \cite{eaton1983multivariate}, it is shown that an $X \sim \mu_{\mathcal{F}_{d, n}}$ can be generated by applying Gram-Schmidt orthogonalization to the columns of a matrix distributed as $\mN(0, 1)^{\otimes d \times n}$, or equivalently by taking the first $d$ columns of a random matrix sampled from the Haar measure $\mu_{\mathcal{O}_n}$ on the orthogonal group $\mathcal{O}_n$.

Now note that given some $X \in \mathbb{R}^{n \times d}$, there is a unique way to write $X = R \Sigma$ where $R \in \mathcal{F}_{d, n}$ and $\Sigma \in \mathbb{R}^{d \times d}$ is positive semidefinite. In particular, $\Sigma = (X^\top X)^{1/2}$ and $R = X (X^\top X)^{1/2}$ are the unique $R$ and $\Sigma$. The next proposition from \cite{eaton1983multivariate} characterizes random matrices that are left invariant under $\mathcal{O}_n$. 

\begin{theorem}[Proposition 7.4 in \cite{eaton1983multivariate}] \label{thm:invwishart}
Let $X$ be a random matrix in $\mathbb{R}^{n \times d}$ with a law left invariant under $\mathcal{O}_n$. Let $R$ and $\Sigma$ be the unique $R \in \mathcal{F}_{d, n}$ and positive semidefinite $\Sigma \in \mathbb{R}^{d \times d}$ with $X = R \Sigma$. Then $R \sim \mu_{\mathcal{F}_{d, n}}$ and $R$ and $\Sigma$ are independently distributed. Conversely, if $R \sim \mu_{\mathcal{F}_{d, n}}$ and a random positive semidefinite matrix $\Sigma \in \mathbb{R}^{d \times d}$ are independent then $X = R\Sigma$ has a distribution that is left invariant under $\mathcal{O}_n$.
\end{theorem}

As a simple corollary of this theorem, we obtain the following lemma showing the correctness of Step 5 in $\textsc{Clique-to-Wishart}$.

\begin{lemma}[Inverse Wishart Sampling] \label{lem:invwishart}
Suppose that $n \ge d$ and let $\Sigma \in \mathbb{R}^{d \times d}$ be a fixed positive definite matrix and let $\Sigma_e \sim \mathcal{W}_n(\Sigma, d)$. Let $R \in \mathbb{R}^{d \times n}$ be the matrix consisting of the first $d$ rows of an $n \times n$ matrix chosen randomly and independently of $\Sigma_e$ from the Haar measure $\mu_{\mathcal{O}_n}$ on $\mathcal{O}_n$. Let $(Y_1, Y_2, \dots, Y_n)$ be the $n$ columns of $\Sigma_e^{1/2} R$, then $Y_1, Y_2, \dots, Y_n \sim_{\textnormal{i.i.d.}} \mN(0, \Sigma)$.
\end{lemma}

\begin{proof}
Let $X \in \mathbb{R}^{d \times n}$ be a random matrix such that its columns $X_1, X_2, \dots, X_n$ are i.i.d. samples from $\mN(0, \Sigma)$. It follows that the columns of $Y = \Sigma^{-1/2} X$ are $\Sigma^{-1/2} X_i$ for $1 \le i \le n$ which are i.i.d. samples from $\mN(0, I_d)$. This implies that $Y \sim \mN(0, 1)^{\otimes d \times n}$ and therefore $\mL(Y) = \mL(YR_1)$ for any $R_1 \in \mathcal{O}_n$. Hence we also have that $\mL(X) = \mL(\Sigma^{1/2} Y) = \mL(\Sigma^{1/2} YR_1) = \mL(XR_1)$ and the distribution of $X^\top$ is left invariant under $\mathcal{O}_n$. Applying Theorem \ref{thm:invwishart} yields that $X^\top = R_2 \Sigma$ where $R_2 \sim \mu_{\mathcal{F}_{d, n}}$ and $\Sigma$ are independent. Note that
$$\Sigma^2 = \Sigma R_2^\top R_2 \Sigma = X^\top X = \sum_{i = 1}^n X_i X_i^\top \sim \mathcal{W}_n(\Sigma, d)$$
by the definition of the Wishart distribution. Therefore $\Sigma$ is distributed as the positive semidefinite square root of a sample from $\mathcal{W}_n(\Sigma, d)$. This implies that $\mL(X) = \mL(\Sigma R_2^\top) = \mL(\Sigma_e^{1/2} R)$, which completes the proof of the lemma.
\end{proof}

\section{Proof of Theorem \ref{thm:wishartsamp}}
\label{sec:proofmain}

We now prove Theorem \ref{thm:wishartsamp} in a sequence of propositions. The proof will consist of bounding the total variation distances of the outputs of each step of the reduction to target intermediate distributions and applying Lemma \ref{lem:tvacc}. Let $\mathcal{A}$ denote the reduction $\textsc{Clique-to-Wishart}$ and $\mathcal{A}_{s\text{-}t}$ denote the $s$th to $t$th step of the reduction. Throughout this proof, let $p, q$ and $\epsilon$ be fixed constants with $0 < q < p \le 1$ and $\epsilon \in (0, 1)$. Also let $Q = 1 - (1 - p)^{5/6}(1 - q)^{1/6} + \mathbf{1}_{\{p = 1\}} \left( q^{1/6} - 1 \right)$ and $Q' = 1 - \sqrt{(1 - p)(1 - q)} + \mathbf{1}_{\{p = 1\}} \left( \sqrt{q} - 1 \right)$ and observe that $p > Q$ and $p > Q'$. We will assume that $d$ and $n$ are less than some $\text{poly}(N)$, $k \ge k_0$ for some sufficiently large constant $k_0$ and that
$$\theta \le c \cdot \sqrt{\frac{k^{2}}{n \log N}} \quad \text{and} \quad n = \omega(N^3) \quad \text{and} \quad d \ge \left( \frac{p}{Q'} + \epsilon \right) N + k$$
for a sufficiently small constant $c > 0$. Note that $d \ge m$ where $m$ is the smallest multiple of $k$ greater than $\left( \frac{p}{Q'} + \epsilon \right) N$. We now proceed to sequentially analyze the steps of $\mathcal{A}$.

\begin{proposition}[Planting Diagonals and Cloning] \label{prop:cloning}
Let $\mathcal{A}_{\text{1-2}}$ have input $G$ and output $(X_1, X_2, X_3)$. Let $\textnormal{Unif}_{k, m}$ be the uniform distribition on subsets of $[m]$ of size $k$. Then it follows that
\begin{align*}
\TV\left(\mathcal{A}_{\text{1-2}} \left( \mathcal{G}(N, k, p, q) \right), \, \bE_{S \sim \textnormal{Unif}_{k, m}} \, \mathcal{M}\left(m, S, p, Q \right)^{\otimes 3} \right) &\le O\left( \frac{k}{\sqrt{N}} \right) + e^{-C_1 N} \\
\TV\left( \mathcal{A}_{\text{1-2}} \left( \mathcal{G}(N, q) \right), \, \left( \textnormal{Bern}\left(Q \right)^{\otimes m \times m} \right)^{\otimes 3} \right) &\le e^{-C_1 N}
\end{align*}
for some constant $C_1 = C_1(p, q, \epsilon) > 0$.
\end{proposition}

\begin{proof}
First observe that $\frac{1 - p}{1 - Q'} = \left( \frac{1 - p}{1 - Q} \right)^3$ by definition if $p \neq 1$ and that this equality also holds if $p = 1$. Furthermore, we have that by AM-GM
$$p^{2/3} Q'^{1/3} \le \frac{2p + Q'}{3} = 1 - \frac{2(1 - p) + (1 - Q')}{3} \le 1 - (1 - p)^{2/3} (1 - Q')^{1/3}$$
If $p \neq 1$, this rearranges to $\left( \frac{p}{Q} \right)^3 \le \frac{p}{Q'}$. If $p = 1$, then $\left( \frac{p}{Q} \right)^3 = \frac{p}{Q'}$ by definition. Applying Lemma \ref{lem:matrixcloning} now yields that $\mathcal{A}_{2}(\text{Bern}(Q')^{\otimes m \times m}) \sim \left( \text{Bern}(Q)^{\otimes m \times m} \right)^{\otimes 3}$ and that
$$\mL\left( \mathcal{A}_{2} \left( \mathcal{M}\left(m, k, p, Q' \right) \right) \right) = \bE_{S \sim \textnormal{Unif}_{k, m}} \, \mL\left( \mathcal{A}_{2} \left( \mathcal{M}\left(m, S, p, Q' \right) \right) \right) = \bE_{S \sim \textnormal{Unif}_{k, m}} \, \mathcal{M}\left(m, S, p, Q \right)^{\otimes 3}$$
Applying Lemma \ref{lem:submatrix} with the data-processing property in Fact \ref{tvfacts} yields that
\begin{align*}
\TV\left(\mathcal{A}_{\text{1-2}} \left( \mathcal{G}(N, k, p, q) \right), \, \bE_{S \sim \textnormal{Unif}_{k, m}} \, \mathcal{M}\left(m, S, p, Q \right)^{\otimes 3} \right) &\le \TV\left(\mathcal{A}_{1} \left( \mathcal{G}(N, k, p, q) \right), \mathcal{M}\left(m, k, p, Q' \right) \right) \\
&\le 4 \cdot \exp\left( - \frac{Q'^2 \epsilon^2 n^2}{16m} \right) + \sqrt{\frac{k^2(1 - Q')}{2mQ'}} \\
&\quad \quad + \sqrt{\frac{k^2Q'}{2m(1 - Q')}} \\
&\le O\left( \frac{k}{\sqrt{N}} \right) + e^{-C_1 N}
\end{align*}
for some $C_1 = C_1(p, q, \epsilon) > 0$ since $N \le m \le \left( \frac{p}{Q'} + \epsilon \right) N + k$. Similarly, observe that
$$\TV\left(\mathcal{A}_{\text{1-2}} \left( \mG(N, q) \right), \, \left( \text{Bern}(Q)^{\otimes m \times m} \right)^{\otimes 3} \right) \le \TV\left(\mathcal{A}_{1} \left( \mathcal{G}(N, q) \right), \text{Bern}(Q')^{\otimes m \times m} \right) \le e^{-C_1 N}$$
completing the proof of the proposition.
\end{proof}

Now let $v_S = \frac{1}{\sqrt{k}} \cdot \mathbf{1}_S \in \mathbb{R}^m$ for each $S \subseteq [m]$ of size $k$. The analysis of $\mathcal{A}_3$ in the next proposition is where we use the fact that $\theta$ is beneath the optimal barrier of $\theta = \tilde{o}(\sqrt{k^2/n})$.

\begin{proposition}[Generating the Mean Matrix] \label{prop:mean}
Let $\mathcal{A}_3$ has input $X_1$ and output $M$. Then for each fixed $S \subseteq [N]$ of size $k$, we have that
\begin{align*}
\TV\left( \mathcal{A}_3 \left( \mathcal{M}\left(m, S, p,Q\right) \right), \, \theta \sqrt{3} \left( \sqrt{n/2} + g \right) v_S v_S^\top + \mN(0, 1)^{\otimes m \times m} \right) &= O(m^{-1}) \\
\TV\left(\mathcal{A}_3 \left( \textnormal{Bern}\left(Q \right)^{\otimes m \times m} \right), \, \mN(0, 1)^{\otimes m \times m} \right) &= O(m^{-1})
\end{align*}
where $g \sim \mN(0, 1)$ is sampled independently.
\end{proposition}

\begin{proof}
First observe that it almost surely holds $|g| \le 4 \sqrt{\log m}$. Therefore it follows that $\mu_{ij} > 0$ for all $i, j \in [m]$ and that
$$\mu_{ij} = \frac{\theta\sqrt{3}}{k} \left( \sqrt{\frac{n}{2}} + g \right) \le \frac{c \sqrt{3}}{\sqrt{n \log N}} \cdot \left( \sqrt{\frac{n}{2}} + 4 \sqrt{\log m} \right) \le \frac{\min\left\{ \log\left( \frac{p}{Q} \right), \log \left( \frac{1 - Q}{1 - p} \right) \right\}}{2\sqrt{6 \log m + 2\log(p - Q)^{-1}}}$$
as long as $c > 0$ is sufficiently small since $m = \Theta(N) = O(n)$. If $X_1 \sim \mathcal{M}(n, S, p, Q)$ and $M = \mathcal{A}_3(X_1)$, then applying Lemma \ref{lem:gaussianize} for each fixed $g = g'$ with $|g'| \le 4 \sqrt{\log m}$ with the conditioning property of Fact \ref{tvfacts} implies that if $g_1 \sim \mN(0, 1)$, then
\begin{align*}
&\TV\left( \mL(M), \, \theta \sqrt{3} \left( \sqrt{n/2} + g_1 \right) v_S v_S^\top + \mN(0, 1)^{\otimes m \times m} \right) \\
&\quad \quad \quad \le \TV\left( \mL(g), \mL(g_1) \right) + \bE_{g' \sim \mL(g)} \left[ \TV\left( \mL(M | g=g'), \, \frac{\theta \sqrt{3}}{k} \left( \sqrt{n/2} + g' \right) \mathbf{1}_S \mathbf{1}_S^\top + \mN(0, 1)^{\otimes m \times m} \right) \right] \\
&\quad \quad \quad = \bP\left[|g_1| > 4 \sqrt{\log m}\right] + O(m^{-1}) \\
&\quad \quad \quad \le 2 \cdot \frac{1}{4\sqrt{2\pi \log m}} \cdot e^{-\left(4\sqrt{\log m}\right)^2/2} + O(m^{-1}) \\
&\quad \quad \quad = O(m^{-1})
\end{align*}
where the last inequality uses the fact that $1 - \Phi(t) \le \frac{1}{\sqrt{2\pi}} \cdot t^{-1} e^{-t^2/2}$ for all $t \ge 1$. Similarly, if $X_1 \sim \textnormal{Bern}\left(Q \right)^{\otimes m \times m}$, then applying Lemma \ref{lem:gaussianize} for each fixed $g = g'$ yields that
$$\TV\left( \mL(M), \, \mN(0, 1)^{\otimes m \times m} \right) \le \bE_{g' \sim \mL(g)} \left[ \TV\left( \mL(M | g=g'), \, \mN(0, 1)^{\otimes m \times m} \right) \right] = O(m^{-1})$$
which completes the proof of the proposition.
\end{proof}

\begin{proposition}[Generating the Cross-Term Matrices] \label{prop:crossterms}
Let $\mathcal{A}_{4L}$ have input $X_2$ and output $C_L$. Then for each fixed $S \subseteq [N]$ of size $k$, we have that
\begin{align*}
\TV\left( \mathcal{A}_{4L} \left( \mathcal{M}\left(m, S, p, Q\right) \right), \, \sqrt{3\theta} \cdot w v_S^\top + \mN(0, 1)^{\otimes m \times m} \right) &\le O(m^{-1}) + k(4e^{-3})^{m/2k} \\
\TV\left(\mathcal{A}_{4L} \left( \textnormal{Bern}\left(Q \right)^{\otimes m \times m} \right), \, \mN(0, 1)^{\otimes m \times m} \right) &= O(m^{-1})
\end{align*}
where $w \sim \mN(0, I_m)$ is sampled independently.
\end{proposition}

\begin{proof}
First observe that the independence of the two terms in $X \sim \mL(\sqrt{3\theta} \cdot w v_S^\top + \mN(0, 1)^{\otimes m \times m})$ implies the entries of this matrix are zero-mean and jointly Gaussian. Furthermore, the rows are i.i.d. with covariance matrix
$$\bE\left[ XX^\top \right] = I_m + 3 \theta \cdot v_S v_S^\top = I_m + \frac{\tau^2 k^2}{4m} \cdot v_S v_S^\top$$
Now note that since $m = O(N) = o(n^{1/3})$, we have that
$$\tau = \frac{2\sqrt{3m\theta}}{k} \le \sqrt[4]{\frac{c^2 m^2}{k^2 n \log N}} \le \frac{\min\left\{ \log\left( \frac{p}{Q} \right), \log \left( \frac{1 - Q}{1 - p} \right) \right\}}{2\sqrt{6 \log m + 2\log(p - Q)^{-1}}}$$
as long as $k \ge k_0$ is sufficiently large. The proposition is now a direct application of Lemma \ref{lem:randomrotations}.
\end{proof}

Note that in Step 6, $Y = 0$ if $\Sigma_e$ is not positive semidefinite. We remark that total variation loss from this failure mode is essentially baked into the bounds in the following lemma. In particular, any coupling representing the total variations below has $\Sigma_e$ not equal to the corresponding Wishart matrix if $\Sigma_e$ is not positive semidefinite.

\begin{proposition}[Forming the Scaled Empirical Covariance Matrix]
Let the input of $\mathcal{A}_5$ be $(M, C_L, C_R)$ and the output of $\mathcal{A}_5$ be $\Sigma_e$. If it holds that $M, C_L, C_R \sim_{\textnormal{i.i.d.}} \mN(0, 1)^{\otimes m \times m}$, then
$$\TV\left( \mL(\mathcal{A}_5(M, C_L, C_R)), \, \mathcal{W}_n(I_m, m) \right) \to 0 \quad \text{as } k \to \infty$$
Furthermore, suppose that $M, C_L, C_R$ are independent with
\begin{align*}
M &\sim \theta \sqrt{3} \left( \sqrt{n/2} + g \right) v_S v_S^\top + \mN(0, 1)^{\otimes m \times m} \\
C_L &\sim \sqrt{3\theta} \cdot w_1 v_S^\top + \mN(0, 1)^{\otimes m \times m} \\
C_R &\sim \sqrt{3\theta} \cdot v_S w_2^\top + \mN(0, 1)^{\otimes m \times m}
\end{align*}
where $g \sim \mN(0, 1)$, $w_1, w_2 \sim_{\textnormal{i.i.d.}} \mN(0, I_m)$ and $S \subseteq [m]$ with $|S| = k$ is fixed. Then it follows that
$$\max_S \TV\left( \mL(\mathcal{A}_5(M, C_L, C_R)), \, \mathcal{W}_n(I_m + \theta v_S v_S^\top, m) \right) \to 0 \quad \text{as } k \to \infty$$
where the maximum is over all $S \subseteq [N]$ with $|S| = k$.
\end{proposition}

\begin{proof}
If $M, C_L, C_R \sim_{\textnormal{i.i.d.}} \mN(0, 1)^{\otimes m \times m}$, it follows that $X = \frac{1}{\sqrt{3}} \left(  M + C_L + C_R \right)$ is also distributed as $\mN(0, 1)^{\otimes m \times m}$ and therefore that $\frac{1}{\sqrt{2}} \left( X + X^\top \right) \sim \textsc{goe}(m)$. Since $n = \omega(m^3)$, it follows that
$$\TV\left( \mathcal{W}_n(I_m, m), \mL(\Sigma_e) \right) = \TV\left( \mathcal{W}_n(I_m, m), n \cdot I_m + \sqrt{n} \cdot \textsc{goe}(m) \right) \to 0 \quad \text{as } k \to \infty$$
by Theorem \ref{thm:bubeck}, proving the first claim in the statement of the proposition. Now suppose that
\begin{align*}
M &= \theta \sqrt{3} \left( \sqrt{n/2} + g \right) v_S v_S^\top + G_1 \\
C_L &= \sqrt{3\theta} \cdot w_1 v_S^\top + G_2 \\
C_R &= \sqrt{3\theta} \cdot v_S w_2^\top + G_3
\end{align*}
where $G_1, G_2, G_3 \sim_{\text{i.i.d.}} \mN(0, 1)^{\otimes m \times m}$, $g \sim \mN(0, 1)$ and $w_1, w_2 \sim_{\textnormal{i.i.d.}} \mN(0, I_m)$ are independent, and $S \subseteq [m]$ is fixed. Now consider the matrix
$$X = \sqrt{n} \cdot \left[ \frac{1}{\sqrt{3}} \left( G_1 + G_2 + G_3 \right) + \sqrt{\theta} \cdot w_1 v_S^\top + \sqrt{\theta} \cdot v_S w_2^\top + \theta g v_S v_S^\top \right]$$
Note that $\frac{1}{\sqrt{3}} \left( G_1 + G_2 + G_3 \right) \sim \mN(0, 1)^{\otimes m \times m}$ and the four terms above are independent, which implies by Lemma \ref{lem:indrep} that $X \sim \textsc{gw}_m(n, S, \theta)$. Now observe that
\begin{align*}
\Sigma_e &= n \cdot I_m + \sqrt{\frac{n}{6}} \left( M + M^\top + C_L + C_L^\top + C_R + C_R^\top \right) \\
&= n \left( I_m + \theta v_S v_S^\top \right) + \frac{1}{\sqrt{2}} \left( X + X^\top \right)
\end{align*}
Now by Lemma \ref{lem:plantedwishart}, it follows that $\max_S \TV\left( \mathcal{W}_n( I_m + \theta v_S v_S^\top, m), \mL(\Sigma_e) \right) \to 0$ as $k \to \infty$ since $n = \omega(m^3)$ and $\frac{1}{\sqrt{2}} \left( X + X^\top \right) \sim \textsc{sgw}_m(n, S, \theta)$. This completes the proof of the proposition.
\end{proof}

We now deduce the correctness of the inverse Wishart sampling step from Lemma \ref{lem:invwishart} and give a simple analysis of the final padding step.

\begin{proposition}[Inverse Wishart Sampling]
If $\mathcal{A}_6$ has input $\Sigma_e$ and output $Y = (Y_1, Y_2, \dots, Y_n)$, then for each fixed $S \subseteq [m]$ of size $k$
\begin{align*}
&\mathcal{A}_{6} \left( \mathcal{W}_n(I_m, m) \right) \sim \mN(0, I_m)^{\otimes n} \quad \text{and}\\
&\mathcal{A}_{6} \left( \mathcal{W}_n(I_m + \theta v_S v_S^\top, m) \right) \sim \mN\left(0, I_m + \theta v_S v_S^\top\right)^{\otimes n}
\end{align*}
\end{proposition}

\begin{proof}
The first statement in the proposition follows from Lemma \ref{lem:invwishart} applied with $\Sigma = I_m$ and the second follows from Lemma \ref{lem:invwishart} applied with $\Sigma = I_m + \theta v_S v_S^\top$.
\end{proof}

\begin{proposition}[Padding]
Let $\mathcal{A}_7$ have input $Y$ and output the columns $(Z_1, Z_2, \dots, Z_n)$ of $Z$. Then for each fixed $S \subseteq [m]$ of size $k$
\begin{align*}
\mathcal{A}_7\left( \mN\left(0, I_m + \theta v_S v_S^\top\right)^{\otimes n} \right) &\sim \bE_{S \sim \textnormal{Unif}_{k, d}} \, \mN\left(0, I_d + \theta v_Sv_S^\top \right)^{\otimes n} \\
\mathcal{A}_7\left(  \mN(0, I_m)^{\otimes n} \right) &\sim \mN(0, I_d)^{\otimes n}
\end{align*}
\end{proposition}

\begin{proof}
It follows by definition that $Z \sim \mN(0, 1)^{\otimes d \times n}$ if $Y \sim \mN(0, 1)^{\otimes m \times n}$, implying the second identity. Now let $S' \subseteq [d]$ be the set of rows of $Z$ that the rows of $Y$ with indices in $S$ are embedded to. It follows that $\mathcal{A}_7( \mN\left(0, I_m + \theta v_S v_S^\top\right)^{\otimes n} )$ conditioned on $S'$ is distributed as $\mN\left(0, I_d + \theta v_{S'}v_{S'}^\top \right)^{\otimes n}$. Now since $S' \sim \textnormal{Unif}_{k, d}$, the first identity follows. This proves the proposition.
\end{proof}

These propositions provide upper bounds on the $\epsilon_i$ needed to apply Lemma \ref{lem:tvacc} to the steps $\mathcal{A}_i$ and complete the proof of Theorem \ref{thm:wishartsamp}.

\begin{proof}[Proof of Theorem \ref{thm:wishartsamp}]
Let $\mathcal{A}_{4R}$ denote the subprocedure of Step 3 of $\mathcal{A}$ that maps the input $X_3$ to $C_R$. Note that a symmetric statement for $\mathcal{A}_{3R}$ as in Proposition \ref{prop:crossterms} holds. With the notation introduced in the propositions above, the steps of $\mathcal{A}$ map inputs to outputs as follows
\begin{align*}
G &\xrightarrow{\mathcal{A}_1} X \xrightarrow{\mathcal{A}_2}(X_1, X_2, X_3) \xrightarrow{\mathcal{A}_{3}} \left( M, X_2, X_3 \right) \xrightarrow{\mathcal{A}_{4L}} \left( M, C_L, X_3 \right) \xrightarrow{\mathcal{A}_{4R}} \left( M, C_L, C_R \right) \xrightarrow{\mathcal{A}_5} \Sigma_e \\
&\xrightarrow{\mathcal{A}_6} Y \xrightarrow{\mathcal{A}_7} (Z_1, Z_2, \dots, Z_n)
\end{align*}
We first prove the desired result in the case that $H_0$ holds. Consider Lemma \ref{lem:tvacc} applied to the steps $\mathcal{A}_i$ above and the following sequence of distributions
\begin{align*}
\mathcal{P}_0 &= \mG(N, q) \\
\mathcal{P}_{\text{1-2}} &= \left( \text{Bern}(Q)^{\otimes m \times m} \right)^{\otimes 3} \\
\mathcal{P}_3 &= \left( \mN(0, 1)^{\otimes m \times m}, \text{Bern}(Q)^{\otimes m \times m}, \text{Bern}(Q)^{\otimes m \times m} \right) \\
\mathcal{P}_{4L} &= \left( \mN(0, 1)^{\otimes m \times m}, \, \mN(0, 1)^{\otimes m \times m}, \, \text{Bern}(Q)^{\otimes m \times m} \right) \\
\mathcal{P}_{4R} &= \left( \mN(0, 1)^{\otimes m \times m}, \, \mN(0, 1)^{\otimes m \times m}, \, \mN(0, 1)^{\otimes m \times m} \right) \\
\mathcal{P}_{5} &= \mathcal{W}_n(I_m, m) \\
\mathcal{P}_{6} &= \mN(0, I_m)^{\otimes n} \\
\mathcal{P}_{7} &= \mN(0, I_d)^{\otimes n}
\end{align*}
As shown in sequence of propositions above, $\TV( \mathcal{A}_i(\mathcal{P}_{i_-}), \mathcal{P}_i) \le \epsilon_i$ for $\epsilon_i \to 0$ as $k \to \infty$, where $i_-$ denotes the step before step $i$ in the order shown above. Therefore it follows that
$$\TV\left( \mathcal{A}\left( G(N, q) \right), \mN(0, I_d)^{\otimes n} \right) = \TV\left( \mathcal{A}\left( \mathcal{P}_0 \right), \mathcal{P}_6 \right) \to 0 \quad \text{as } k \to \infty$$
which establishes the theorem in the case that $H_0$ holds. We now prove the theorem under $H_1$. Consider first applying the conditioning property in Fact \ref{tvfacts} to mix over the choice $S \sim \textnormal{Unif}_{k, m}$ in the $H_1$ case of each of the propositions proven in this section. Now apply Lemma \ref{lem:tvacc} to the steps $\mathcal{A}_i$ and to the sequence of distributions
\begin{align*}
\mathcal{P}_0 &= \mG(N, k, p, q) \\
\mathcal{P}_{\text{1-2}} &= \bE_{S \sim \textnormal{Unif}_{k, m}} \, \mathcal{M}\left(m, S, p, Q \right)^{\otimes 3} \\
\mathcal{P}_3 &= \bE_{S \sim \textnormal{Unif}_{k, m}} \, \left( \theta\sqrt{3}\left(\sqrt{n/2} + g \right) v_S v_S^\top + \mN(0, 1)^{\otimes m \times m}, \, \mathcal{M}\left(m, S, p, Q \right), \, \mathcal{M}\left(m, S, p, Q \right) \right) \\
\mathcal{P}_{4L} &= \bE_{S \sim \textnormal{Unif}_{k, m}} \, \left( \theta\sqrt{3}\left(\sqrt{n/2} + g \right) v_S v_S^\top + \mN(0, 1)^{\otimes m \times m}, \, \theta\sqrt{3 \theta}\cdot w_1 v_S^\top + \mN(0, 1)^{\otimes m \times m}, \right.\\
&\quad \quad \quad \quad \quad \quad \quad \left. \, \mathcal{M}\left(m, S, p, Q \right) \right) \\
\mathcal{P}_{4R} &= \bE_{S \sim \textnormal{Unif}_{k, m}} \, \left( \theta\sqrt{3}\left(\sqrt{n/2} + g \right) v_S v_S^\top + \mN(0, 1)^{\otimes m \times m}, \, \theta\sqrt{3 \theta}\cdot w_1 v_S^\top + \mN(0, 1)^{\otimes m \times m}, \right. \\
&\quad \quad \quad \quad \quad \quad \quad \left. \, \theta\sqrt{3 \theta}\cdot v_S w_2^\top + \mN(0, 1)^{\otimes m \times m} \right) \\
\mathcal{P}_{5} &= \bE_{S \sim \textnormal{Unif}_{k, m}} \, \mathcal{W}_n(I_m + \theta v_S v_S^\top, m) \\
\mathcal{P}_{6} &= \bE_{S \sim \textnormal{Unif}_{k, m}} \, \mN\left(0, I_m + \theta v_S v_S^\top \right)^{\otimes n} \\
\mathcal{P}_{7} &= \bE_{S \sim \textnormal{Unif}_{k, d}} \,\mN\left(0, I_d + \theta v_S v_S^\top \right)^{\otimes n}
\end{align*}
where $g \sim \mN(0, 1)$ and $w_1, w_2 \sim \mN(0, I_m)$ are sampled independently. By the same reasoning as above, we have that
$$\TV\left( \mathcal{A}\left( G(N, k, p, q) \right), \mN\left(0, I_d + \theta v_S v_S^\top \right)^{\otimes n} \right) \to 0 \quad \text{as } k \to \infty$$
This completes the proof of the theorem.
\end{proof}

\section{Lower Bounds for Sparse PCA via $\chi^2$ Random Rotations}
\label{sec:randomrotations}

In this section, we present subsampling and cloning internal reductions within sparse PCA which, when combined with $\chi^2\textsc{-Random-Rotations}$, yield tight computational lower bounds for $\textsc{cbspca}$ in the regime $k = \tilde{O}(\sqrt{n})$ and $\textsc{ubspca}$ in the regime $k = \omega(\sqrt{n})$, respectively.

\subsection{Subsampling Internal Reduction and $\chi^2$ Random Rotations}

\begin{figure}[t!]
\begin{algbox}
\textbf{Algorithm} \textsc{Subsampling-Random-Rotations}

\vspace{2mm}

\textit{Inputs}: Graph $G \in \mG_N$, edge probabilities $0 < q < p \le 1$ with $q = N^{-O(1)}$, planted subgraph size $K \ll \sqrt{N}$, a parameter $\tau > 0$, target sample count $n$, dimension $d$, approximate sparsity $k$ and $\theta$ where $n$ is the smallest multiple of $K$ larger than $\left( \frac{p}{Q} + \epsilon \right) N + K$ where $Q = 1 - \sqrt{(1 - p)(1 - q)} + \mathbf{1}_{\{p = 1\}} \left( \sqrt{q} - 1 \right)$ and $\epsilon > 0$, $d \ge n$, $k \le K/2$ and $\theta = \tilde{o}\left(\sqrt{k^2/n}\right)$
\begin{enumerate}
\item Form $M \in \mathbb{R}^{n \times n}$ by setting $M \gets \textsc{To-Bernoulli-Submatrix}(G)$ applied with initial dimension $N$, edge probabilities $p$ and $q$ and target dimension $n$
\item Form $X \in \mathbb{R}^{n \times n}$ by setting $X \gets \chi^2\textsc{-Random-Rotation}(M)$ applied with dimension $n$, Bernoulli probabilities $p$ and $Q$ and $\tau = 2 \sqrt{\frac{n\theta}{K \cdot k}}$
\item Form $Y \in \mathbb{R}^{d \times n}$ by embedding each row of $X$ independently with probability $\frac{k}{K}$ into a random empty row of $Y$ and sampling all remaining rows of $Y$ independently from $\mN(0, I_n)$
\item Output the columns $Y_1, Y_2, \dots, Y_n$ of $Y$
\end{enumerate}
\vspace{1mm}
\end{algbox}
\caption{$\chi^2$ Random Rotations reduction with subsampling from $\textsc{pc}$ and $\textsc{pds}$ to $\textsc{cbspca}$.}
\label{fig:subsampling}
\end{figure}

In this section, we observe a subsampling property of sparse PCA and combine this with $\chi^2$ random rotations to show tight lower bounds for the composite formulation $\textsc{cbspca}$ in the regime $k \ll \sqrt{n}$. Our reduction $\textsc{Subsampling-Random-Rotations}$ is described in Figure \ref{fig:subsampling}. Together with the lower bounds shown in \cite{brennan2018reducibility}, this reduction establishes the first tight lower bounds for all sparsities $k$ in a variant of the spiked covariance model of sparse PCA. Omitting Step 3 from $\textsc{Subsampling-Random-Rotations}$ establishes lower bounds for the non-composite formulation $\textsc{ubspca}$ up to the barrier of $\theta = \tilde{o}(k^2/n)$, which is only tight at the single point $k = \tilde{\Theta}(\sqrt{n})$ and $\theta = \tilde{\Theta}(1)$, matching the reduction $\textsc{Random-Rotations}$ in \cite{brennan2018reducibility}. 

We remark that this reduction crucially relies on the planted clique conjecture up to $k = o(\sqrt{n})$. Unlike $\textsc{Clique-to-Wishart}$, given the planted clique conjecture for $k = \tilde{\Theta}(n^{\alpha})$ where $\alpha < 1/2$, $\textsc{Subsampling-Random-Rotations}$ shows computational lower bounds that are not tight at any sparsity $k$. If $\alpha = 1/3$, then it fails to show computational lower bounds exceeding the information theoretic lower bound of $k = \tilde{\Theta}(\sqrt{k/n})$. We now establish the subsampling property of sparse PCA. Let $\mathcal{A}$ denote $\textsc{Subsampling-Random-Rotations}$ and $\mathcal{A}_3$ be Step 3 with input $X$ and output $Y$.

\begin{lemma}[Subsampling Property of Sparse PCA] \label{lem:subsampling}
Let $\theta_1 > 0$ and $d \ge \frac{2kn}{K}$ and $k \le K/2$. Then there is a distribution $\pi$ over pairs $(\theta, v)$ supported on $A_{\theta_2} \times B_{k + \frac{1}{2}\gamma \sqrt{k}}$ where $A_{\theta_2}$, $B_{k + \frac{1}{2}\gamma \sqrt{k}}$ and $\gamma$ are as in the definition of $\textsc{cbspca}$ in Section \ref{sec:problems} and $\theta_2 = \frac{\theta_1 k}{K}$ satisfying
\begin{align*}
&\TV\left( \mathcal{A}_3 \left( \mN\left(0, I_n + \theta_1 v_S v_S^\top \right)^{\otimes n} \right), \,  \bE_{(\theta, v) \sim \pi} \, \mN\left(0, I_d + \theta vv^\top \right)^{\otimes n} \right) \le 2 \exp\left( - \frac{\gamma^2}{16} \right) \\
&\mathcal{A}_3 \left( \mN(0, I_n)^{\otimes n} \right) \sim \mN(0, I_d)^{\otimes n} 
\end{align*}
for each fixed subset $S \subseteq [n]$ with $|S| = K$.
\end{lemma}

\begin{proof}
Note that the second identity $\mathcal{A}_3 \left( \mN(0, I_n)^{\otimes n} \right) \sim \mN(0, I_d)^{\otimes n}$ follows from the definition of $\mathcal{A}_3$. Now fix a subset $S \subseteq [n]$ with $|S| = K$. Since the entries of the two matrices shown below are jointly Gaussian and have the same covariance matrix, we have that
$$\mN\left(0, I_n + \theta_1 v_S v_S^\top \right)^{\otimes n} = \mL\left( \sqrt{\frac{\theta_1}{K}} \cdot \mathbf{1}_S w^\top + \mN(0, 1)^{\otimes n \times n} \right)$$
where $w \sim \mN(0, I_n)$ is sampled independently. Let $T \subseteq [d]$ with $|T| \le K$ be the indices of the rows of $Y$ that the rows indexed by $S$ in $X$ are successfully embedded to. If $T' \subseteq [d]$ is a fixed subset with $|T'| \le K$, then observe that
\begin{align*}
\mL\left( \mathcal{A}_3 \left( \mN\left(0, I_n + \theta_1 v_S v_S^\top \right)^{\otimes n} \right) \Big| T = T' \right) &\sim \mL\left( \sqrt{\frac{\theta_1}{K}} \cdot \mathbf{1}_{T'} w^\top + \mN(0, 1)^{\otimes d \times n} \right) \\
&= \mL\left( \sqrt{\frac{\theta_1|T'|}{K}} \cdot v_{T'} w^\top + \mN(0, 1)^{\otimes d \times n} \right) \\
&= \mN\left(0, I_d + \frac{\theta_1 |T'|}{K} \cdot v_{T'} v_{T'}^\top \right)^{\otimes n}
\end{align*}
Now note that $T$ is a subset of $[d]$ with size $|T| \sim \text{Bin}(K, k/K)$ chosen uniformly at random. Now let $\pi'$ be the distribution over pairs $(\theta, v)$ where $\theta = \frac{\theta_1 |T|}{K}$ and $v = v_{T}$ induced by the distribution of $T$. Now observe that if $k - \frac{1}{2} \gamma \sqrt{k} \le |T| \le k + \frac{1}{2} \gamma \sqrt{k}$ then
$$k + \frac{1}{2}\gamma \sqrt{k} - \gamma \sqrt{k + \gamma \sqrt{k}} \le |T| \le k + \frac{1}{2} \gamma \sqrt{k}$$
and thus $v = v_T \in B_{k + \gamma \sqrt{k}}$. Furthermore, it holds that
$$\left| \theta - \theta_2 \right| = \frac{\theta_1}{K} \cdot \left| |T| - k \right| \le \frac{\theta_2 \gamma}{2\sqrt{k}}$$
and therefore $\theta \in A_{\theta_2}$. Now let $\pi$ be the distribution on $(\theta, v)$ conditioned on the event $k - \frac{1}{2} \gamma \sqrt{k} \le |T| \le k + \frac{1}{2} \gamma \sqrt{k}$. Now observe that
\begin{align*}
\bP\left[ |T| \ge k + t \right] &\le \exp\left( - K \cdot D\left( \frac{k + t}{K} \Big\| \frac{k}{K} \right) \right) \\
\bP\left[ |T| \le k - t \right] &\le \exp\left( - K \cdot D\left( \frac{k - t}{K} \Big\| \frac{k}{K} \right) \right)
\end{align*}
where $D(\cdot \| \cdot)$ denotes the binary relative entropy function. These two inequalities are well-known and can be derived by standard Chernoff bounds. Note that if $t/k \to 0$ and $k \le K/2$, then Taylor expanding yields that
\begin{align*}
K \cdot D\left( \frac{k + t}{K} \Big\| \frac{k}{K} \right) &= (k + t) \log \left( \frac{k + t}{k} \right) + (K - k - t) \log \left( \frac{K - k - t}{K - k} \right) \\
&= (k + t) \left( \frac{t}{k} - \frac{t^2}{2k^2} + O\left( \frac{t^3}{k^3} \right) \right) \\
&\quad \quad - (K - k - t) \left( \frac{t}{K - k} + \frac{t^2}{2(K - k)^2} + O\left( \frac{t^3}{(K - k)^3} \right) \right) \\
&= \frac{t^2}{k} + \frac{t^2}{K - k} - \frac{t^2(k + t)}{2k^2} - \frac{t^2(K - k - t)}{2(K - k)^2} + O\left( \frac{t^3}{k^2} \right) \\
&= \frac{t^2}{2k} + \frac{t^2}{2(K - k)} + O\left( \frac{t^3}{k^2} \right)
\end{align*}
By a nearly identical computation, we have that
$$K \cdot D\left( \frac{k - t}{K} \Big\| \frac{k}{K} \right) = \frac{t^2}{2k} + \frac{t^2}{2(K - k)} + O\left( \frac{t^3}{k^2} \right)$$
By the data-processing and conditioning property in Fact \ref{tvfacts}, we have
\begin{align*}
&\TV\left( \mathcal{A}_3 \left( \mN\left(0, I_n + \theta_1 v_S v_S^\top \right)^{\otimes n} \right), \,  \bE_{(\theta, v) \sim \pi} \, \mN\left(0, I_d + \theta vv^\top \right)^{\otimes n} \right) \\
&\quad \quad \le \TV\left( \bE_{(\theta, v) \sim \pi'} \, \mN\left(0, I_d + \theta vv^\top \right)^{\otimes n}, \bE_{(\theta, v) \sim \pi} \, \mN\left(0, I_d + \theta vv^\top \right)^{\otimes n} \right) \\
&\quad \quad = \TV\left( \pi, \pi' \right) \\
&\quad \quad = \bP \left[ \left| |T| - k \right| > \frac{1}{2} \gamma \sqrt{k} \right] \\
&\quad \quad \le 2 \exp\left( - \frac{t^2}{2k} - \frac{t^2}{2(K - k)} + O\left( \frac{t^3}{k^2} \right) \right) \\
&\quad \quad \le 2 \exp\left( - \frac{\gamma^2}{8} + o(1) \right) \le 2 \exp\left( - \frac{\gamma^2}{16} \right)
\end{align*}
where $t = \frac{1}{2} \gamma \sqrt{k}$, which completes the proof of the lemma.
\end{proof}

We now combine this lemma with the guarantees from Section \ref{sec:mappingsubmatrix} for $\textsc{To-Bernoulli-Submatrix}$ and $\chi^2\textsc{-Random-Rotations}$ to derive total variation guarantees for $\mathcal{A}$.

\begin{theorem}[Subsampling Random Rotations] \label{thm:subsamplingrandomrot}
Let $0 < q < p \le 1$ and $\epsilon > 0$ be constants. Let $N, K, n, k, d$ and $\theta > 0$ be parameters such that $n$ is the smallest multiple of $K$ large than $\left( \frac{p}{Q} + \epsilon \right) N + K$ where $Q = 1 - \sqrt{(1 - p)(1 - q)} + \mathbf{1}_{\{p = 1\}} \left( \sqrt{q} - 1 \right)$,
$$d \ge n \quad \textnormal{and} \quad k \le K/2 \quad \textnormal{and} \quad \theta \le c \cdot \frac{K \cdot k}{n \sqrt{\log N}}$$
for a sufficiently small constant $c > 0$. Then $\mathcal{A} = \textsc{Subsampling-Random-Rotations}$ runs in $\textnormal{poly}(N)$ time and
\begin{align*}
\TV\left( \mathcal{A} \left( \mG(N, k, p, q) \right), \,  \bE_{(\theta', v) \sim \pi} \, \mN\left(0, I_d + \theta' vv^\top \right)^{\otimes n} \right) &\le O\left( \frac{K}{\sqrt{N}} \right) + e^{-C_1 N} + O(n^{-1}) \\
&\quad \quad + K(4e^{-3})^{n/2K} + 2 \exp\left( - \frac{\gamma^2}{16} \right) \\
\TV\left( \mathcal{A} \left( \mG(N, q) \right), \, \mN\left(0, I_d \right)^{\otimes n} \right) &\le e^{-C_1 N} + O(n^{-1})
\end{align*}
for some distribution $\pi$ over pairs $(\theta, v)$ supported on $A_{\theta} \times B_{k + \frac{1}{2}\gamma \sqrt{k}}$ and some fixed constant $C_1 = C_1(p, q, \epsilon) > 0$.
\end{theorem}

\begin{proof}
Consider applying Lemma \ref{lem:tvacc} to steps $\mathcal{A}_1, \mathcal{A}_2$ and $\mathcal{A}_3$ and the sequence of distributions
$$\mP_0 = \mG(N, q), \quad \mP_1 = \text{Bern}(Q)^{\otimes n \times n}, \quad \mP_2 = \mN\left(0, I_n \right)^{\otimes n}, \quad \mP_3 = \mN\left(0, I_d \right)^{\otimes n}$$
By Lemmas \ref{lem:submatrix}, \ref{lem:randomrotations} and \ref{lem:subsampling}, we have that $\TV\left( \mathcal{A}_i(\mP_{i - 1}), \mP_i \right) \le \epsilon_i$ for each $1 \le i \le 3$ where
$$\epsilon_1 = 4 \exp\left( - \frac{Q^2 \epsilon^2 N^2}{16n} \right) \le e^{-C_1 N}, \quad \epsilon_2 = O(n^{-1}), \quad \epsilon_3 = 0$$
for some constant $C_1 = C_1(p, q, \epsilon) > 0$. Applying Lemma \ref{lem:tvacc} now yields the second bound in the theorem statement on $\TV\left( \mathcal{A} \left( \mG(N, q) \right), \, \mN\left(0, I_d \right)^{\otimes n} \right)$. Now consider applying Lemma \ref{lem:tvacc} to steps $\mathcal{A}_1, \mathcal{A}_2$ and $\mathcal{A}_3$ and the sequence of distributions
\begin{align*}
\mP_0 &= \mG(N, K, p, q) \\
\mP_1 &= \mathcal{M}(n, K, p, Q) \\
\mP_2 &= \bE_{S \sim \text{Unif}_{K, n}} \, \mN\left( 0, I_n + \frac{\tau^2 K^2}{4n} \cdot v_S v_S^\top \right)^{\otimes n} = \bE_{S \sim \text{Unif}_{K, n}} \, \mN\left( 0, I_n + \frac{\theta K}{k} \cdot v_S v_S^\top \right)^{\otimes n} \\
\mP_3 &= \bE_{(\theta', v) \sim \pi} \, \mN\left( 0, I_d + \theta' vv^\top \right)^{\otimes n}
\end{align*}
where $\pi$ is the distribution in Lemma \ref{lem:randomrotations} supported on $A_\theta \times B_{k + \frac{1}{2} \gamma \sqrt{k}}$. Now by Lemmas \ref{lem:submatrix} and \ref{lem:randomrotations} and the conditioning property of Fact \ref{tvfacts} applied to the total variation bound in Lemma \ref{lem:subsampling}, we have that $\TV\left( \mathcal{A}_i(\mP_{i - 1}), \mP_i \right) \le \epsilon_i$ for each $1 \le i \le 3$ where
\begin{align*}
\epsilon_1 &= 4 \exp\left( - \frac{Q^2 \epsilon^2 N^2}{16n} \right) + \sqrt{\frac{K^2(1 - Q)}{2nQ}} + \sqrt{\frac{K^2 Q}{2n (1 - Q)}} \le O\left( \frac{K}{\sqrt{N}} \right) + e^{-C_1 N} \\
\epsilon_2 &= O(n^{-1}) + K(4e^{-3})^{n/2K} \\
\epsilon_3 &\le 2 \exp\left( - \frac{\gamma^2}{16} \right)
\end{align*}
Applying Lemma \ref{lem:tvacc} now yields the bound on $\TV\left( \mathcal{A} \left( \mG(N, k, p, q) \right), \,  \bE_{(\theta, v) \sim \pi} \, \mN\left(0, I_d + \theta vv^\top \right)^{\otimes n} \right)$, completing the proof of the theorem.
\end{proof}

Now set $K = \left\lceil \kappa^{-1} \sqrt{N} \right\rceil$ for an arbitrarily slow-growing function $\kappa \to \infty$ in Theorem \ref{thm:subsamplingrandomrot}. This proves hardness for $\textsc{cbspca}(n, k', d, \theta)$ for true sparsity parameter $k' = \left\lceil k + \frac{1}{2} \gamma \sqrt{k} \right\rceil$ as long as $d \ge n$, $k' \ll \sqrt{n}$ and
$$\theta = O\left( \frac{K \cdot k}{n \sqrt{\log N}} \right) = O\left( \kappa^{-1} \cdot \sqrt{\frac{k'^2}{n \log N}} \right)$$
for this choice of $K$. Note that $k$ is the approximate sparsity parameter described in Figure \ref{fig:subsampling}. More formally, we arrive at the following tight hardness for $\textsc{cbspca}$ given the strongest form of the $\textsc{pc}$ conjecture up to the boundary of $K = o(\sqrt{n})$.

\begin{corollary} \label{thm:chirandomrot}
Let $(n_t, k_t, d_t, \theta_t)$ be a sequence of parameters with $d_t \ge n_t$, $k_t \ll \sqrt{n_t}$ and
$$\theta_t \ll \sqrt{\frac{k_t^2}{n_t \log n_t}}$$
For any sequence of randomized polynomial time tests $\mathcal{A}_t$, the asymptotic Type I$+$II error of $\mathcal{A}_t$ on the detection problem $\textsc{cbspca}(n_t, k_t, d_t, \theta_t)$ is at least $1$ assuming the $\textsc{pds}$ conjecture for all $k = o(\sqrt{n})$ and some pair of constant densities $0 < q < p \le 1$.
\end{corollary}

\subsection{Internal Cloning Reductions within Sparse PCA}
\label{sec:internal}

\begin{figure}[t!]
\begin{algbox}
\textbf{Algorithm} \textsc{Sparsity-Cloning}

\vspace{2mm}

\textit{Inputs}: Sparse PCA samples $X_1, X_2, \dots, X_n \in \mathbb{R}^d$, number of iterations $\ell$
\begin{enumerate}
\item Begin by setting $X^0_i = X_i$ for each $1 \le i \le n$
\item For $j = 1, 2, \dots, \ell$ do:
\begin{enumerate}
\item[(1)] Sample $G_1, G_2, \dots, G_n \sim_{\text{i.i.d.}} \mN(0, I_{2^{j - 1} d})$
\item[(2)] For each $1 \le i \le n$, form $X_i^j \in \mathbb{R}^{2^j d}$ by concatenating the two vectors
$$\left(X_i^j\right)_{1:2^{j - 1} d} = \frac{1}{\sqrt{2}} \left( X^{j - 1}_i + G_i \right) \quad \text{and} \quad \left(X_i^j\right)_{2^{j - 1} d + 1:2^j d} = \frac{1}{\sqrt{2}} \left( X^{j - 1}_i - G_i \right)$$
\end{enumerate}
\item Sample a permutation $\pi$ of $[2^\ell d]$ u.a.r. and output $(X^{\ell}_1)^\pi, (X^{\ell}_2)^\pi, \dots, (X^{\ell}_n)^\pi$
\end{enumerate}
\vspace{1mm}
\end{algbox}
\caption{Internal cloning reduction for increasing the sparsity $k$ of an instance of sparse PCA.}
\label{fig:internalcloning}
\end{figure}

In this section, we give a simple internal reduction showing that a lower bound for an instance of the spiked covariance model with at a particular triple $(n, k, d, \theta)$ implies lower bounds at $(n, tk, td, \theta)$ for $t(n) \ge 1$ with $t(n) = \text{poly}(n)$. This reduction is shown in Figure \ref{fig:internalcloning}. The correctness of this reduction essentially reduces to the following simple lemma.

\begin{lemma}[One Step of Sparsity Cloning]
Let $\mathcal{B}$ be the first iteration of Step 2 with $j = 1$ of the algorithm $\textsc{Sparsity-Cloning}$. Then it holds for all subsets $S \subseteq [d]$ that
\begin{align*}
\mathcal{B}\left( \mN\left(0, I_d + \theta v_S v_S^\top \right)^{\otimes n} \right) &\sim \mN\left(0, I_{2d} + \theta v_T v_T^\top \right)^{\otimes n} \\
\mathcal{B}\left( \mN\left(0, I_d \right)^{\otimes n} \right) &\sim \mN\left(0, I_{2d} \right)^{\otimes n}
\end{align*}
where $T = S \cup \{ d + i : i \in S\} \subseteq [2d]$.
\end{lemma}

\begin{proof}
First suppose that $X_1^0, X_2^0, \dots, X_n^0 \sim \mN\left(0, I_d + \theta v_S v_S^\top \right)^{\otimes n}$. Since the $X_i^0$ and $G_i$ are independent, it follows that the $X_i^1$ are independent by construction. The same decomposition as in Lemma \ref{lem:subsampling} yields that the $X_i^0$ can be written as
$$X_i^0 = \sqrt{\frac{\theta}{k}} \cdot g_i \mathbf{1}_S + H_i$$
where the $g_i \sim \mN(0, 1)$ and $H_i \sim \mN(0, I_d)$ are independent for $1 \le i \le n$. Now note that the entries of the $2d$-dimensional vector $\frac{1}{\sqrt{2}} \cdot (H_i + G_i, H_i - G_i)^\top$ are zero-mean and jointly Gaussian since the entries of $H_i$ and $G_i$ are i.i.d. $\mN(0, 1)$. Observe that $(H_i + G_i, H_i - G_i)^\top = (H_i, H_i)^\top + (G_i,- G_i)^\top$ and these two terms are independent and zero-mean. The covariance matrix of $\frac{1}{\sqrt{2}}(H_i + G_i, H_i - G_i)^\top$ is the sum of the covariance matrices of these two terms and therefore equal to
$$\frac{1}{2} \cdot \bE\left[ \left( \begin{matrix} H_i \\ H_i \end{matrix} \right) \left( \begin{matrix} H_i \\ H_i \end{matrix} \right)^\top \right] + \frac{1}{2} \cdot \bE\left[ \left( \begin{matrix} G_i \\ -G_i \end{matrix} \right) \left( \begin{matrix} G_i \\ -G_i \end{matrix} \right)^\top \right] = \frac{1}{2} \cdot \left[ \begin{matrix} I_d & I_d \\ I_d & I_d \end{matrix} \right] + \frac{1}{2} \cdot \left[ \begin{matrix} I_d & -I_d \\ -I_d & I_d \end{matrix} \right] = I_{2d}$$
Therefore we have that
$$X_i^1 = \sqrt{\frac{\theta}{2k}} \cdot \left( \begin{matrix} \mathbf{1}_S \\ \mathbf{1}_S \end{matrix} \right) + \frac{1}{\sqrt{2}} \cdot \left( \begin{matrix} H_i + G_i \\ H_i - G_i \end{matrix} \right) \sim \mL\left( \sqrt{\theta} \cdot v_T + \mN(0, I_{2d})\right)$$
Thus it follows that $(X_1^1, X_2^1, \dots, X_n^1) \sim \mN\left(0, I_{2d} + \theta v_T v_T^\top \right)^{\otimes n}$ since the $X_i^1$ are independent. Applying this with $\theta = 0$ yields that $\mathcal{B}\left( \mN\left(0, I_d \right)^{\otimes n} \right) \sim \mN\left(0, I_{2d} \right)^{\otimes n}$. This completes the proof of the lemma.
\end{proof}

Applying this lemma to each iteration in Step 2 of $\textsc{Sparsity-Cloning}$, applying induction and accounting for the permutation in Step 3, we arrive at the following guarantees for the reduction.

\begin{lemma}[Sparsity Cloning]
If $\mathcal{A} = \textsc{Sparsity-Cloning}$, then $\mathcal{A}$ runs in $\textnormal{poly}(2^\ell d, n)$ time and it holds for all subsets $S \subseteq [d]$ that
\begin{align*}
\mathcal{A}\left( \mN\left(0, I_d + \theta v_S v_S^\top \right)^{\otimes n} \right) &\sim \bE_{T \sim \textnormal{Unif}_{2^\ell k, 2^\ell d}} \, \mN\left(0, I_{2^\ell d} + \theta v_T v_T^\top \right)^{\otimes n} \\
\mathcal{A}\left( \mN\left(0, I_d \right)^{\otimes n} \right) &\sim \mN\left(0, I_{2^\ell d} \right)^{\otimes n}
\end{align*}
\end{lemma}

Note that this lemma implies the following conditional hardness result for sparse PCA, effectively increasing $k$ and $d$ while preserving hardness.

\begin{corollary} \label{thm:internalcloning}
Let $(n_t, k_t, d_t, \theta_t)$ be a sequence of parameters and $2^{\ell_t} = \textnormal{poly}(n_t)$. Then if the asymptotic Type I$+$II error is at least $1$ for any sequence $\mathcal{A}_t$ of randomized polynomial time tests for $\textsc{ubspca}(n_t, k_t, d_t, \theta_t)$, then the same is true for $\textsc{ubspca}(n_t, 2^{\ell_t} k_t, 2^{\ell_t} d_t, \theta_t)$.
\end{corollary}

This implies that the tight lower bounds derived from the planted clique conjecture up to $K = o(N^\alpha)$ in our main theorem can also be extended to lower bounds with $k = \omega(n^{\alpha/3})$. Although not tight to the threshold of $\theta = \tilde{\Theta}(\sqrt{k^2/n})$, these lower bonds still show nontrivial statistical-computational gaps.

We remark that there is another natural cloning map that re-derives the tight computational lower bounds for $\textsc{fcspca}$ at $\theta = \tilde{o}(1)$ when $d = O(n)$ and $k = \omega(\sqrt{n})$ from \cite{brennan2018reducibility}. Consider the cloning map with $X_i^j \in \mathbb{R}^d$ for all $j$ that iteratively computes
$$X_i^j= \frac{1}{\sqrt{2}} \left( X^{j - 1}_i + \mathcal{R} X^{j - 1}_i \right) \quad \text{and} \quad X_i^j = \frac{1}{\sqrt{2}} \left( X^{j - 1}_i - \mathcal{R} X^{j - 1}_i \right)$$
where $\mathcal{R}$ is the operator reversing the indices of a vector. If $d$ is even and $k \ll d$, a similar analysis to that of $\textsc{Reflection-Cloning}$ in \cite{brennan2018reducibility} shows that this procedure maps an $\textsc{fcspca}(n, k, d, \theta)$ instance to $\textsc{fcspca}(n, 2^\ell k, d, \theta)$ approximately in total variation. Combining this with the hardness of $\textsc{fcspca}$ when $k = \tilde{\omega}(\sqrt{n})$ and $\theta = \tilde{\Theta}(1)$ from Corollary \ref{thm:chirandomrot}, yields tight lower bounds for all of $k = \omega(\sqrt{n})$, matching Theorem 8.5 from \cite{brennan2018reducibility}.

\section{Discussion, Extensions and Open Problems}
\label{sec:conclusion}

The computational lower bounds and techniques we introduce have several simple extensions and implications. The following is a sketch of these extensions.

\paragraph{Non-constant $p$ and $q$.} While our two main theorems derive computational lower bounds for sparse PCA from the $\textsc{pds}$ conjecture with $p$ and $q$ constant, both of the reductions $\textsc{Clique-to-Wishart}$ and $\textsc{Subsampling-Random-Rotations}$ have guarantees when $p$ and $q$ vary with $n$. In particular, if $q = \Theta(1)$, $1 - q = \Theta(1)$ and $p - q = \Theta(n^{-\beta})$ then we obtain weaker lower bounds, degrading as $\beta$ grows, that are still nontrivial.

\paragraph{Estimation and Recovery Lower Bounds.} Throughout this paper, we have focused on showing lower bounds for sparse PCA formulated as a detection problem. As noted in Section 10 of \cite{brennan2018reducibility}, reductions showing detection lower bounds for sparse PCA can often be modified to yield lower bounds for partially recovering the support of the planted vector. This can be achieved by cloning the instance at some early stage in the reduction and then running the reduction on each of the two clones while coupling coordinate permutations. The end result is two instances of sparse PCA coupled to have the same latent planted vector but otherwise independent. Consider a given partial recovery blackbox on one instance to output a set $S$ of size $k$ and then looking at the principal minor at $S$ of the empirical covariance matrix of the other instance. Thresholding the largest eigenvalue of this minor can be shown to solve the detection problem as long as the sparse PCA instance is statistically possible. In the case of $\textsc{Clique-to-Wishart}$, this additional cloning step can be implemented in Step 2 by making six copies of $X$ instead of three. The output of the reduction will then be approximately be two i.i.d. copies from $\mN(0, I_d + \theta v_S v_S^\top)^{\otimes n}$ where $S$ is random but shared between the two copies. Since the planted vectors in $\textsc{ubspca}$ are uniform over their support, an estimation blackbox yields a partial recovery algorithm by adding a step thresholding the entries of the output vector.

\paragraph{Reducing from Symmetric Biclustering and Sparse Spiked Wigner.} Our reductions can also be adapted to show hardness for sparse PCA instead based on lower bounds for symmetric Gaussian biclustering and sparse spiked Wigner, which are formulated below.
\begin{center}
\begin{tabular}{l l}
\vspace{2mm}
Symmetric Biclustering & Sparse Spiked Wigner \\
$H_0 : \mN(0, 1)^{\otimes N \times N}$ & $H_0 :  \mN(0, 1)^{\otimes N \times N}$ \\
$H_1 : \mu \cdot \mathbf{1}_S \mathbf{1}_S^\top + \mN(0, 1)^{\otimes N \times N}$ & $H_1 : \mu \cdot vv^\top + \mN(0, 1)^{\otimes N \times N}$
\end{tabular}
\end{center}
where $S \sim \text{Unif}_{k, N}$ and the entries of $v$ are chosen uniformly at random from the set of all $k$-sparse $N$ dimensional vectors with all nonzero entries equal to $\pm 1/\sqrt{k}$. Note that thresholding the entries of an instance of symmetric biclustering at zero produces an instance of planted dense subgraph, which then can be mapped to sparse PCA through our reductions. For sparse spiked Wigner, more modifications to $\textsc{Clique-to-Wishart}$ are needed. First produce three clones of the instance of sparse spiked Wigner with a similar trick to Step 2.2 of $\textsc{Sparsity-Cloning}$ and then apply the $\textsc{Random-Rotations}$ reduction from \cite{brennan2018reducibility} to the second and third clones. If $X$ is the first clone, then consider sampling $g \sim \mN(0, 1)$ and setting $M \gets \tau X + \sqrt{1 - \tau^2} \cdot \mN(0, 1)^{\otimes N \times N}$ where $\tau = \frac{1}{2} + (2n)^{-1/2} g$. With a similar analysis to that presented here, these steps can be shown to produce $(M, C_L, C_R)$ approximately in total variation for appropriate choices of parameters in subroutines. Proceeding with the remainder of $\textsc{Clique-to-Wishart}$ produces an instance of sparse PCA.

\paragraph{Quasipolynomial Time Algorithms.} While there is a an $N^{O(\log N)}$ quasipolynomial time algorithm for solving planted clique, no such algorithm is known for planted dense subgraph if $p - q = \Theta(n^{-\epsilon})$ and $q, 1 - q = \Theta(1)$ for some $\epsilon > 0$ where the planted dense subgraph has size $K = o(N^{\alpha})$. As outlined above, our reductions extend to this regime of non-constant $p$ and $q$, and can be verified to show lower bounds for $\pr{ubspca}(n, k, d, \theta)$ where
$$\theta = \tilde{o} \left( \sqrt{\frac{k^2}{n^{1 + 2\epsilon/3}}} \right)$$
Note that our reductions from planted dense subgraph with $N$ vertices map to sparse PCA in randomized polynomial time with parameters that can be chosen such that $d = O(N)$ and $n = \tilde{O}(N^3)$. Therefore if there is an algorithm solving sparse PCA in time $T(n, d)$ then there is an algorithm solving planted dense subgraph in $T(O(N), \tilde{O}(N^3)) + \text{poly}(N)$ time. Taking $\epsilon$ to be small shows that our reduction also excludes quasipolynomial time algorithms for sparse PCA for these $(n, k, d, \theta)$ assuming there are none for planted dense subgraph in the described parameter regime.

\subsection{Open Problems}

There are a number of questions that remain unresolved about computational lower bounds for sparse PCA based on the planted clique conjecture. The following is an overview of some problems that are open after this work.

\begin{itemize}
\item \textbf{Reducing the $k = o(n^{\alpha/3})$ Condition:} To use the $\textsc{goe}$-Wishart convergence theorem from \cite{bubeck2016testing} in our reduction from the $\textsc{pc}$ conjecture up to $K = o(N^{\alpha})$, it was necessary that the target sparse PCA parameters satisfied $k = o(n^{\alpha/3})$. Is there a polynomial time reduction improving this condition be improved to $k = o(n^{\alpha})$, so that there is no degradation in the relative level of sparsity? We suspect that answering this question affirmatively would require different techniques.
\item \textbf{Equivalence of Planted Dense Subgraph and Sparse PCA:} We show a reduction from planted dense subgraph to sparse PCA. Is there an explicit polynomial time reduction from planted dense subgraph to sparse PCA? Is there a reduction to planted dense subgraph and then back to sparse PCA preserving the relative sparsity $k$ i.e. from sparse PCA with $k = \Theta(n^\beta)$ to planted dense subgraph to sparse PCA with parameters $k' = \Theta(n'^\beta)$? We suspect that answering this question would require reducing the $k = o(n^{\alpha/3})$ condition.
\item \textbf{Hardness of Simple vs. Simple Hypothesis Testing:} Our reductions show tight hardness for the simple vs. simple hypothesis testing formulation $\textsc{ubspca}$ for $k = o(n^{\alpha/3})$ assuming the planted clique conjecture up to $K = o(N^\alpha)$ and for the simple vs. composite formulation $\textsc{cbspca}$ for $k = \tilde{O}(\sqrt{n})$. This leaves open the tight hardness of $\textsc{ubspca}$ for $n^{1/3} \ll k \ll n^{1/2}$. Is there a reduction from planted clique showing tight lower bounds for $\textsc{ubspca}$ in this parameter regime?
\item \textbf{Other Internal Reductions within Sparse PCA:} In Section \ref{sec:internal}, we give an internal reduction increasing the sparsity $k$ and dimension $d$ of a sparse PCA instance while preserving $n$ and $\theta$. Applying a similar trick to instead double the number of samples $n$ fails to produce independent samples because of nonzero correlations between entries on the planted part. Is there a reduction increasing $n$ or decreasing the relative sparsity $k$ while appropriately scaling $\theta$ to preserve the tight hardness of the sparse PCA instance? Another question is whether there is an internal reduction to increase $k$ and decrease $\theta$ while preserving tight hardness. This would provide a method of reducing the $k = o(n^{\alpha/3})$ condition.
\item \textbf{Universality of Computational Lower Bounds:} Our reduction heavily makes use of the isotropy of the Gaussian noise distribution in the spiked covariance model. Are there reductions showing tight computational lower bounds for formulations of sparse PCA in other noise models similar to the universal lower bounds for submatrix detection in \cite{brennan2019universality}?
\end{itemize}

\section*{Acknowledgements}

We thank Wasim Huleihel, Philippe Rigollet, Yury Polyanskiy, Elchanan Mossel, Frederic Koehler, Vishesh Jain, Yash Deshpande, Enric Boix, Austin Stromme, Dheeraj Nagaraj and Govind Ramnarayan for inspiring discussions on related topics. We also thank Sam Hopkins for discussions on the implications of our reductions for quasipolynomial algorithms solving sparse PCA. This work was supported in part by the grants ONR N00014-17-1-2147 and NSF CCF-1565516.

\bibliography{GB_BIB.bib}
\bibliographystyle{alpha}

\appendix

\section{Relationships Among Variants of Sparse PCA}
\label{sec:relationshipsspca}

In this section, we discuss the relationships among the variants of sparse PCA introduced in Section \ref{sec:problems} as well as the following variants.
\begin{itemize}
\item \textit{Composite unbiased sparse PCA} is denoted as $\pr{cspca}(n, k, d, \theta)$ where $A_\theta = \{ \theta \}$ and
$$B_k = \left\{ v \in \mathbb{S}^{d-1}  : k - \tau \sqrt{k} \le \| v \|_0 \le k \text{ and } |v_i| \ge \frac{1}{\sqrt{k}} \text{ for } i \in \text{supp}(v) \right\}$$
for some fixed parameter $\tau$ satisfying that $\tau \to \infty$ as $n \to \infty$.
\item \textit{Uniform unbiased sparse PCA} is denoted as $\pr{uspca}(n, k, d, \theta)$ where $A_\theta = \{ \theta \}$ and $B_k$ is the set of all $k$-sparse unit vectors in $\mathbb{S}^{d - 1}$ with nonzero coordinates equal to $\pm 1/\sqrt{k}$. 
\end{itemize}
As in the case of $\pr{ubspca}$, this formulation of $\pr{uspca}$ is equivalent to the simple vs. simple hypothesis testing problem where $v$ is drawn uniformly at random from this set $B_k$ under $H_1$. Furthermore, $\pr{uspca}$ can be exactly obtained from $\pr{ubspca}$ by multiplying the $i$th coordinate of each sample by a Rademacher random variable $x_i$. This implies that any reductions to $\pr{ubspca}$ also apply to $\pr{uspca}$. However, as observed in \cite{brennan2018reducibility}, $\pr{ubspca}$ appears to be a strictly easier problem than $\pr{uspca}$ when $k = \Omega(\sqrt{n})$. This is because the sum of the entries of the covariance matrix of $\pr{ubspca}$ allows detection at lower $\theta$ than the best known polynomial time algorithms for $\pr{uspca}$ in this regime.

Since the sets of distributions $H_1$ for $\pr{cbspca}$ and $\pr{cspca}$ include that of $\pr{ubspca}$, computational lower bounds for $\pr{ubspca}$ are the strongest. Similarly, lower bounds for $\pr{fcspca}$ are the weakest. As previously mentioned, reductions to simple vs. simple hypothesis testing formulations are technically more challenging than for composite problems. As previously mentioned, our main reduction $\textsc{Clique-to-Wishart}$ applies to $\pr{ubspca}$ while $\textsc{Subsampling-Random-Rotations}$ requires the composite formulation in $\pr{cbspca}$. Previous reductions to sparse PCA have mostly been to composite hypothesis testing formulations. In \cite{berthet2013complexity} and \cite{wang2016statistical}, lower bounds are shown for a composite vs. composite formulation of sparse PCA in a sub-Gaussian noise model rather than the spiked covariance model. This formulation and these reductions are discussed in more detail in Appendix \ref{sec:appendixprevred}. In \cite{gao2017sparse}, planted clique lower bounds were shown for sparse PCA in a formulation similar to $\pr{fcspca}$ at the suboptimal barrier of $\theta = \tilde{\Theta}(k^2/n)$. In \cite{brennan2018reducibility}, optimal lower bounds when $d = \tilde{\Theta}(n)$ and $k = \Omega(\sqrt{n})$ were shown for $\pr{cspca}$ and suboptimal lower bounds at $\theta = \tilde{\Theta}(k^2/n)$ were shown for $\pr{ubspca}$. Most algorithmic papers for sparse PCA consider either $\pr{ubspca}$ or $\pr{uspca}$ for the detection task or a variant of $\pr{cspca}$ where $B_k = \{ v \in \mathbb{S}^{d - 1} : \| v \|_0 = k \text{ and } |v_i| = \Omega(1/\sqrt{k}) \text{ for } i \in \text{supp}(v) \}$ for the recovery task.

\section{Previous Reductions to Sparse PCA}
\label{sec:appendixprevred}

Part of the rationale behind our reductions comes from a careful understanding of pre-existing reductions to sparse PCA and the dependences they induce in $n, k$ and $\theta$. In this section, we give a short overview of some key ideas in the reductions of \cite{berthet2013complexity}, \cite{brennan2018reducibility}, \cite{wang2016statistical} and \cite{gao2017sparse}, focusing on \cite{berthet2013complexity} and \cite{brennan2018reducibility}.

\subsection{Berthet-Rigollet (2013)}

The reduction in \cite{berthet2013complexity} begins with an instance $G$ of the planted clique problem $\pr{pc}(N, N^{1/2 - \epsilon}, 1/2)$ and maps it to a sub-Gaussian variant of sparse PCA with $\theta = \Theta(k/n^{1/2 + \epsilon})$. The reduction proceeds as follows:
\begin{enumerate}
\item Randomly permute the vertices of $G$ and let $B$ be the lower left $kN^{1/2 + \epsilon} \times n$ submatrix of the matrix $2 \cdot A(G) - \mathbf{1}_{N \times N}$ where $A(G)$ is the adjacency matrix of $G$ where $k \le \frac{1}{2} N^{1/2 - \epsilon}$.
\item For each column $B_i$ of $B$, multiply $B_i$ by a sign chosen u.a.r. from $\{-1, 1\}$ and then append $d - kN^{1/2 + \epsilon}$ u.a.r. samples from $\{-1, 1\}$ to the end of $B_i$.
\item Output $B_1, B_2, \dots, B_n \in \{-1, 1\}^d$ as the target sparse PCA samples.
\end{enumerate}
In \cite{berthet2013complexity}, $n$ is chosen such that $n = \Theta(N)$ and $n/N$ is a sufficiently small constant, $k$ satisfies $k \le \frac{1}{2}N^{1/2 - \epsilon}$ and $d$ satisfies $d \ge kN^{1/2 + \epsilon}$.

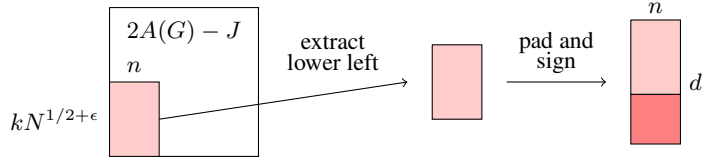
\begin{figure*}[t!]
\centering
\begin{tikzpicture}[scale=0.33]
\tikzstyle{every node}=[font=\footnotesize]
\draw (0, 3) -- (0, -3) -- (6, -3) -- (6, 3) -- (0, 3);
\draw [fill=red!20] (0, -3) -- (2, -3) -- (2, 0) -- (0, 0) -- (0, -3);
\node at (3, 2) {$2 A(G) - J$};
\node at (-2.2, -1.5) {$kN^{1/2 + \epsilon}$};
\node at (1, 0.6) {$n$};
\draw[->] (2, -1.5) -- (12, 0);
\node at (9, 1.6) {extract};
\node at (9, 0.8) {lower left};
\draw [fill=red!20] (13, 1.5) -- (15, 1.5) -- (15, -1.5) -- (13, -1.5) -- (13, 1.5);
\draw[->] (16, 0) -- (20, 0);
\node at (18, 1.6) {pad and};
\node at (18, 0.8) {sign};
\draw [fill=red!20] (21, 2.5) -- (23, 2.5) -- (23, -0.5) -- (21, -0.5) -- (21, 2.5);
\draw [fill=red!50] (21, -0.5) -- (23, -0.5) -- (23, -2.5) -- (21, -2.5) -- (21, -0.5);
\node at (22, 3) {$n$};
\node at (23.6, 0) {$d$};
\end{tikzpicture}
\caption{Visual depiction of the reduction in \cite{berthet2013complexity} from $\pr{pc}(N, N^{1/2 - \epsilon}, 1/2)$ to sub-Gaussian sparse PCA with parameters $n = \Theta(N)$, $k  = O(N^{1/2 - \epsilon})$, $d \ge kN^{1/2 + \epsilon}$ and $\theta = \Theta(k/n^{1/2 + \epsilon})$.}
\label{fig:berthetrigollet}
\end{figure*}

Since the output $B_i$ take values in $\{-1, 1\}^d$, this reduction does not show planted clique-hardness for sparse PCA in the spiked covariance model. Instead, \cite{berthet2013complexity} uses this reduction to show lower bounds for a sub-Gaussian variant of sparse PCA. Given a sequence of samples $X = (X_1, X_2, \dots, X_n) \in \mathbb{R}^d$, let $\hat{\Sigma}(X) = \frac{1}{n} \sum_{i = 1}^n X_i X_i^\top$ be the empirical covariance matrix of the samples. If $X$ is distributed according to $H_0$ in the spiked covariance model, then the spectral norm of $\hat{\Sigma}(X)$ satisfies $\| \hat{\Sigma}(X) - I_d \| = O(1/\sqrt{n})$ with high probability. If $X$ is distributed according to $H_1$, then it follows that $v^\top \hat{\Sigma}(X) v = 1 + \theta + O(1/n + \sqrt{k \theta /n})$ with high probability where $v$ denotes the planted $k$-sparse principal component. Extending to all distributions obeying these concentration inequalities yields the following composite vs. composite hypothesis testing formulation:
\begin{align*}
H_0 : X \sim \mP_0^{\otimes n} \quad \text{ where } \mP_0 &\text{ s.t. } \bP_{X \sim \mP_0^{\otimes n}} \left[ \| \hat{\Sigma}(X) - I_d \| = \Omega\left(\frac{1}{\sqrt{n}} \right) \right] \le \epsilon \\
H_1 : X \sim \mP_1^{\otimes n} \quad \text{ where } \mP_1 &\text{ s.t. } \bP_{X \sim \mP_1^{\otimes n}} \left[ v^\top \hat{\Sigma}(X) v \le  1 + \theta - \Omega\left(\frac{1}{n} + \sqrt{\frac{k \theta}{n}}\right) \right] \le \delta \\
&\text{ for some } v \in \mathcal{B}_0(k)
\end{align*}
In \cite{berthet2013complexity}, the constants in the $\Omega(\cdot)$ above and $\delta$ are specified. Under this formulation of sparse PCA, the Berthet-Rigollet reduction shows a lower bound against algorithms that can solve sparse PCA over all distributions satisfying the sub-Gaussian inequalities above.

To show correctness for this reduction, it suffices to show that $H_i$ of $\pr{pc}$ is mapped to a distribution with small total variation distance from a distribution $\mP_i^{\otimes n}$ in $H_i$ above for each $i = 1, 2$. The marginals of the columns $B_i$ can be verified to satisfy the sub-Gaussian concentration inequalities above. Therefore it suffices to show that the columns $B_1, B_2, \dots, B_n$ are close in total variation to independent. This amounts to showing that the number of the $n$ columns drawn from the original $N$ that correspond to clique vertices is close in total variation to $\text{Bin}(N, K/N)$. This follows from finite de Finetti's theorem and the fact that $n/N$ is small. We now argue that the instance of sub-Gaussian sparse PCA resulting from this reduction has $\theta = \Theta(k/n^{1/2 + \epsilon})$. Under $H_1$, let $v \in \mathbb{R}^d$ be the unit vector with entries $\pm 1/\sqrt{k}$ on the intersection of the rows of $B$ with the planted clique, with signs chosen to match those used to form $B$ in Step 2. Then $\theta$ is given by
$$1 + \theta = \bE\left[v^\top \hat{\Sigma}(B) v\right] = \bE\left[ \frac{1}{n} \sum_{i = 1}^n (B_i^\top v)^2 \right] = \bE\left[(B_1^\top v)^2\right] = 1 - \frac{1}{N^{1/2 + \epsilon}} + \frac{k}{N^{1/2 + \epsilon}}$$
Since vertex $1$ is in the clique with probability $N^{-1/2-\epsilon}$, in which case $(B_1^\top v)^2 = k$. Otherwise, $(B_1^\top v)^2$ has conditional expectation $1$. The reduction in \cite{wang2016statistical} adapts the analysis of this reduction to show lower bounds for the sparse PCA estimation task.

\subsection{Gao-Ma-Zhou (2017) and Brennan-Bresler-Huleihel (2018)}

\begin{figure*}[t!]
\centering
\begin{tikzpicture}[scale=0.33]
\tikzstyle{every node}=[font=\footnotesize]
\draw (0, 3) -- (0, -3) -- (6, -3) -- (0, 3);
\draw [fill=red!30] (0, 0) -- (0, 3) -- (3, 0) -- (0, 0);
\node at (1.25, 0.48) {clique};
\node at (2.5, -2.3) {lower adj.};
\draw[->] (6, 0) -- (12, 0);
\node at (9, 1.8) {gaussianize and};
\node at (9, 0.8) {plant diagonals};
\draw (13, 3) -- (13, -3) -- (19, -3) -- (13, 3);
\draw [fill=red!30] (13, 0) -- (13, 3.5) -- (19.5, -3) -- (19, -3) -- (16, 0) -- (13, 0);
\node at (16.8, 2) {$\mN(\mu, 1)$};
\node at (15.5, -2.3) {$\mN(0, 1)$};
\draw[->] (20, 0) -- (26, 0);
\node at (23, 1.8) {asymmetrize};
\node at (23, 0.8) {and permute};
\draw (27, 3.4) -- (27, -3.4) -- (33.8, -3.4) -- (33.8, 3.4) -- (27, 3.4);
\draw [fill=red!30] (27, 3.4) -- (27, 0) -- (30.4, 0) -- (30.4, 3.4) -- (27, 3.4);
\node at (28.75, 1.7) {$\mN(\mu, 1)$};
\node at (30.4, -1.7) {$\mN(0, 1)$};
\draw[->] (0, -10) -- (6, -10);
\node at (3, -8.2) {randomly};
\node at (3, -9.2) {rotate};
\draw (7, -6.6) -- (7, -13.4) -- (13.8, -13.4) -- (13.8, -6.6) -- (7, -6.6);
\draw [fill=red!10] (7, -6.6) -- (7, -10) -- (13.8, -10) -- (13.8, -6.6) -- (7, -6.6);
\draw[->] (9.2, -5) arc (112:428:3cm and 0.4cm);
\node at (10.4, -8.3) {$\mu \sqrt{k} \cdot \mathbf{1}_S u^\top$};
\node at (10.4, -11.7) {$+ \mN(0, 1)^{\otimes N \times N}$};
\draw[->] (15, -10) -- (21, -10);
\node at (18, -9.2) {subsample};
\draw (22, -6.6) -- (22, -13.4) -- (27, -13.4) -- (27, -6.6) -- (22, -6.6);
\draw [fill=red!10] (22, -6.6) -- (22, -10) -- (27, -10) -- (27, -6.6) -- (22, -6.6);
\end{tikzpicture}
\caption{Visual depiction of the reduction in \cite{brennan2018reducibility} as applied under $H_1$.}
\label{fig:randomrotations}
\end{figure*}
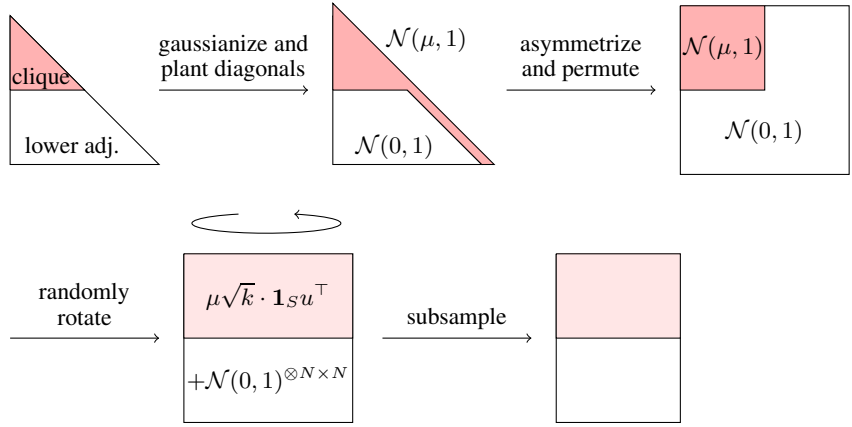

The reduction in \cite{gao2017sparse} is similar to that in \cite{berthet2013complexity}, but instead Gaussianizes the adjacency matrix of planted clique and then maps the coordinates corresponding to the planted clique to mixtures of two gaussians. This ultimately produces samples from the spiked covariance model, showing the suboptimal lower bound at $\theta = \tilde{\Theta}(k^2/n)$.

We focus in this section on providing an overview of the reduction in \cite{brennan2018reducibility} and its differences from \cite{berthet2013complexity}. This reduction obtains the same guarantees of the reduction in \cite{gao2017sparse} in the stronger simple vs. simple hypothesis testing formulation of the spiked covariance model. In order to extend the reduction in \cite{berthet2013complexity} to the spiked covariance model, a natural approach would be to devise a randomized map $f$ such that $f(B_i)$ is approximately distributed as a sample from the spiked covariance model under each of $H_0$ and $H_1$. However, observe that
$$\TV\left( \mL_{H_0}(B_1), \mL_{H_1}(B_1) \right) \le \frac{K}{N} \asymp \frac{k}{n}$$
for each $i \in [n]$. This is because under $H_1$, $B_i$ is distributed as a mixture with probability $1 - K/N$ of $\mL_{H_0}(B_1)$ as described above. However, such a map would satisfy that
\begin{align*}
\TV\left( \mL_{H_0}(f(B_1)), \mL_{H_1}(f(B_1)) \right) &\asymp \TV\left( \mN(0, I_d), \mN(0, I_d + \theta vv^\top) \right) \\
&\asymp \min \left\{ 1, \| \theta vv^\top \|_F \right\} = \theta
\end{align*}
by Theorem 1.1 in \cite{devroye2018total}. Therefore if such a map were to exist, the data processing inequality would require that $\theta = O(k/n)$, which is always below the statistic limit. Thus there are two criteria that any reduction from planted clique to the spiked covariance model would need to fulfill: (1) the reduction would need to somehow mix the columns of the adjacency matrix rather than process them in isolation; and (2) the reduction would need to output a jointly Gaussian distribution.

The key insight of the reduction to sparse PCA in \cite{brennan2018reducibility} and our $\chi^2\textsc{-Random-Rotations}$ reduction is that randomly rotating the \textit{rows} -- and thus across the different samples -- of a Gaussianized planted clique adjacency matrix fulfills these two criteria simultaneously. The reduction in \cite{brennan2018reducibility} maps an instance $G$ of the planted clique problem $\pr{pc}(N, K, 1/2)$ as follows:
\begin{enumerate}
\item Plant diagonal entries of $1$ in $A(G)$ and compute $M$ by applying a rejection kernel mapping $1 \to \mN(\mu, 1)$ and $\text{Bern}(1/2) \to \mN(0, 1)$ where $\mu = \Theta(1/\sqrt{\log n})$ symmetrically to $A(G)$.
\item Asymmetrize $M$ by setting $M_{ij} \gets \frac{1}{\sqrt{2}} \left(M_{ij} + G_{ij} \right)$ and $M_{ji} \gets \frac{1}{\sqrt{2}} \left(M_{ij} - G_{ij} \right)$ for each $i \le j$ where $G_{ij} \sim_{\text{i.i.d.}} \mN(0, 1)$ and randomly permute the rows of $M$.
\item Generate $R$ from the Haar measure on the orthogonal group $\mO_N$ and update $M \gets MR$.
\item Output the first $n$ columns of $M$ as the target sparse PCA samples where $n/N$ is small.
\end{enumerate}
It is shown through a $\chi^2$ divergence computation in \cite{brennan2018reducibility} that after Steps 1 and 2, the planted diagonals are hidden by the random rotations. The remainder of the correctness of this reduction follows an argument similar to the analysis of $\chi^2\textsc{-Random-Rotations}$, but instead of planting $\chi^2$-mean entries during Gaussianization, relies on the following finite de Finetti-style theorem for coordinates of random unit vectors.

\begin{theorem}[Diaconis and Freedman \cite{diaconis1987dozen}]
Suppose that $(v_1, v_2, \dots, v_n)$ is uniformly distributed according to the Haar measure on $\mathbb{S}^{n-1}$. Then for each $1 \le m \le n - 4$,
$$\TV\left( \mL\left( v_1, v_2, \dots, v_m \right), \mN(0, n^{-1})^{\otimes m} \right) \le \frac{2(m+3)}{n - m - 3}$$
\end{theorem}

\subsection{Reducing Directly to Samples vs. the Empirical Covariance Matrix}

We now provide some rough intuition as to why one could expect mapping to the empirical covariance matrix to produce an optimal dependence among the parameters $n, k$ and $\theta$ and thus yield strong hardness for sparse PCA based on weak forms of the planted clique conjecture. In both of the reductions outlined above, the $\ell_2$ norm of the planted signal remains of the same order throughout the reduction. The $\ell_2$ norm squared of the clique component of the adjacency matrix of a planted clique instance is $\Theta(k^2)$. The samples from the spiked covariance model can be expressed under $H_1$ as $X_i = \sqrt{\theta} \cdot g v^\top + \mN(0, I_d)$ where $g \sim \mN(0, 1)$. The expected total $\ell_2$ norm squared of the components $\sqrt{\theta} \cdot g v^\top$ is $n \theta$ since $v$ is a unit vector and $\bE[g^2] = 1$. Thus a reduction that roughly preserves the total $\ell_2$ norm squared of the planted signal and that maps directly to samples would satisfy that $n \theta = O(k^2)$ and would show a lower bound with the suboptimal dependence $\theta = O(k^2/n)$. Note that each entry of the empirical covariance matrix of the $X_i$ has variance $\Theta(n)$. Scaling so that this is $\Theta(1)$ yields that the expectation of the empirical covariance matrix is $\sqrt{n} \cdot ( I_d + \theta vv^\top)$. Thus the $\ell_2$ norm squared of the planted part $\theta \sqrt{n} \cdot vv^\top$ is $\Theta(\theta^2 n)$. Equating this with the $\ell_2$ norm squared of the planted clique now yields the optimal dependence $\theta = O(\sqrt{k^2/n})$.

\end{document}